\documentclass[12pt]{article}
\usepackage[latin1]{inputenc}
\usepackage{amsfonts,amssymb,amsmath, epsfig}

\usepackage{mathptmx} 
\usepackage{bm}

\usepackage{amsfonts} 
\usepackage{amsmath} 
\usepackage{enumerate} 
\usepackage{dsfont} 

\usepackage{cite} 

\usepackage{graphicx}
\usepackage{subfigure} 
\usepackage{psfrag} 
\usepackage[usenames]{color}

\def\paragraph#1{{\bf #1\ }}

\newtheorem{lemma}{Lemma}[section]  

\newtheorem{theorem}[lemma]{Theorem}

\newtheorem{proposition}[lemma]{Proposition}

\newtheorem{Formal Statement}{Formal Statement} 

\newenvironment{proof}[1][Proof]{\begin{trivlist}
\item[\hskip \labelsep {\bfseries #1}]}{\end{trivlist}}

\def\Box{\leavevmode\vbox{\hrule
     \hbox{\vrule\kern4pt\vbox{\kern4pt}%
           \vrule}\hrule}}
\def\blackbox{\leavevmode\vrule height 5pt width 4pt depth 0pt\relax}
\def\qed{\null\hfill {$\blackbox$}\bigskip}

\def\R{\mathbb{R}}

\def\S{\mathbb{S}}

\def\eps{\varepsilon}
\title{Congestion in a macroscopic model of self-driven particles modeling gregariousness}

\author{Pierre Degond$^{1,2}$,
        Laurent Navoret$^{1,2}$,       
        Richard Bon$^{3,4}$,						
        David Sanchez$^{1,2}$}

\date{}

\begin{document}

\maketitle

\vspace{0.5 cm}

\begin{center}
              1. Universit\'e de Toulouse; UPS, INSA, UT1, UTM; Institut de Math\'ematiques de Toulouse; F-31062 Toulouse, France \\
              2. CNRS; Institut de Mathmatiques de Toulouse UMR 5219; F-31062 Toulouse, France.\\
              email: pierre.degond, laurent.navoret, david.sanchez @math.univ-toulouse.fr     
\end{center}
           
\begin{center}             
3. Universit\'e de Toulouse; UPS, INSA, UT1, UTM; Centre de Recherches sur la Cognition Animale; F-31062 Toulouse, France \\
              4. CNRS; Centre de Recherches sur la Cognition Animale UMR 5169; F-31062 Toulouse, France.\\
              email: rbon@cict.fr     
\end{center}

\begin{abstract}
We analyze a macroscopic model with a maximal density constraint which describes short range repulsion in biological systems. This system aims at modeling finite-size particles which cannot overlap and repel each other when they are too close. The parts of the fluid where the maximal density is reached behave like incompressible fluids while lower density regions are compressible. This paper investigates the transition between the compressible and incompressible regions. To capture this transition, we study a one-dimensional Riemann problem and introduce a perturbation problem which regularizes the compressible-incompressible transition. Specific difficulties related to the non-conservativity of the problem are discussed.
\end{abstract}


\section{Introduction}
\label{intro}
We consider a macroscopic model of self-driven particles which describes the dynamics of a large number of social interactive agents. More specifically, we are interested in modeling short range repulsion effects due to the fact that finite-size agents (e.g. sheep in a herd) cannot overlap (non-overlapping or steric constraints). To this aim, we derive a hyperbolic problem with a density constraint as a limit of an unconstrained system with a repulsive force which turns on suddenly when the density becomes close to the maximal one. The limit model requires transmission conditions at the transition between an unclustered region (where the maximal density is not reached) and a clustered region. In unclustered regions, the fluid is compressible while it becomes incompressible in the clustered ones. Therefore, this paper aims at providing a description of this transition between a compressible and an incompressible fluid. Unfortunately, the formal perturbative approach which we implement does not directly provide information about these transmission conditions. In order to retrieve this information, we rigorously analyze special solutions of the perturbation problem: the Riemann problem. These solutions are explicitely known  and allow us to carry out the limit rigorously and to recover the required transmission conditions. We postulate that these conditions, which are rigorously proven only for Riemann problem solutions, do extend to all solutions. However, being non-rigorous for general solutions, these conditions are stated as formal conditions in the "formal statement 1" below, which constitutes the main result of the present paper. Still, the rigorous analysis of Riemann problem solutions is quite technical and the proofs of many statements are deferred to appendices. 

The modeling of biological systems undergoing flocking or herding dynamics has been the subject of a vast literature. A first class of models relies on the alignement interaction between neighbouring self-propelled particles. The simplest of these models is an individual-based (or microscopic) model proposed by Vicsek \cite{95_Vicsek_PhasTrans2d,2004_GregChate_CohesivMotion}. A macroscopic version of the Vicsek model is derived in \cite{2008_ContinuumLimit_DM} and a collisional Vicsek model is proposed in \cite{2006_BoltzSelfPropel_BertinGregoire}. A variant of the Vicsek model has been proposed by Cucker and Smale \cite{2007_MathEmergence_CuckerSmale,2007_EmergenceFlocks_CuckerSmale} (see also \cite{2008_Flocking_HaTadmor,2009_CuckerSmale_CarilloToscani} for recent mathematical results). By incorporating long-range attractive and short-range repulsive forces to the Vicsek model, one obtains the three zones model of Aoki \cite{1982_SimuStud_Aoki,1987_FlocksHerdsSchools_Reynolds,2002_Couzin_CollectMem}, originally devised to describe fish schools. Models with repulsive-attractive interaction only (without alignement interaction) have been studied in \cite{2003_IntPot_MogKesh,2006_OrsognaBertozzi,2007_BertozziOrsogna}. Such models have been used for pedestrian interactions \cite{1995_SocialForce_HelbingMolnar,2009_ExpStudy_Moussaid}. Other kinds of macroscopic models of drift-diffusion type have been analyzed in \cite{2001_TrafficReview_Helbing,1999_NonLocalModel_MogKesh,06_Nonlocal_TopazBertozzi,2008_CoarsInterfac_Slepcev} and different hyperbolic models are compared in \cite{2008_ModCrowd_Bellomo}. For biological reviews, we can refer to \cite{2003_Vertebrates_CouzinKrause,2009_InformProcess_MoussaidHelbing}.

As outlined above, we focus here on the congestion constraint: animals or individuals cannot overlap (steric constraint). As a consequence this congestion constraint leads to the existence of a maximal density $\rho^{\ast}$, which cannot be exceeded inside the flock. This problem has been analyzed before and schematically two methods have been proposed. A first one consists in modeling repulsion through forces or diffusion terms \cite{2003_IntPot_MogKesh,2006_OrsognaBertozzi,2007_BertozziOrsogna,1999_NonLocalModel_MogKesh,06_Nonlocal_TopazBertozzi}. However, in this approach, the individuals are point particles  and their finite size is not explicitely described. So the maximal density constraint is not explicitely taken into account. To explicitely take this maximal density constraint inco account, in \cite{2007_contact_MauryVenel,2008_Crowd_MauryVenel},  the authors have developped an alternative approach: the particles are first evolved freely over one time step and then projected towards the "closest" admissible non-overlapping configuration. This leads to non-local interaction between the particles which contradicts the local character of the interactions in most biological systems. By contrast, we developped a third route inspired by multi-phase flows \cite{2PhaseFlow_Bouchut_al} and traffic jam modeling \cite{2008_Traffic_DegondRascle,2008_Trafficflow_Berthelin_D}. The repulsive force is  modeled by a nonlinear pressure law $p(\rho)$ which becomes singular as the density approaches the maximal density $\rho^{\ast}$. Additionnally a small parameter $\eps$  allows to describe the fact that the regularized pressure is very small of order $\eps$ as long as the density $\rho$ is smaller than $\rho^{\ast}$ and turns on suddenly to a finite or even large value when $\rho$ becomes close to $\rho^{\ast}$. In the limit $\eps \rightarrow 0$ of this model, two distinct phases  appear: a pressureless compressible phase which describes free motion in unclustered regions and an incompressible phase which describes the motion inside the clusters. The major difficulty is to find the transmission conditions between the compressible and incompressible phases. 

The present paper is a multi-dimensional extension of the methodology presented in \cite{2PhaseFlow_Bouchut_al,2008_Traffic_DegondRascle,2008_Trafficflow_Berthelin_D} for multi-phase flows or traffic. However, an additional difficulty arises due to the non-conservative character of the original hyperbolic  model. Indeed, momentum is not a conserved quantity because the particles in the underlying particle system  are self-propelled particles which have constant (in-time) and uniform (in-space) velocities. Therefore, the model which is at the starting point of this paper is a non-conservative hyperbolic system which as such presents an ambiguity in the definition of weak solutions. We will show that this ambiguity  can be partly removed for one-dimensional Riemann problem solutions. We believe that the strategy developped in this paper to analyze congestion effects can apply to other systems such as bacteria populations \cite{ClusterRod_PDB}, economic systems like supply chains \cite{InteractMachine_ADR} or physical systems like granular materials \cite{1995_SizeSegregegation_PoschelHermann,2005_GranularCollective_Barrat_al}.  

The organization of this article is as follows. In section 2, we present the perturbation model and its limit. We also provide the connection conditions between the compressible and incompressible phases of the limit model, which are the main result of the paper. A remark on collision of clusters is also formulated. With these informations, we show that the available information is sufficient to provide a well-defined dynamics at least in the case of a single cluster. Section 3 is devoted to the study of the one-dimensional Riemann problem for the perturbed problem and the limits of its solutions as $\eps \rightarrow 0$.  As stated above, this analysis provides a strong support for (but not a proof of) the postulated transmission conditions at the compressible-incompressible interface  which are provided in section 2. Appendix A  provides a formal derivation of the initial model from an individual based model with long-range attractive and short-range repulsive interactions, which describes the aggregation of gregarious animals like sheep. Appendices B to E provide proofs of technical lemmas and theorems needed in the analysis of the Riemann problem.

\section{Model and goals}
\label{incomp_comp_trition}

\subsection{The model and its rescaled form}

Our starting point is the following model, written in dimensionless form:
\begin{eqnarray}
&&\partial_{t}\rho + \nabla_{\vec{x}}\cdot(\rho\Omega) = 0,\label{Eq:rho}\\
&&\partial_{t}\Omega + (\Omega\cdot\nabla_{\vec{x}})\Omega + (\mbox{Id} - \Omega\otimes\Omega)\nabla p(\rho) = 0.\label{Eq:Omega}
\end{eqnarray}
where $\rho = \rho(\vec{x},t)$ is the particle density and $\Omega = \Omega(\vec{x},t)$ is the particle velocity. The problem is posed on the 2-dimensional plane $\vec{x} \in \R^{2}$ and $t > 0$ is the time. The velocity $\Omega(\vec{x},t) \in \R^{2}$ is supposed to satisfy the normalization constraint
\begin{equation}
|\Omega(\vec{x},t)| = 1,\ \forall \vec{x} \in \R^{2},\ \forall t > 0.
\end{equation} 
Therefore, $\Omega(\vec{x},t) \in \S^{1}$, the unit sphere, at any point in space-time. The function $p(\rho)$ is an increasing function such that $p(\rho) \sim \rho^{\gamma}$ when $\gamma \ll 1$ and $p(\rho) \rightarrow + \infty$ when $\rho \rightarrow \rho^{\ast}$ where $\rho^{\ast}$ is the so-called congestion density. In this paper, we will consider
\begin{equation}
p(\rho) = \frac{1}{\left(\frac{1}{\rho} - \frac{1}{\rho^{\ast}}\right)^{\gamma}},
\end{equation}
for simplicity but any other function with similar behaviour would lead to similar results. The operators $\nabla_{\vec{x}}\cdot$ and $(\Omega\cdot\nabla_{\vec{x}})$ are defined, for a vector field $\vec{A} = (A_{1},A_{2})(\vec{x})$, by
\begin{eqnarray}
&&\nabla_{\vec{x}}\cdot \vec{A} = \partial_{x_{1}}A_{1} + \partial_{x_{2}}A_{2},\\
&&(\Omega\cdot\nabla_{\vec{x}})\vec{A} = ((\Omega_{1}\partial_{x_{1}} + \Omega_{2}\partial_{x_{1}})A_{1},(\Omega_{1}\partial_{x_{1}} + \Omega_{2}\partial_{x_{1}})A_{2})^{T}, 
\end{eqnarray} 
where $T$ denotes the transpose operator. Finally, $(\mbox{Id} - \Omega\otimes\Omega)$ is the projection matrix onto the line spanned by $\Omega^{\perp}$, where $\Omega^{\perp}$ is the vector $\Omega$ rotated by the angle $\pi/2$. Alternatively, we have, for a vector $\vec{A}$:
\begin{equation}
(\mbox{Id} - \Omega\otimes\Omega)\vec{A} = \vec{A} - (\Omega\cdot \vec{A})\vec{A},
\end{equation}
where $(\Omega\cdot \vec{A})$ is the dot product $\Omega\cdot \vec{A} = \Omega_{1}A_{1} + \Omega_{2}A_{2}$.

We show in appendix $A$ that this model well describes the behaviour of a system of particles subjected to long-range attraction and short-range repulsion in the spirit of a model proposed by Aoki \cite{1982_SimuStud_Aoki} or Couzin et al \cite{2002_Couzin_CollectMem} for modelling gregariousness and swarming. More precisely, in appendix $A$, we derive this system from such a particle system through successive changes of scales via mean-field and hydrodynamic theories. In the form (\ref{Eq:Omega}), we have dropped the force term describing long-range attraction. Indeed, this force term would add the quantity $(\mbox{Id} - \Omega\otimes\Omega)\vec{\vec{\xi}}_{a}$ at the right-hand side of (\ref{Eq:Omega}), with 
\begin{equation*}
\vec{\xi_{a}}(\vec{x},t) = \frac{\int K_{a}(|\vec{y} -\vec{x}|)(\vec{y}-\vec{x})\rho(\vec{y},t)d\vec{y}}{\int K_{a}(|\vec{y} -\vec{x}|)\rho(\vec{y},t)d\vec{y}},
\end{equation*}   
where $K_{a}$ is a bounded positive kernel. This terms does not add any differential operator and all the subsequent analysis will stay unaltered by adding this term. 

Our main concern is the study of the congestion effects brought by the singularity of $p(\rho)$ near the congestion density $\rho^{\ast}$. Indeed, a herd of animals can be viewed, at large scales, as a domain of space where the density $\rho$ is close to the saturation density $\rho^{\ast}$. Therefore, the geometrical domain occupied by the herd at time $t$ can be identified to a set $H_t = \{ x \in {\mathbb R}^2 \, | \, \rho^* - \delta \rho < \rho(x,t) < \rho^* \}$ where the parameter $\delta \rho >0$ must be suitably tuned. Therefore, with the initial model (\ref{Eq:rho}), (\ref{Eq:Omega}), the definition of a herd depends on an arbitrary parameter $\delta \rho$, which makes it ambiguous. 

A way to unambiguously define the herd is to force the system (\ref{Eq:rho})-(\ref{Eq:Omega}) to make clear-cut phase transitions from unclustered $\rho < \rho^{\ast}$ to clustered $\rho = \rho^{\ast}$ phases. In the spirit of the works \cite{2008_Traffic_DegondRascle,2008_TraficNum_D_Delitala,2008_Trafficflow_Berthelin_D} for traffic, this can be achieved in an asymptotic regime which amounts to supposing that there is merely no repulsive interactions at all as long as $\rho < \rho^{\ast}$, and that repulsive "pressure" forces turn on suddenly when $\rho$ hits the congestion density $\rho^{\ast}$. This can be done by rescaling $p(\rho)$ into $\eps p(\rho)$ where $\eps \ll 1$ is a small parameter. In this way, repulsive interactions are $O(\eps)$ as long as $\rho < \rho^{\ast}$, but become $O(1)$ when $\rho  = \rho^{\ast}$ (see fig. \ref{Fig:pressure}).

\begin{figure}
\begin{center}\null
\hfill
\psfrag{rho}{$\rho$}
\psfrag{p(rho)}{$p(\rho)$}
\subfigure[$\eps = 1$]{\includegraphics[width=0.45\textwidth]{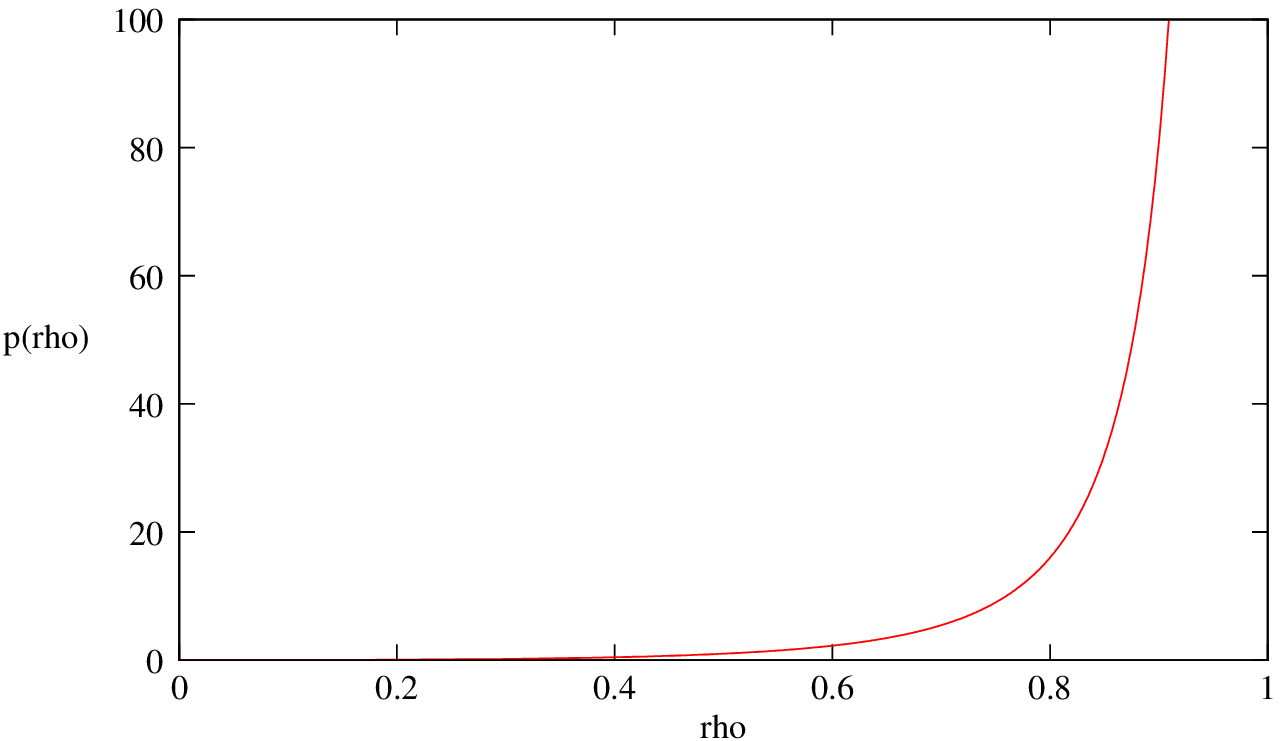}}
\hfill
\psfrag{rho}{$\rho$}
\psfrag{p(rho)}{$\eps p(\rho)$}
\subfigure[$\eps = 10^{-2}$]{\includegraphics[width=0.45\textwidth]{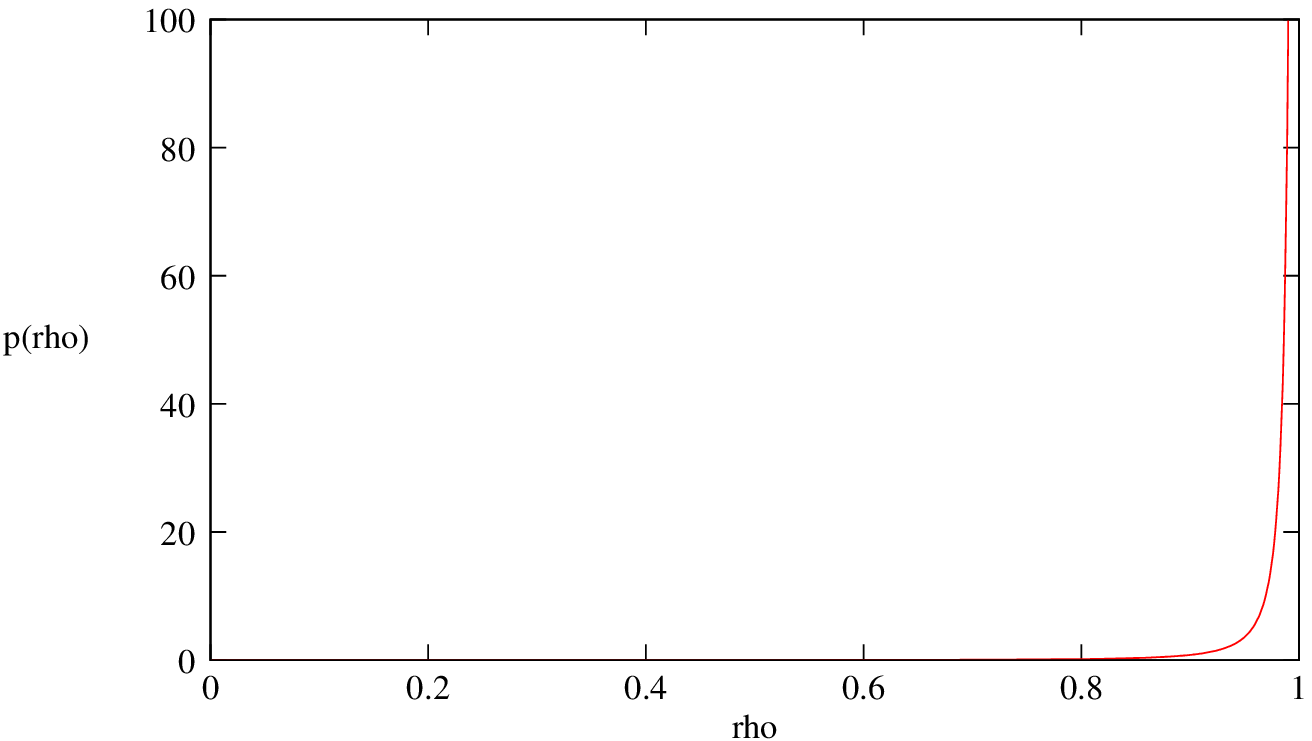}}
\hfill
\null
\caption{The "potential" for repulsive interaction $p(\rho)$ (left) and $\eps p(\rho)$ (right) after scaling by a small parameter $\eps = 10^{-2}$, with $\gamma = 2$ and $\rho^{\ast} = 1$. From the right picture, it is clear that the repulsive interaction turns on only when $\rho$ is very close to $\rho^{\ast}$}
\label{Fig:pressure}
\end{center}
\end{figure}
Biologically, this assumption amounts to saying that the animals do not change their directed motion by the presence  of their neighbours unless they touch them  and need to modify their trajectory to bypass them. The parameter $\eps \ll 1$ is related to the time scale at which this change of trajectory occurs  and is therefore supposed small. Let us also note that our model considers that all animals move with speed unity and never stop. Obviously the model will require improvements by taking into account the fact that a certain fraction of animals are steady, while foraging or resting. 

Therefore, our main concern in this paper is the study of the following perturbation problem:
\begin{eqnarray}
&&\partial_{t}\rho^{\eps} + \nabla_{\vec{x}}\cdot(\rho^{\eps}\Omega^{\eps}) = 0,\label{Eq:rho_eps}\\
&&\partial_{t}\Omega^{\eps} + (\Omega^{\eps}\cdot\nabla_{\vec{x}})\Omega^{\eps} + \eps(\mbox{Id} - \Omega^{\eps}\otimes\Omega^{\eps})\nabla_{\vec{x}} p(\rho^{\eps}) = 0,\label{Eq:Omega_eps}\\
&&|\Omega^{\eps}| = 1.\label{Eq:constraint_eps}
\end{eqnarray} 
We will be interested in the formal limit $\eps \rightarrow 0$. A rigorous theory of this type of problems is unfortunately still out of reach up to our knowledge. In the following section, we show that the limit $\eps \rightarrow 0$ leads to a phase transition between compressible and incompressible regimes.

\subsection{The singular limit $\eps \rightarrow 0$: transition between compressible and incompressible motion}

As $\eps \rightarrow 0$, $\eps p(\rho^{\eps})$ becomes significant only where the convergence $\rho^{\eps} \rightarrow \rho^{\ast}$ is fast enough. Therefore, in the limit, either $\rho^{\eps} \rightarrow \rho < \rho^{\ast}$ and $\eps p(\rho^{\eps}) \rightarrow 0$ or $\rho^{\eps} \rightarrow \rho^{\ast}$ and $\eps p(\rho^{\eps}) \rightarrow \bar{p}$ with $\bar{p}$ possibly non zero. In other words, the equation $(\rho^{\ast} - \rho)\bar{p} = 0$ holds in the limit. If additionally $\bar{p} < +\infty$, straighforward inspection shows that  
\begin{equation}
\rho^{\ast} - \rho^{\eps} = O(\eps^{\frac{1}{\gamma}}).
\label{speedconvergence}
\end{equation}
Therefore, the formal limit $\eps \rightarrow 0$ of system (\ref{Eq:rho_eps})-(\ref{Eq:Omega_eps})-(\ref{Eq:constraint_eps}) is given by the following system:
\begin{eqnarray} 
&&\partial_{t}\rho + \nabla_{\vec{x}}\cdot (\rho\Omega) = 0,\label{Eq:rho_lim}\\
&&\partial_{t}\Omega + \Omega \cdot \nabla_{\vec{x}}\Omega + (\mbox{Id} -  \Omega\otimes\Omega)\nabla_{\vec{x}}\bar{p}
= 0,\label{Eq:Omega_lim}\\
&&|\Omega | = 1\label{Eq:constraint_Omega_lim},\\
&&(\rho^{\ast} -\rho)\bar{p} = 0.\label{Eq:constraint_rho_lim}
\end{eqnarray}
In the non-congested domain $\rho < \rho^{\ast}$, the system reduces to a pressureless compressible gaz dynamics model with a speed constraint
\begin{eqnarray}
&&\partial_{t}\rho + \nabla_{\vec{x}}\cdot \rho\Omega = 0,\\
&&\partial_{t}\Omega + \Omega \cdot \nabla_{\vec{x}}\Omega
= 0,\\
&&|\Omega | = 1.
\end{eqnarray}
This system describes the behaviour of the system outside the congested region. It is a compressible system. Biologically, it describes the behaviour of dispersed animals outside the herd.  Mathematical studies of this system are outside the scope of this article and the reader can refer to \cite{PGD_Bouchut} for standard pressureless gas dynamics models (without speed constraint). We note that this system exhibits vacuum regions where $\rho = 0$ as it will be seen below. 

\subsection{Study of the congested region}

The congested part of the flow is defined as the region where the congestion constraint $\rho = \rho^{\ast}$ is reached. Biologically, it defines the domain of space occupied by the herd. Its connected components will be called "clusters". In the congested domain, system (\ref{Eq:rho_lim})-(\ref{Eq:constraint_rho_lim}) turns into an incompressible Euler model with speed constraint:
\begin{eqnarray} 
&&\nabla_{\vec{x}}\cdot\Omega = 0,\label{Eq:Incomp}\\
&&\partial_{t}\Omega + \Omega \cdot \nabla_{\vec{x}}\Omega 
+ (\mbox{Id} -  \Omega\otimes\Omega)\nabla_{\vec{x}}\bar{p} = 0.\label{Eq:Omega_lim_bis}\\
&&|\Omega | = 1,\label{Eq:constraint_Omega_lim_bis}\\
&&\rho = \rho^{\ast}\label{Eq:constraint_rho_max}, 
\end{eqnarray}

We first note that smooth incompressible vector fields of constant norm in $\R^{2}$ have a very special structure which is outlined in the following. 

\begin{proposition} Let $\Omega(x)$ be a smooth  vector field on a domain $\Theta \subseteq \R^{2}$ with values in $\S^{1}$ and which satisfies the incompressibility constraint $\nabla_{\vec{x}}\cdot\Omega = 0$. Then the integral lines of $\Omega^{\perp}$ are straight lines and $\Omega$ is constant along these lines (where $\Omega^{\perp}$ is rotated by an angle of $\pi/2$) and the integral lines of $\Omega$ are parallel curves to each other.\label{Prop:OmegCongest}
\end{proposition} 
The proof of this proposition simply results from introducing the angle $\theta$ so that $\Omega(\vec{x},t) = (\cos(\theta(\vec{x},t)),\sin(\theta(\vec{x},t)))$ and noting that $\theta$ satisfies the "transport equation"
\begin{equation*}
\partial_{x_{2}}\theta - (\tan\theta)\partial_{x_{1}}\theta= 0.
\end{equation*}
This property implies that the knowledge of $\Omega$ on the cluster boundaries suffices to know $\Omega$ everywhere inside the clusters. 

The integral curves of $\Omega$ provide a mathematical description of the animal files in the herd. These curves being parallel to each other, they are consistant with the intuition and the observation of animal files in a herd (see fig. \ref{Fig:herd}).

\begin{figure}
\begin{center}
\null
\hfill
\psfrag{rhoast}{$\rho = \rho^{\ast}$}
\psfrag{rho<rhoast}{$\rho < \rho^{\ast}$}
\subfigure{\includegraphics[width=0.45\textwidth]{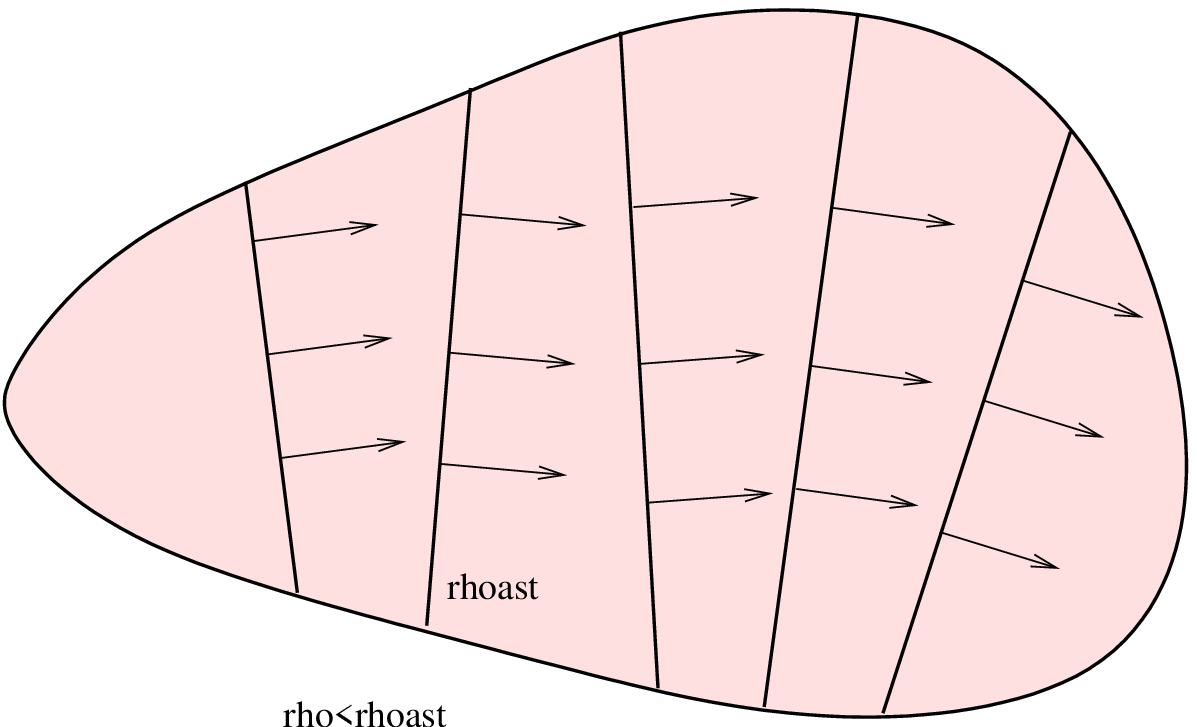}}
\hfill
\psfrag{n}{\textcolor{blue}{$n$}}
\psfrag{Ol}{$\Omega_{\ell}$}
\psfrag{Or}{$\Omega_{r}$}
\psfrag{rhoast}{$\rho_{r} = \rho^{\ast}$}
\psfrag{rho<rhoast}{$\rho_{\ell} < \rho^{\ast}$}
\psfrag{thl}{$\theta_{\ell}$}
\psfrag{thr}{$\theta_{r}$}
\psfrag{x2}{\footnotesize{$x_{2}$}}
\psfrag{x1}{\footnotesize{$x_{1}$}}
\subfigure{\includegraphics[width=0.45\textwidth]{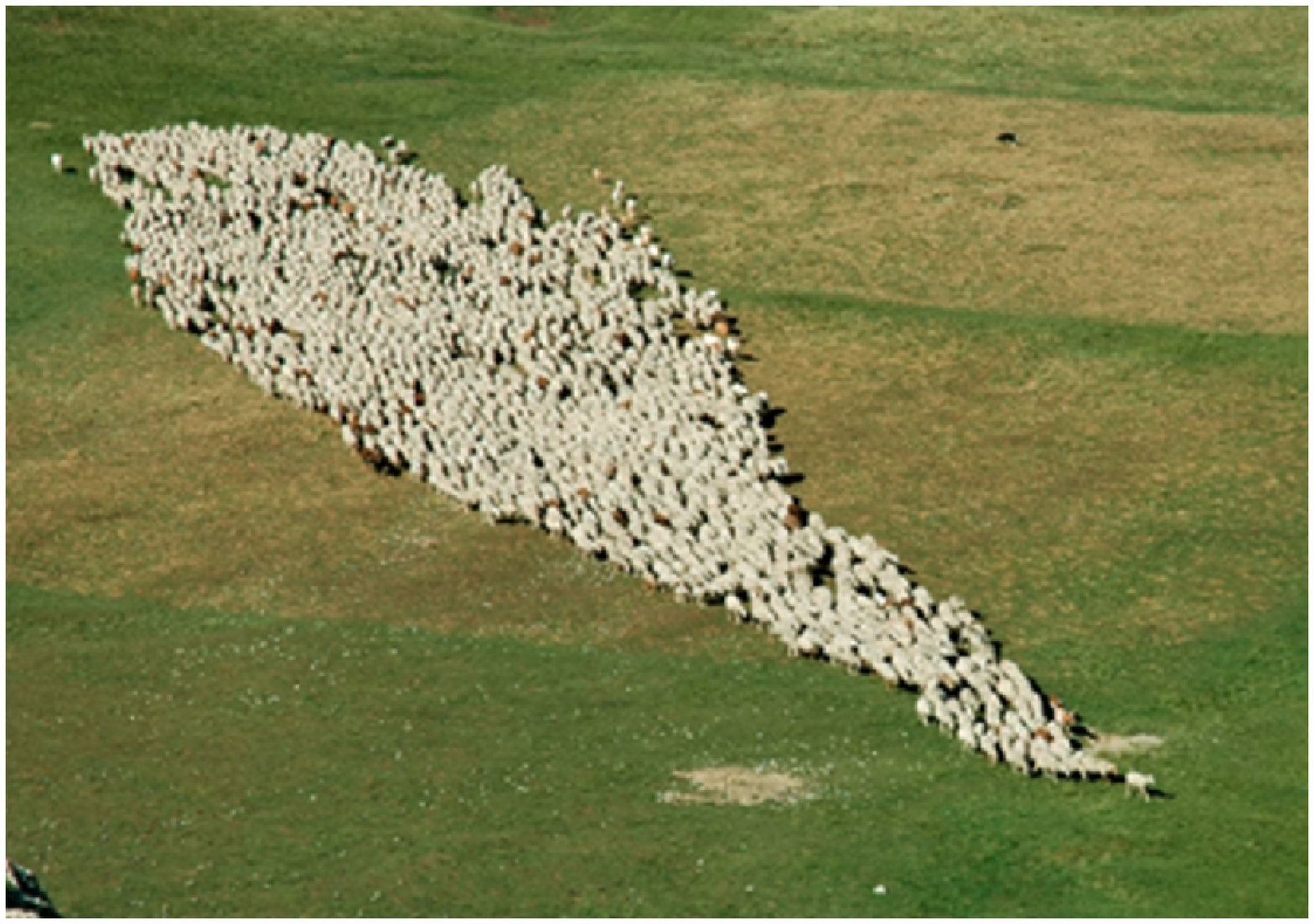}}
\hfill
\null
\caption{Left: schematic figure of a congested zone, where the arrows design the vectors $\Omega$. Right: picture of a sheep herd.}
\label{Fig:herd}
\end{center}
\end{figure}

The pressure $\bar{p}$ satisfies an elliptic equation. Indeed, by taking the divergence of the equation (\ref{Eq:Omega_lim_bis}) and after easy computations, we get
\begin{equation}
\nabla_{\vec{x}}\cdot\left(\left(\mbox{Id} -  \Omega\otimes\Omega\right)\nabla_{\vec{x}}\bar{p}\right) = \text{Tr}((\nabla_{\vec{x}}\Omega)(\nabla_{\vec{x}}\Omega)^{T}),
\label{Eq:pressure}
\end{equation}
where Tr is the trace of a matrix and the exponent $T$ denotes the transpose operator. This equation can be equivalently written:
\begin{equation}
- (\Omega^{\perp}\cdot\nabla_{\vec{x}})^{2}\bar{p} - (\nabla_{\vec{x}}\cdot\Omega^{\perp})(\Omega^{\perp}\cdot\nabla_{\vec{x}})\bar{p} = -\text{Tr}((\nabla_{\vec{x}}\Omega)(\nabla_{\vec{x}}\Omega)^{T}),
\label{Eq:pressure_bis}
\end{equation}
and only involves the operator $(\Omega^{\perp}\cdot\nabla_{\vec{x}})$  applied to $\bar{p}$. Since the integral lines of $\Omega^{\perp}$ are straight lines, equation (\ref{Eq:pressure_bis}) is just a one-dimensional elliptic problem for $\bar{p}$ posed on this straight line. Knowing the boundary values of $\bar{p}$ where this straight line meets the boundary of the cluster allows to compute $\bar{p}$ everywhere on this lines and consequently inside the cluster (see fig. \ref{Fig:herd}). Hence, once $\Omega$ is known inside the cluster, the resolution of this equation only requires the knowledge of the boundary conditions for $\bar{p}$ at the boundaries of the cluster. 

To close the system, i.e. to determine how the solution in the congested domain evolves, we need to determine these boundary conditions. They are not given by the formal limit and, in order to determine them, we need to explore another route. For this pupose we look at the solutions of the Riemann problem for the perturbed and limit systems. Note that if we abandon the constraint of constant norm $|\Omega | = 1$, the non conservative term $(\Omega\otimes\Omega)\nabla_{\vec{x}}\bar{p}$ in the momentum conservation equation (\ref{Eq:Omega_lim}) drops out, and we recover a conservative model expressing mass and momentum equation. Then, the Rankine-Hugoniot conditions across the boundary between the compressible and incompressible regions provide the boundary conditions for the pressure at the cluster boundary. The constant norm constraint prevents from using this strategy. Therefore, we need to find a different route to specify these boundary conditions. 

\subsection{Conditions at the boundary of the clusters}

To find the boundary conditions on the cluster boundaries, we need to extract more information from the perturbation system (\ref{Eq:rho_eps})-(\ref{Eq:constraint_eps}) than the mere limit system (\ref{Eq:rho_lim})-(\ref{Eq:constraint_rho_lim}). As such, this system is underdetermined. The strategy is to extract such information by passing to the limit $\eps \rightarrow 0$ in some special solutions of this system. To underline the difficulty resulting from the non-conservativity, let us first look at the Rankine-Hugoniot conditions. We have the following proposition.  
\begin{proposition} 
\begin{enumerate}
	\item If $\rho$ and $\Omega$ are smooth on both sides of a dicontinuity line $\Gamma$, then we have 
	\begin{equation}
	   \left[\rho(\Omega\cdot \vec{n} - \sigma)\right]_{\Gamma} = 0,
	\end{equation}
	where $\vec{n}$ is the normal to $\Gamma$ and $\sigma$ is the speed of $\Gamma$.
	\item If $\Omega$ is smooth (i.e. $\mathcal{C}^{1}$) across $\Gamma$ and $\rho$ is smooth on both sides of $\Gamma$, then we have 
	\begin{equation}
	   \left[ \bar{p}\right]_{\Gamma}(\Omega\cdot \vec{n}^{\bot}) = 0,
	\end{equation}
	where $\vec{n}^{\bot}$ is a unit vector tangent to $\Gamma$.
\end{enumerate}
\label{Prop:RH_partiel}
\end{proposition}
The proof of this proposition is omitted. The second relation provides us information when the mean velocity is not tangent to the cluster. In this condition, if the mean velocity is continuous, the pressure is also continuous. This implies that the pressure is zero on a cluster boundary if the mean velocity is continuous. This fact will be supported by the forthcoming analysis. However, as regards the interface dynamics, such an analysis is incomplete because the second equation supposes that $\Omega$ is continuous. 

\begin{figure}
\begin{center}
\null
\hfill
\psfrag{n}{\textcolor{blue}{$n$}}
\psfrag{Ol}{$\Omega_{\ell}$}
\psfrag{Or}{$\Omega_{r}$}
\psfrag{rhoast}{$\rho_{r} = \rho^{\ast}$}
\psfrag{rho<rhoast}{$\rho_{\ell} < \rho^{\ast}$}
\psfrag{thl}{$\theta_{\ell}$}
\psfrag{thr}{$\theta_{r}$}
\psfrag{x2}{\footnotesize{$x_{2}$}}
\psfrag{x1}{\footnotesize{$x_{1}$}}
\subfigure{\includegraphics[scale=0.55]{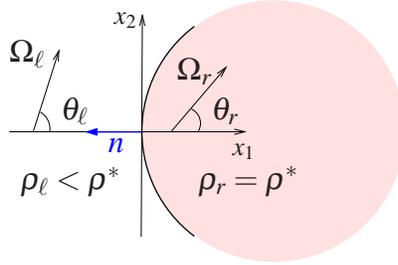}}
\hfill
\null
\caption{Notations at the interface.}
\label{Fig:notation_interface}
\end{center}
\end{figure}

So as to capture the correct boundary conditions for the pressure $\bar{p}$ and the velocity $\Omega$ at a cluster boundary where $\rho$ and $\Omega$ may be discontinuous, we consider a one dimensional problem  in the normal direction $n$ to this boundary (cf. figure \ref{Fig:notation_interface}). In order to justify this simplication, we introduce the coordinate system $(x_{1},x_{2})$ in the normal and tangent direction to the boundary. The angle $\theta$ is defined so that $\Omega(\vec{x},t) = (\cos(\theta(\vec{x},t)),\sin(\theta(\vec{x},t)))$ in this basis. System (\ref{Eq:rho_eps})-(\ref{Eq:Omega_eps}) then becomes (the index $\eps$ is omitted):
\begin{eqnarray}
&&\hspace{-1.2cm}\partial_{t}\rho + \partial_{x_{1}}(\rho\cos\theta) + \partial_{x_{2}}(\rho\sin\theta)  = 0,\label{Eq:rho_eps_theta}\\
&&\hspace{-1.2cm}\left[\partial_{t}\theta + (\cos\theta\partial_{x_{1}}\theta + \sin\theta\partial_{x_{2}}\theta)
+ \left(-\sin\theta\partial_{x_{1}}\eps p(\rho) + \cos\theta\partial_{x_{2}} \eps p(\rho)\right)\right]\left(\begin{array}{c} -\sin\theta\\\cos\theta \end{array}\right) = 0.\label{Eq:theta_eps_theta}
\end{eqnarray}
We suppose that all quantities have locally smooth variations in the direction tangent to the boundary and we focus on the possible sharp variations or discontinuities in the normal direction. To analyze this situation, we perform a coordinate dilation in the $x_{1}$ direction and in time: $x_{1}' = \delta x_{1}$, $x_{2}' = \delta x_{2}$, $t' = \delta t$, with $\delta \ll 1$. In these new variables, all $x_{1}$ and $t$ derivatives are multiplied by $1/\delta$. Letting $\delta \rightarrow 0$, we are led to the following one-dimensional system with $x_{1} = x$:
\begin{eqnarray}
&&\partial_{t}\rho + \partial_{x}(\rho\cos\theta) = 0,\label{Eq:rho_1d}\\
&&\partial_{t}\theta + \cos\theta\partial_{x}\theta + \eps\sin^{2}\theta\partial_{x}\eps p(\rho)  = 0.\label{Eq:theta_1d}
\end{eqnarray}
Hyperbolic systems like (\ref{Eq:rho_1d})-(\ref{Eq:theta_1d}) have analytical solutions which are those of the Riemann problem. These solutions are associated to initial conditions which consist of a discontinuity between two constant states. We will construct the solutions of the Riemann problem for system (\ref{Eq:rho_1d})-(\ref{Eq:theta_1d}) and analyze their limits as $\eps \rightarrow 0$. This analysis will give rise to jump conditions at the cluster boundaries for these solutions. We will then postulate that these jump conditions are generic and valid for all solutions of the limit problem (\ref{Eq:rho_lim})-(\ref{Eq:constraint_rho_lim}).

As underlined above, the non-conservative form of system (\ref{Eq:rho_1d})-(\ref{Eq:theta_1d}) induces a lack of information about the jump conditions across a boundary. In order to waive the ambiguity, we have to make further assumptions. One of them is to consider the following conservative system as a way to select discontinuities
\begin{eqnarray}
&&\partial_{t}\rho + \partial_{x}(\rho\cos\theta) = 0,\label{Eq:rho_1D_cons}\\
&&\partial_{t}\Psi(\cos(\theta)) + \partial_{x}(\Phi(\cos\theta) + \eps p(\rho)) = 0.\label{Eq:theta_1D_cons}
\end{eqnarray}
where $\Psi(\cos\theta) = - \ln|\tan(\theta/2)|$ and $\Phi(\cos\theta) = - \ln|\sin\theta|$. It is the simplest conservation form that system (\ref{Eq:rho_1d})-(\ref{Eq:theta_1d}) can take. It is obtained by dividing (\ref{Eq:theta_1d}) by $\sin^{2}\theta$. The functions $\Psi$ and $\Phi$ satisfy:
\begin{equation} 
\frac{d}{d\theta}(\Psi(\cos\theta)) = \frac{1}{\sin^{2}\theta},\quad \frac{d}{d\theta}(\Phi(\cos\theta)) = \frac{\cos\theta}{\sin^{2}\theta}.
\end{equation} 
Other conservative forms of (\ref{Eq:rho_1d})-(\ref{Eq:theta_1d}) do exist (see appendix \ref{Appendix:conservative_laws}) but we consider this form because it is the simplest. Note that this conservative form is not equivalent to the original form (\ref{Eq:rho_1d})-(\ref{Eq:theta_1d}) because $\Psi(\cos\theta)$ is an even function of $\theta$. Hence it does not provide information on the sign of $\theta$. However, this information will easily be recovered at the end. We remind that, if all conservative forms are equivalent for smooth solutions, they differ for weak solutions. Therefore, the choice of a particular conservative form must be made on physical considerations. Such physical considerations are not available here. In front of this lack of information, the choice of the simplest of these conservative forms seems to be the most natural one. 

Classical hyperbolic system theory will enable us to solve the Riemann problem for (\ref{Eq:rho_1D_cons})-(\ref{Eq:theta_1D_cons}) and to take the limit $\eps \rightarrow 0$ of these solutions. The limit solutions will satisfy some jump relations which we will assume generic of all solutions of the limit problem (\ref{Eq:rho_lim})-(\ref{Eq:constraint_rho_lim}). We now present the result of this analysis for such generic solutions.  We call "unclustered" region (UC) the domains where $0<\rho<\rho^{\ast}$, by contrast to vacuum (V) where $\rho = 0$ or clusters (C) where $\rho = \rho^{\ast}$. 
\begin{Formal Statement}\label{Formal_Statement}
The boundary conditions at cluster boundaries or vacuum boundaries of system (\ref{Eq:rho_lim})-(\ref{Eq:constraint_rho_lim}) are as follows:
\begin{itemize} 
\item {\bf Interface (C)-(UC).} The pressure jump is given by
\begin{equation}
[\bar{p}] = \frac{\left[\Psi(\cos\theta)\right]\left[\rho\cos\theta\right]}{\left[\rho\right]}  - \left[\Phi(\cos\theta)\right]
\label{Eq:FS_CUC_pressure}
\end{equation}
 and the shock speed is given by the Rankine-Hugoniot relation 
\begin{equation}
\sigma = [\rho\cos(\theta)]/[\rho],
\label{Eq:FS_CUC_speed}
\end{equation}
where the angle brackets denote the jumps across the interface. We note that $p_{UC} = 0$ and that specifying $[\bar{p}]$ actually specifies the boundary value of $\bar{p}$ at the cluster boundary.
\item {\bf Interface (UC)-(V).} The interface speed $\sigma$ is equal to the fluid normal speed $\sigma = \cos\theta = \Omega\cdot \vec{n}$ at the boundary of the (UC)  region
\begin{equation}
\sigma = (\cos\theta)_{UC},
\label{Eq:FS_UCV_speed}
\end{equation}
and the pressure $\bar{p}$ is identically zero. 
\item {\bf Interface (C)-(V).} The interface speed is equal to the normal speed $\cos\theta = \Omega\cdot \vec{n}$ at the cluster boundary  and the boundary value of $\bar{p}$ is zero
 \begin{equation}
\sigma = (\cos\theta)_{C},\quad \bar{p}_{C} = 0.
\label{Eq:FS_CV_speed_pressure}
\end{equation}
\item {\bf Interface (UC)-(UC).} This is a contact discontinuity between two regions of different $\rho$. The normal velocity  is continuous and equal to the speed of the discontinuity 
\begin{equation}
[\cos\theta] = 0,\quad \sigma = \cos\theta,
\label{Eq:FS_UCUC_speed}
\end{equation}
and the pressure is identically zero.
\end{itemize}
\end{Formal Statement}
We note that all these statements are consistent with proposition \ref{Prop:RH_partiel}. Section \ref{Chap:Riemann_Pb} provides the  detailed analysis which leads to these relations. The dynamics of the interface between two clusters (C)-(C) does not follow from the analysis of the Riemann problem. We provide a separate analysis of it by introducing the so-called cluster dynamics.       

We note that (C)-(UC) interfaces or contact discontinuities (UC)-(UC) may incorporate a flip of the sign of $\theta$ in the velocity jump. This has no influence on the boundary values of $\bar{p}$ at the cluster boundary which is the quantity we wish to determine by this analysis. The Formal Statement \ref{Formal_Statement} is illustrated in figure \ref{Fig:Formal_Statement}.
\begin{figure}
\begin{center}
\null
\hfill
\psfrag{rho}{$\rho$}
\psfrag{x}{$x$}
\psfrag{rhoast}{$\rho^{\ast}$}
\psfrag{thetal}{$\theta_{\textrm{\tiny  C}}$}
\psfrag{thetar}{$\theta_{\textrm{\tiny  UC}}$}
\psfrag{pl}{$\bar{p}_{\textrm{\tiny  C}}$}
\psfrag{pr}{$\bar{p}_{\textrm{\tiny  UC}} = 0$}
\psfrag{s}{$s$}
\subfigure[Interface (C)-(UC). $\bar{p}_{\textrm{\tiny  C}}$ and $s$ are given by (\ref{Eq:FS_CUC_pressure})-(\ref{Eq:FS_CUC_speed}).]{\includegraphics[width=0.45\textwidth]{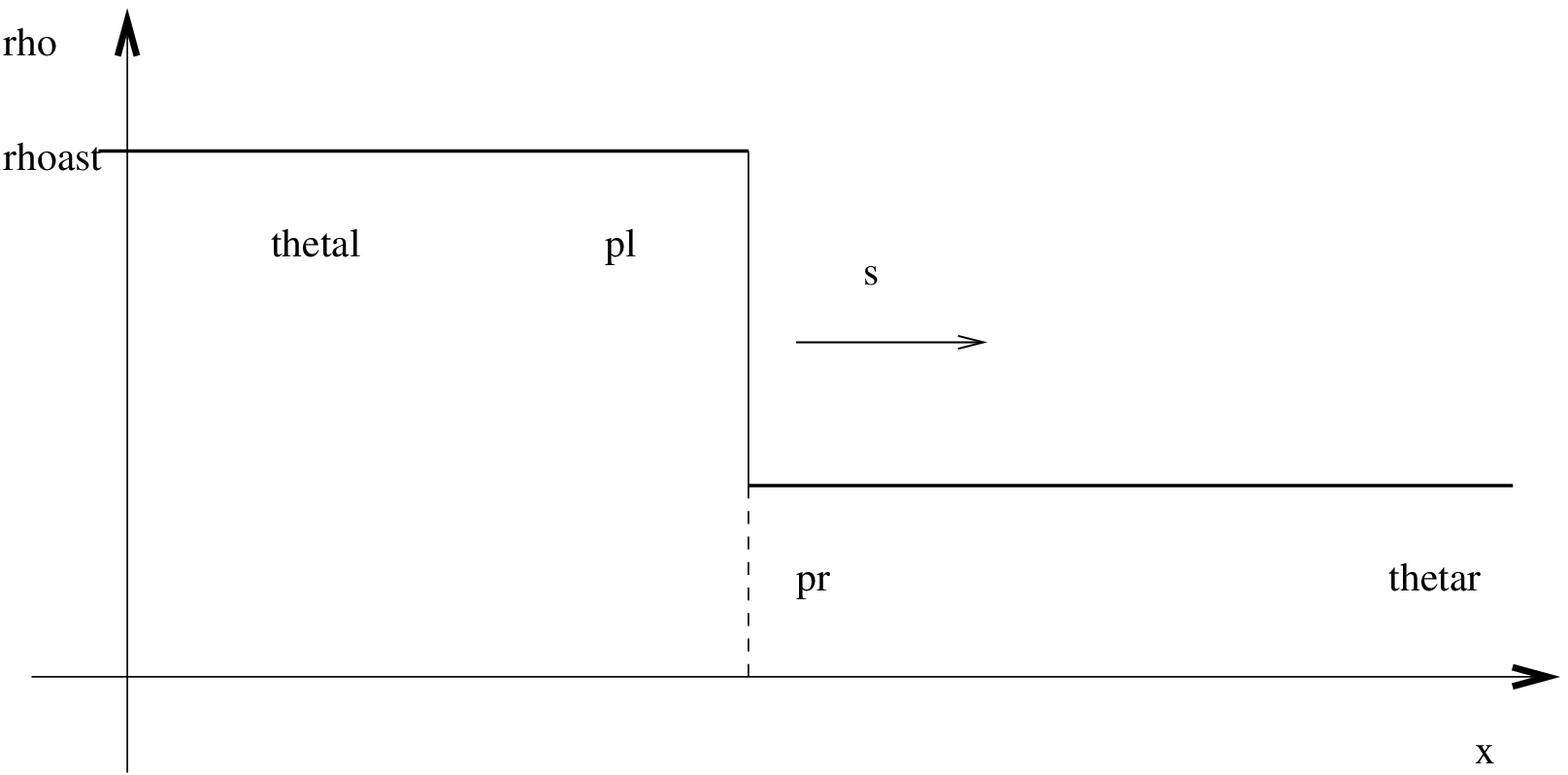}}
\hfill
\psfrag{rho}{$\rho$}
\psfrag{x}{$x$}
\psfrag{rhoast}{$\rho^{\ast}$}
\psfrag{thetal}{$\theta_{\textrm{\tiny  UC}}$}
\psfrag{rhol}{$\rho_{\textrm{\tiny  UC}}$}
\psfrag{pl}{$\bar{p}_{\textrm{\tiny  UC}} = 0$}
\psfrag{s}{$\cos\theta_{\textrm{\tiny  UC}}$}
\subfigure[Interface (UC)-(V)]{\includegraphics[width=0.45\textwidth]{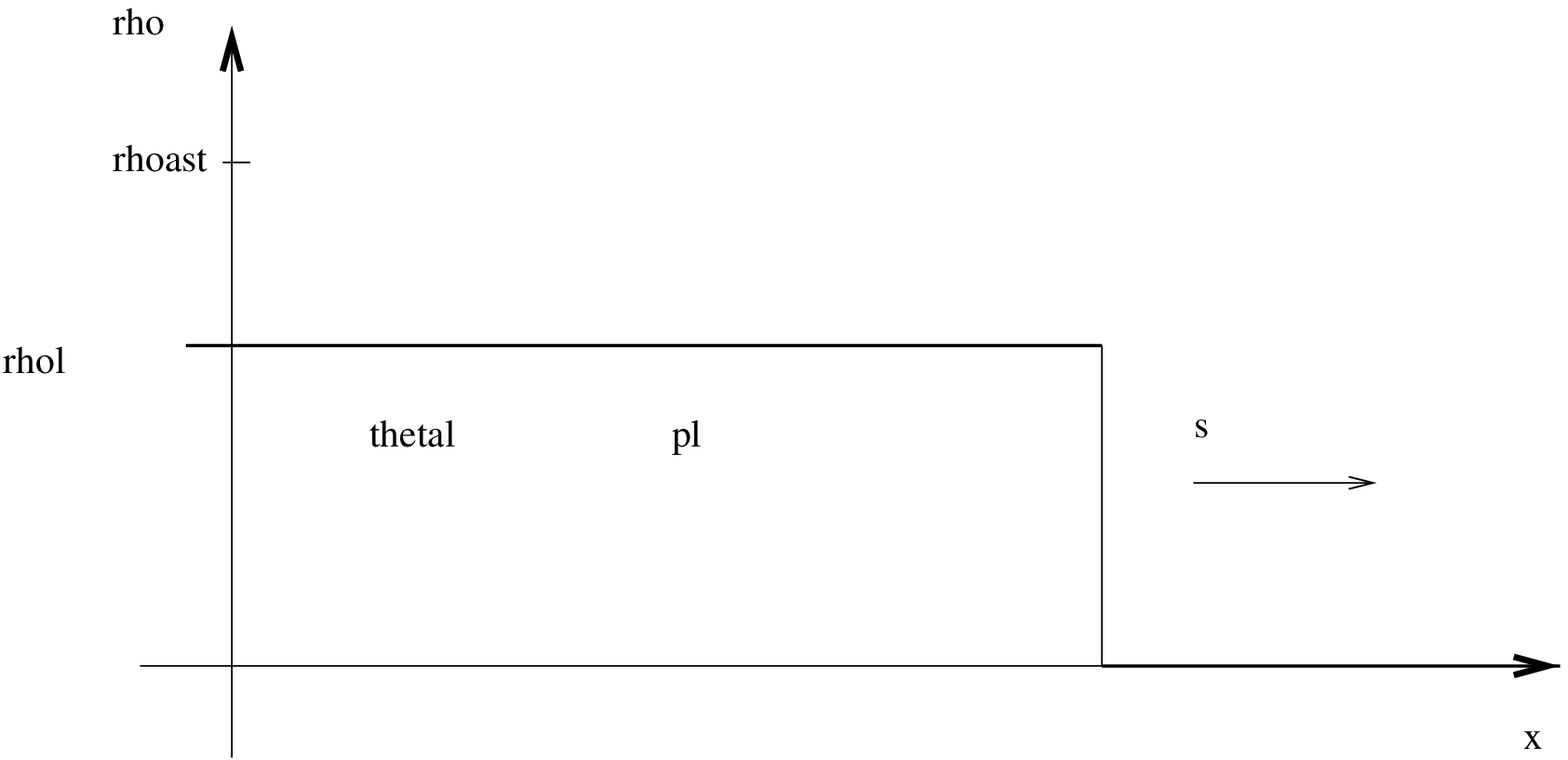}}
\hfill
\null

\null
\hfill
\psfrag{rho}{$\rho$}
\psfrag{x}{$x$}
\psfrag{rhoast}{$\rho^{\ast}$}
\psfrag{thetal}{$\theta_{\textrm{\tiny  C}}$}
\psfrag{pl}{$\bar{p}_{\textrm{\tiny  C}} = 0$}
\psfrag{s}{$\cos\theta_{\textrm{\tiny  C}}$}
\subfigure[Interface (C)-(V)]{\includegraphics[width=0.45\textwidth]{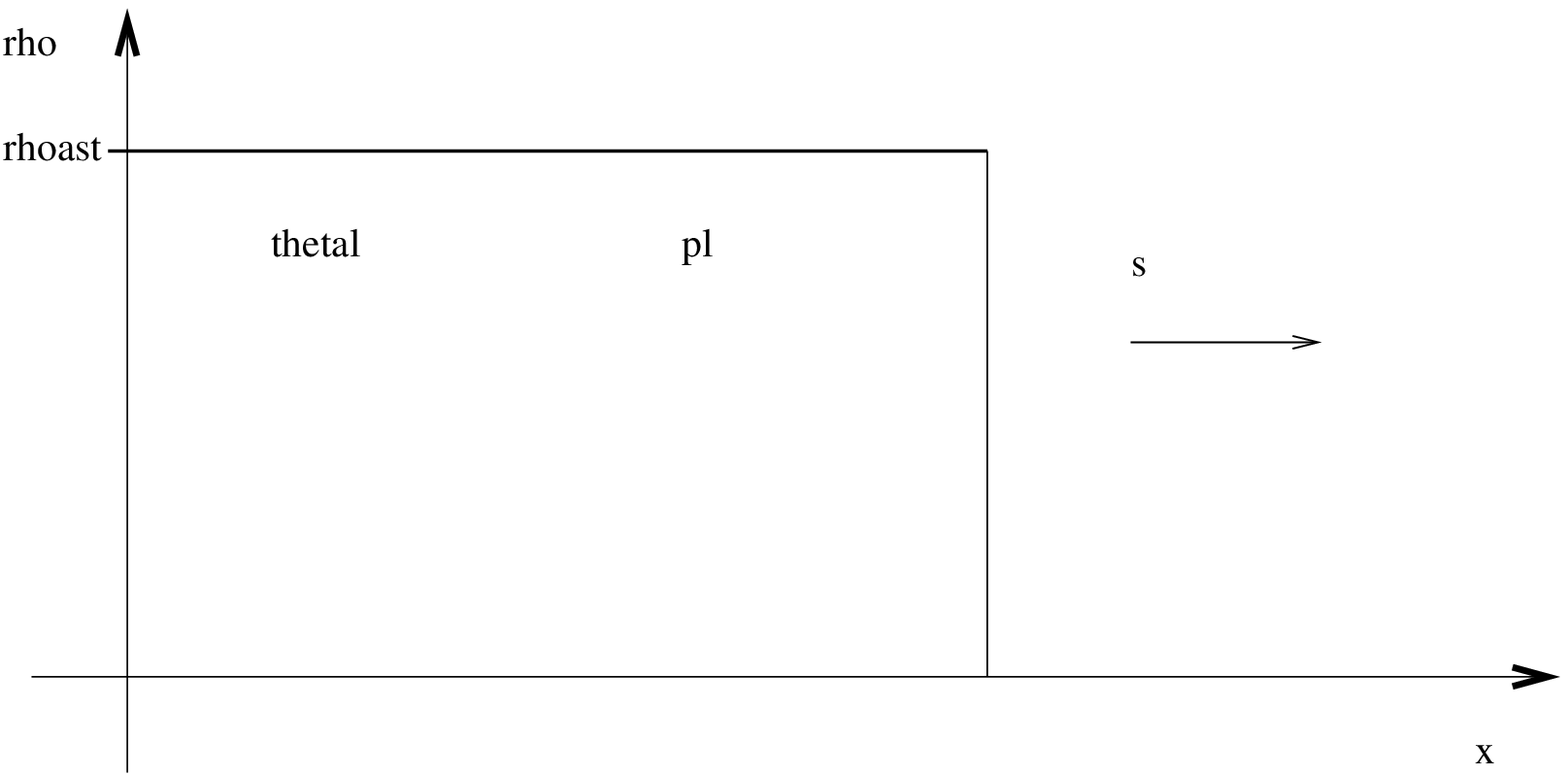}}
\hfill
\psfrag{rho}{$\rho$}
\psfrag{x}{$x$}
\psfrag{rhoast}{$\rho^{\ast}$}
\psfrag{thetal}{$\theta_{\textrm{\tiny  UC}}$}
\psfrag{thetar}{$\theta_{\textrm{\tiny  UC}}$}
\psfrag{pl}{$\bar{p}_{\textrm{\tiny  UC}} = 0$}
\psfrag{pr}{$\bar{p}_{\textrm{\tiny  UC}} = 0$}
\psfrag{s}{$\cos\theta_{\textrm{\tiny  UC}}$}
\subfigure[Interface (UC)-(UC)]{\includegraphics[width=0.45\textwidth]{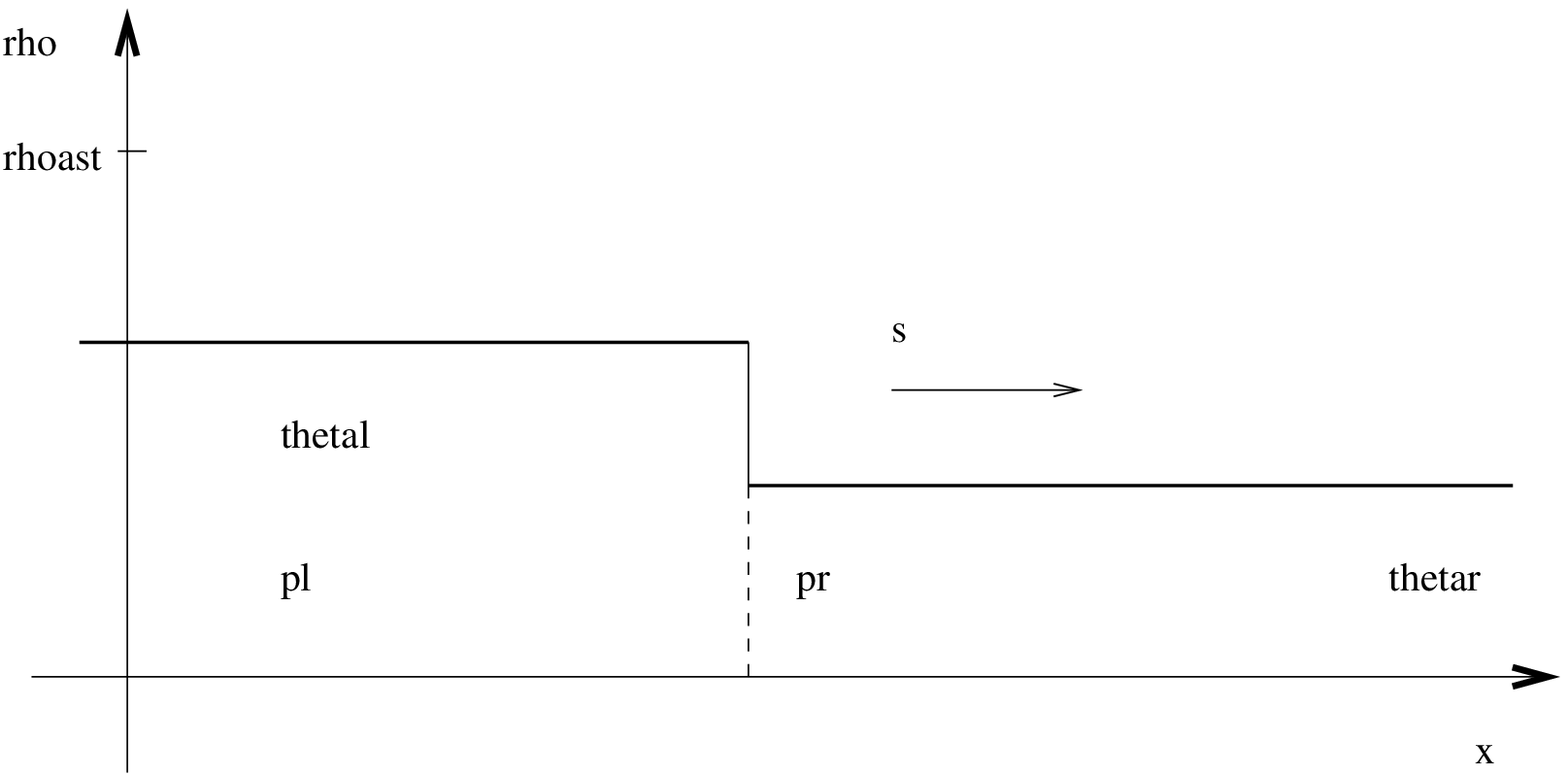}}
\hfill
\null

\caption{Interfaces}
\label{Fig:Formal_Statement}
\end{center}
\end{figure}


\subsection{Clusters dynamics}
\label{Chap:Clusters_dynamics}

We now focus on the interface (C)-(C), i.e. a collision of two clusters. The procedure using limits $\eps \rightarrow 0$ of the Riemann problem does not lead to any conclusion since the pressure becomes infinite. Note that this is also the case when dealing with the same limit in the standard Euler problem. Therefore, we have to find another strategy than using the Riemann problem. We turn our attention to the collision between two clusters of finite size and we show that the pressure involves a Dirac delta at the time of the collision. Such an analysis is inspired by the sticky block solutions presented in \cite{2PhaseFlow_Bouchut_al}.

Consider two one-dimensional clusters which collide at a time $t_{c}$ (see fig. \ref{Fig:cluster_collision}). Before collision, the left (resp. right) cluster at time $t < t_{c}$ extends between $a_{\ell}(t)$ and $b_{\ell}(t)$ (resp. $a_{r}(t)$ and $b_{r}(t)$) and moves with speed 
\begin{equation}
\cos\theta_{\ell} = a_{\ell}'(t) = b_{\ell}'(t)\quad \text{(resp.}\ \cos\theta_{r} = a_{r}'(t) = b_{r}'(t)\text{)},
\end{equation} 
After the collision, the two clusters agregate and form a new cluster at time $t > t_{c}$ extending between $a(t)$ and $b(t)$ and moving with speed $\cos\theta = a'(t) = b'(t)$. Therfore, $\rho$ and $\theta$ are given for $t < t_{c}$ by
\begin{equation*} 
\rho = \rho^{\ast}\mathds{1}_{[a_{\ell}(t),b_{\ell}(t)]} + \rho^{\ast}\mathds{1}_{[a_{r}(t),b_{r}(t)]},\quad \theta = 
\theta_{\ell}\mathds{1}_{[a_{\ell}(t),b_{\ell}(t)]} + \theta_{r}\mathds{1}_{[a_{r}(t),b_{r}(t)]},
\end{equation*}
and for $t > t_{c}$ by
\begin{equation*}
\rho = \rho^{\ast} \mathds{1}_{[a(t),b(t)]},\quad \theta = \theta \mathds{1}_{[a(t),b(t)]}.
\end{equation*}
where $\mathds{1}_{I}$ denotes the indicator function of the interval $I$ (i.e. $\mathds{1}_{I}(x) = 1$ if $x \in I$ and $0$ otherwise). We denote by $m = b_{\ell}(t_{c}) = a_{r}(t_{c})$ the collision point. We look for a pressure written as $\bar{p}(x,t) = \pi(x)\delta(t-t_{c})$. The following proposition provides conditions for such type of solutions to exist. 
\begin{proposition}\label{Prop:cluster_collision} 1- Supposing that $\bar{p}(x,t) = \pi(x)\delta(t-t_{c})$ where $\pi$ is continuous and zero outside the clusters, then $\theta$ and $\pi$ satisfy 
\begin{eqnarray}
&&(\Psi(\cos\theta) - \Psi(\cos\theta_{\ell}))(m - a(t_{c})) + (\Psi(\cos\theta) - \Psi(\cos\theta_{r}))(b(t_{c}) - m) = 0,\nonumber\\
&&\label{Eq:collision_angle}\\
&&\pi(x) = \left\{\begin{array}{ll}(\Psi(\cos\theta) - \Psi(\cos\theta_{\ell}))(m - x)\\
 \quad + (\Psi(\cos\theta) - \Psi(\cos\theta_{r}))(b(t_{c}) - m),&\text{ if }x \in [a(t_{c}),m],\\
(\Psi(\cos\theta) - \Psi(\cos\theta_{r}))(b(t_{c}) - x),&\text{ if }x \in [m,b(t_{c})],\end{array}\right.\label{Eq:collision_pi}
\end{eqnarray}
2 - Under conditions (\ref{Eq:collision_angle})-(\ref{Eq:collision_pi}), $(\rho,\theta,p)$ is a solution (in the distributional sense) of (\ref{Eq:rho_1D_cons})-(\ref{Eq:theta_1D_cons}).
\end{proposition}
The proof of this proposition is developed in appendix \ref{Appendix:cluster_collision}.
\begin{figure}
\begin{center}

\psfrag{a(tc+d)}{\footnotesize{$a(t_{c}+\delta)$}}
\psfrag{b(tc+d)}{\footnotesize{$b(t_{c}+\delta)$}}
\psfrag{a(tc)}{\footnotesize{$a(t_{c})$}}
\psfrag{b(tc)}{\footnotesize{$b(t_{c})$}}
\psfrag{al(tc-d)}{\footnotesize{$a_{\ell}(t_{c}-\delta)$}}
\psfrag{bl(tc-d)}{\footnotesize{$b_{\ell}(t_{c}-\delta)$}}
\psfrag{ar(tc-d)}{\footnotesize{$a_{r}(t_{c}-\delta)$}}
\psfrag{br(tc-d)}{\footnotesize{$b_{r}(t_{c}-\delta)$}}
\psfrag{tc-d}{\footnotesize{$t_{c}-\delta$}}
\psfrag{tc}{\footnotesize{$t_{c}$}}
\psfrag{tc+d}{\footnotesize{$t_{c}+\delta$}}
\psfrag{x0-eta}{\footnotesize{$x_{0}-\eta$}}
\psfrag{x0}{\footnotesize{$x_{0}$}}
\psfrag{x0+eta}{\footnotesize{$x_{0}+\eta$}}
\psfrag{D}{$D$}
\psfrag{D'}{$D'$}
\psfrag{D''}{$D''$}
\psfrag{t}{\footnotesize{$t$}}
\psfrag{x}{\footnotesize{$x$}}
\psfrag{m}{\footnotesize{$m$}}
\subfigure{\includegraphics[scale=0.7]{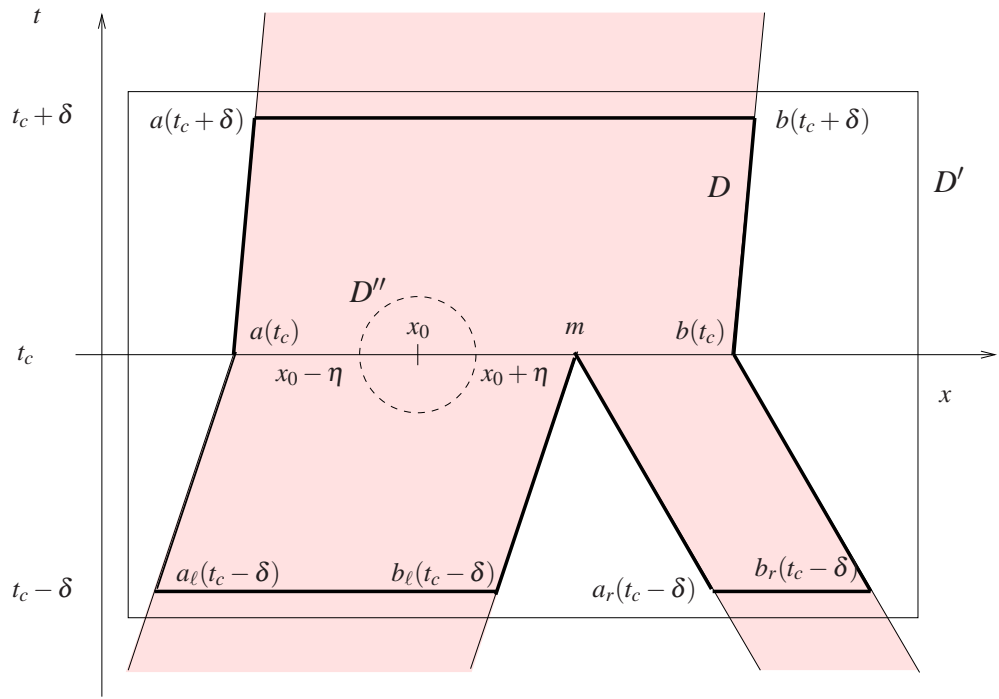}}

\caption{Collision of clusters. In the filled domain: clusters ($\rho = \rho^{\ast}$).}
\label{Fig:cluster_collision}
\end{center}
\end{figure}

\subsection{Conclusion of the analysis}

The underdetermined problem (\ref{Eq:rho_lim})-(\ref{Eq:constraint_rho_lim}) must be complemented with the Formal Statement \ref{Formal_Statement} which determines the boundary values of $\bar{p}$ at cluster boundaries and by proposition \ref{Prop:cluster_collision} which determines the evolution of two clusters when they meet.  Strictly speaking, proposition \ref{Prop:cluster_collision} only gives the collision dynamics of two clusters in dimension 1. In dimension 2, clusters may have complicated shapes. So, the collision dynamics of two clusters in dimension 2 is a complicated problem which will be examined in a future work. At the present stage, problem (\ref{Eq:rho_lim})-(\ref{Eq:constraint_rho_lim}) complemented with statement \ref{Formal_Statement} fully determines the dynamics of the limit system as long as two clusters do not meet.

A rigorous theory of the well-posedness of system (\ref{Eq:rho_lim})-(\ref{Eq:constraint_rho_lim}) complemented with statement \ref{Formal_Statement}  is outside the scope of the present paper. Let us just mention how a time discretized version of the problem can be computed. Suppose that $\rho^{n}(x)$, $\Omega^{n}(x)$, $\bar{p}^{n}(x)$ are approximations of $\rho(x,t^{n})$, $\Omega(x,t^{n})$, $\bar{p}(x,t^{n})$ at time $t^{n} = n\Delta t$. We solve the implicit system
\begin{eqnarray}
&&\frac{\rho^{n+1} - \rho^{n}}{\Delta t} + \nabla_{\vec{x}}\cdot(\rho^{n+1}\Omega^{n+1}) = 0,\\
&&\frac{\Omega^{n+1} - \Omega^{n}}{\Delta t} + (\Omega^{n}\cdot\nabla_{\vec{x}})\Omega^{n+\frac{1}{2}} + (\mbox{Id} -  \Omega^{n+\frac{1}{2}}\otimes\Omega^{n+\frac{1}{2}})\nabla_{\vec{x}}\bar{p}^{n+1} = 0,
\end{eqnarray}
with
\begin{equation}
\Omega^{n+\frac{1}{2}} = \frac{\Omega^{n} + \Omega^{n+1}}{|\Omega^{n} + \Omega^{n+1}|}.
\end{equation}
This form guarantees that $|\Omega^{n+1}|^{2} = |\Omega^{n}|^{2} = 1$ (by taking the dot product by $\Omega^{n+\frac{1}{2}}$ and using that $|\Omega^{n+\frac{1}{2}}| = 1$). $\bar{p}^{n+1}$ is determined by solving the elliptic equation
\begin{equation}
-\nabla_{\vec{x}}\cdot((\mbox{Id} -  \Omega^{n+\frac{1}{2}}\otimes\Omega^{n+\frac{1}{2}})\nabla_{\vec{x}}\bar{p}^{n+1}) = \nabla_{\vec{x}}\cdot((\Omega^{n}\cdot\nabla_{\vec{x}})\Omega^{n+\frac{1}{2}}).
\end{equation}
on every connected component of the cluster region defined at time $t$ by $\{ x \in {\mathbb R}^2 \, | \, \rho^{n+1} (x,t) = \rho^{\ast} \}$. This equation must be supplemented with suitable boundary conditions on $\bar{p}$ at the boundary of the cluster. These boundary conditions are actually given by the Formal Statement \ref{Formal_Statement}, with right-hand sides evaluated at time $t^{n+1}$. The resolution of this equation guarantees that $\nabla\cdot\Omega^{n+1} = 0$ on every connected component of a cluster, and shows that $\rho^{n+1} = \rho^{n} = \rho^{\ast}$ on such a cluster. Of course, the implicitness of the discretization leads to a nonlinear stationary problem, and the question of the existence of solutions for such a problem is not clear. However, intuitively, it seems that the prescription of the boundary values of $\bar{p}$ at cluster boundaries through the Formal Statement \ref{Formal_Statement} leads to a well-posed problem, at least as long as two clusters do not meet.  

\section{The one-dimensional Riemann Problem}
\label{Chap:Riemann_Pb}

\subsection{Methodology}

To find out jump relations satisfied by the solutions of the system (\ref{Eq:rho_lim})-(\ref{Eq:constraint_rho_lim}), the strategy is to solve the Riemann problem of the one-dimensional perturbation system (\ref{Eq:rho_1D_cons})-(\ref{Eq:theta_1D_cons}) and to take the limit $\eps \rightarrow 0$ of its solutions. This strategy was successfully adopted for a model of traffic jams in \cite{2008_Traffic_DegondRascle}. 

We note that the eigenvalues and eigenvectors of the hyperbolic system (\ref{Eq:rho_1D_cons})-(\ref{Eq:theta_1D_cons}) are
\begin{equation}
\lambda_{\pm}^{\eps}(\rho,\theta) = \cos\theta \pm \sqrt{\eps p'(\rho)\rho}|\sin\theta|,\quad  
\vec{r}_{\pm}^{\eps}(\rho,\theta) = \left(\begin{array}{c} \pm \rho|\sin\theta|\\ \sqrt{\eps p'(\rho)\rho}\end{array}\right).
\end{equation} 
In this conservative system, the domain of $\theta$ is restricted to the interval $]0,\pi[$. But this is not a problem since our main concern is to find  the missing conditions on $\bar{p}$ at the cluster boundary, and these only depend on jump conditions as functions of $\cos\theta$.

\subsection{Solutions to the Riemann problem for (\ref{Eq:rho_1D_cons})-(\ref{Eq:theta_1D_cons}).}

\subsubsection{Genuinely nonlinear fields} 

The Lax theorem provides the local entropic solutions of the Riemann problem provided that all the fields are totally genuinely nonlinear ($\nabla \lambda_{\pm}^{\eps} \cdot \vec{r}_{\pm}^{\eps} \neq 0$) or totally linearly degenerate ($\nabla \lambda_{\pm}^{\eps} \cdot \vec{r}_{\pm}^{\eps} = 0$). Unfortunately, the following result implies that the fields are genuinely nonlinear except on a one-dimensional manifold. 
\begin{proposition} \begin{enumerate} \item The linearly degenerate set ($\nabla \lambda_{\pm}^{\eps} \cdot \vec{r}_{\pm}^{\eps} = 0$) consists of two curves $\mathcal{C}^{\eps}_{\pm}$ (each of them corresponds to one characteristic field):
\begin{equation*}
\mathcal{C}^{\eps}_{\pm}  = \left\{\left(\rho,\theta\right),\ \rho \in [0,\rho^{\ast}[,\ \text{cotan}\theta  = \mp G^{\eps}(\rho) \right\}.
\end{equation*}
where 
\begin{equation*}
G^{\eps}(\rho) := \frac{1}{\sqrt{\eps}}\frac{(p''(\rho)\rho + 3p'(\rho))\rho}{(p'(\rho)\rho)^{3/2}}\quad \underset{\rho \rightarrow \rho^{\ast}}{\sim} C\frac{(\rho^{\ast} - \rho)^{\frac{\gamma - 1}{2}}}{\sqrt{\eps}}.
\end{equation*} 
\item For $\gamma = 1$, the linearly degenerate set tends to the straight lines $\left\{\theta = 0\right\}$ and $\left\{\theta = \pi\right\}$ as $\eps$ tends to $0$. For $\gamma > 1$, $\mathcal{C}^{\eps}_{+}$ (resp. $\mathcal{C}^{\eps}_{-}$) is a one to one, onto mapping from $[0,\rho^{\ast}]$ to $[\pi/2,\pi]$ (resp. $[0,\pi/2]$) for all $\eps$, called $\theta^{\eps}_{\text{ld},\pm}(\rho)$. For a fixed $\theta \in ]0,2\pi[$, the inverse map $\rho^{\eps}_{\text{ld},\pm}(\theta)$ satisfies: $\rho^{\ast} - \rho^{\eps}_{\text{ld},\pm}(\theta) = O(\eps^{\frac{1}{\gamma-1}})$.
\end{enumerate}
\end{proposition}
The proof of this proposition is easy and is omitted. Thus, the Lax theorem is valid at least locally in the neighbourhood of all the states except those which are on the one-dimensional manifolds. The second part of the previous proposition shows that all the states have locally genuinely nonlinear fields as $\eps$ tends to $0$. Indeed, even if the state converges to a congested state as $\eps \rightarrow 0$, its convergence is like $O(\eps^{1/\gamma})$ (cf. (\ref{speedconvergence})), which is slower than the convergence of the linearly degenerate field when $\gamma > 1$. Therefore, there exists $\eps '$ such that for all $\eps < \eps '$ the fields of the converging state are genuinely non-linear. It is also trivially the case when $\gamma$ equals $1$. According to standard nonlinear conservation theory \cite{SCLBook_Serre}, (for $\eps$ small enough) the solutions of the Riemann problem consist of two simple waves (shock waves and/or rarefaction waves) of the first and second characteristic fields, separated by constant states.

\subsubsection{Shock and rarefaction waves.} 

A \textbf{shock wave} between two constant states $(\rho_{\ell},\theta_{\ell})$ and $(\rho_{r},\theta_{r})$ travelling with a constant speed $\sigma$ satisfies the Rankine-Hugoniot relations: 
\begin{eqnarray}
\left[\rho\cos(\theta)\right] &=& \sigma\left[\rho\right],\label{Eq:RH_1}\\
\left[\Phi(\cos(\theta)) + \eps p(\rho)\right] &=& \sigma\left[\Psi(\cos(\theta))\right],\label{Eq:RH_2}
\end{eqnarray}
where $\left[f\right] := f_{r} - f_{\ell}$ denotes the difference between the right value and the left value of any quantity $f$. By eliminating $\sigma$ in these equations, we get a non-linear relation between the left and right states:
\begin{equation}
H_{\eps}(\rho_{\ell},\theta_{\ell},\rho_{r},\theta_{r})  := \left[\Phi(\cos(\theta)) + \eps p(\rho)\right]\left[\rho\right] - \left[\Psi(\cos(\theta))\right]\left[\rho\cos(\theta)\right] = 0.
\label{Eq:Shock}
\end{equation} 
With a fixed left state, the zero set of $H_{\eps}$ is called the Hugoniot locus and represents all the admissible right states, connected to this left state by a shock wave. 
\begin{proposition} The Hugoniot locus consists of two Hugoniot curves $\mathcal{H}_{\pm}^{\eps}$ associated to the two caracteristic fields. 
\begin{enumerate}
  	\item The Hugoniot curve $\mathcal{H}_{-}^{\eps}$ associated to $\lambda_{-}^{\eps}$ (resp. $\mathcal{H}_{+}^{\eps}$ to $\lambda_{+}^{\eps}$) is strictly increasing (resp. strictly decreasing) in the $(\rho,\theta)$-plane. Let $h_{-}^{\eps} :\ ](h_{-}^{\eps})^{-1}(0),\pi[ \rightarrow [0,\rho^{\ast}[$ and $h_{+}^{\eps} :\ ]0,(h_{+}^{\eps})^{-1}(0)[ \rightarrow [0,\rho^{\ast}[$ be the Hugoniot curves as functions of $\theta$ on their domains of definition. 
  	\item The Hugoniot locus tends to the union of the straight lines $\left\{\theta = \theta_{\ell}\right\}$ and $\left\{\rho = \rho^{\ast}\right\}$. 
\end{enumerate}
\label{Prop:Hugoniot_loci}
\end{proposition}
The proof of this proposition is developed in appendix \ref{Appendix:Hugoniot_loci}.

A \textbf{rarefaction wave} is a continuous self-similar solution $(\rho(\frac{x}{t}),\theta(\frac{x}{t}))$. It satisfies the diffential equation
\begin{equation*}
\left(\begin{array}{c} \rho'(s)\\ \eta'(s)\end{array}\right) = \frac{\vec{r}_{\pm}^{\eps}(\rho(s),\eta(s))}{\nabla \lambda_{\pm}^{\eps}(\rho(s),\eta(s)) \cdot \vec{r}_{\pm}^{\eps}(\rho(s),\eta(s))},
\end{equation*}
where $\eta = \Psi(\cos(\theta))$ is the conservative unknown. Therefore, $(\rho,\eta)$  belong to the integral curve of $\vec{r}_{\pm}^{\eps}$. By changing the parametrization of the integral curve, we obtain  
\begin{equation*}
\rho' = \pm \rho |\sin\theta|,\quad \theta' = - \sqrt{\eps p'(\rho)\rho}|\sin\theta|,
\end{equation*} 
and then the following integral equation
\begin{equation}
\theta - \theta_{\ell} = \mp \int_{\rho_{\ell}}^{\rho} \sqrt{\frac{\eps p'(u)}{u}}du.
\label{OndeDetente}
\end{equation} 
It defines two integral curves $\mathcal{O}_{\pm}^{\eps}$ issued from the state $(\rho_{\ell},\theta_{\ell})$. The following proposition summarizes their main properties. 
\begin{proposition}
\begin{enumerate}
  \item The integral curve $\mathcal{O}_{-}^{\eps}$ of $r_{-}^{\eps}$ (resp. $\mathcal{O}_{+}^{\eps}$ of $r_{+}^{\eps}$) is strictly increasing (resp. stricly decreasing) in the $(\rho,\theta)-$plane. Let $i_{-}^{\eps} :\ ](i_{-}^{\eps})^{-1}(0),\pi[ \rightarrow [0,\rho^{\ast}[$ and $i_{+}^{\eps} : \, ]0,$ $(i_{+}^{\eps})^{-1}$ $(0)[ \rightarrow [0,\rho^{\ast}[$ the rarefaction curves as functions of $\theta$ on their domains of definition. 
  \item For all $\gamma \geq 1$, the rarefaction curves tend to the union of the straight lines $\left\{\theta = \theta_{\ell}\right\}$ and $\left\{\rho = \rho^{\ast}\right\}$. Moreover, for $\theta \in ]\theta_{\ell},\pi[$ (resp. $\theta \in ]0,\theta_{\ell}[$), $\rho^{\ast} - i_{-}^{\eps}(\theta) = O(\eps^{\frac{1}{\gamma-1}})$ (resp. $\rho^{\ast} - i_{+}^{\eps}(\theta) = O(\eps^{\frac{1}{\gamma-1}})$). 
  \item Suppose that the state $\rho_{\ell}^{\eps}$ is such that $\rho_{\ell}^{\eps} \rightarrow \rho^{\ast}$ and $\eps p(\rho_{\ell}^{\eps}) \rightarrow \bar{p}_{\ell}$.
	 For all $\rho < \rho_{\ell}^{\eps}$, $(i_{\pm}^{\eps})^{-1}(\rho)$ satisfies:
	\begin{equation*}
	|(i_{\pm}^{\eps})^{-1}(\rho) - \theta_{r}| \leq |(i_{\pm}^{\eps})^{-1}(0) - \theta_{r}| = O(\eps^{\frac{1}{2\gamma}}).
	\end{equation*}
\end{enumerate}
\label{Prop:Integral_curve}
\end{proposition}
The proof is developed in appendix \ref{Appendix:Integral_curve}.

\paragraph{Entropy conditions.} In order to satisfy the Lax entropy condition, each Hugoniot curve $\mathcal{H}_{\pm}^{\eps}$ is restricted to right states which have a smaller associated eigenvalue than the left state. 
\begin{proposition} The eigenvalue $\lambda_{-}^{\eps}$ (resp. $\lambda_{+}^{\eps}$) is a decreasing function of $\rho$ on the Hugoniot curve $\mathcal{H}_{-}^{\eps}$ (resp. an increasing function of $\rho$ on $\mathcal{H}_{+}^{\eps}$) for $\theta < \text{cotan}^{-1}\left((-1/\sqrt{\eps p'(\rho)\rho}\right)$ (resp. $\theta > \text{cotan}^{-1}\left(1/\sqrt{\eps p'(\rho)\rho}\right)$).  
\end{proposition}
\begin{proof} Let $g : \rho \in [0,\rho^{\ast}] \rightarrow g(\rho) \in [0,\pi]$ be an arbitrary function. The variation of $\lambda_{\pm}^{\eps}$ on the graph of $g$ is given by
\begin{equation*}
\nabla \lambda_{\pm}^{\eps} \cdot \left(\begin{array}{c} 1\\ g'(\rho)\end{array}\right) = \pm (\chi^{\eps})'(\rho)\sin\theta + g'(\rho)(-\sin\theta \pm \chi^{\eps}(\rho)\cos\theta).
\end{equation*}
where $\chi^{\eps}(\rho) = \sqrt{\eps p'(\rho)\rho}$. Since $\chi'(\rho)$ is positive and the Hugoniot curve $(h_{-}^{\eps})^{-1}$ is increasing, $\lambda_{-}^{\eps}$ is a decreasing function of $\rho$ on this curve for $\theta \in ]0,\pi[$ such that $(\sin\theta + \chi^{\eps}(\rho)\cos\theta) > 0$. Similarly, since the Hugoniot curve $(h_{+}^{\eps})^{-1}$ is decreasing, $\lambda_{+}^{\eps}$ is a increasing function of $\rho$ on this curve for $\theta \in ]0,\pi[$ such that $(-\sin\theta + \chi^{\eps}(\rho)\cos\theta) > 0$.
\qed
\end{proof}
So, in the limit $\eps \rightarrow 0$, the reachable right states are those belonging to the upper half-domain. We denote by $S_{\pm}^{\eps} =  \mathcal{H}_{\pm}^{\eps}\cap\left\{(\rho,\theta),\ \lambda_{\pm}^{\eps}(\rho,\theta) \leq \lambda_{\pm}^{\eps}(\rho_{\ell},\theta_{\ell})\right\}$ the shock curves. 

Concerning the integral curves $\mathcal{O}_{\pm}^{\eps}$, the admissibility conditions select the curves with increasing eigenvalues and so the curves on the lower half-space. Therefore, the rarefaction curves $R_{\pm}^{\eps}$ satisfy $R_{\pm}^{\eps} \subset \mathcal{O}_{\pm}^{\eps}\cap\left\{(\rho,\theta), \theta \in ]0,\theta_{\ell}[\right\}$. The union of the shock and the rarefaction curves form the forward wave curve $W^{f,\eps}_{\pm} = S_{\pm}^{\eps} \cup R_{\pm}^{\eps}$, while the union of their complementary sets form the backward wave curve  $W^{b,\eps}_{\pm} = \mathcal{H}_{\pm}^{\eps}\backslash S_{\pm}^{\eps} \cup \mathcal{O}_{\pm}^{\eps}\backslash R_{\pm}^{\eps}$.

\begin{figure}
\begin{center}

\null
\hfill
\psfrag{rho}{$\rho$}
\psfrag{theta}{$\theta$}
\subfigure[$\eps = 1$]{\includegraphics[width=0.45\textwidth]{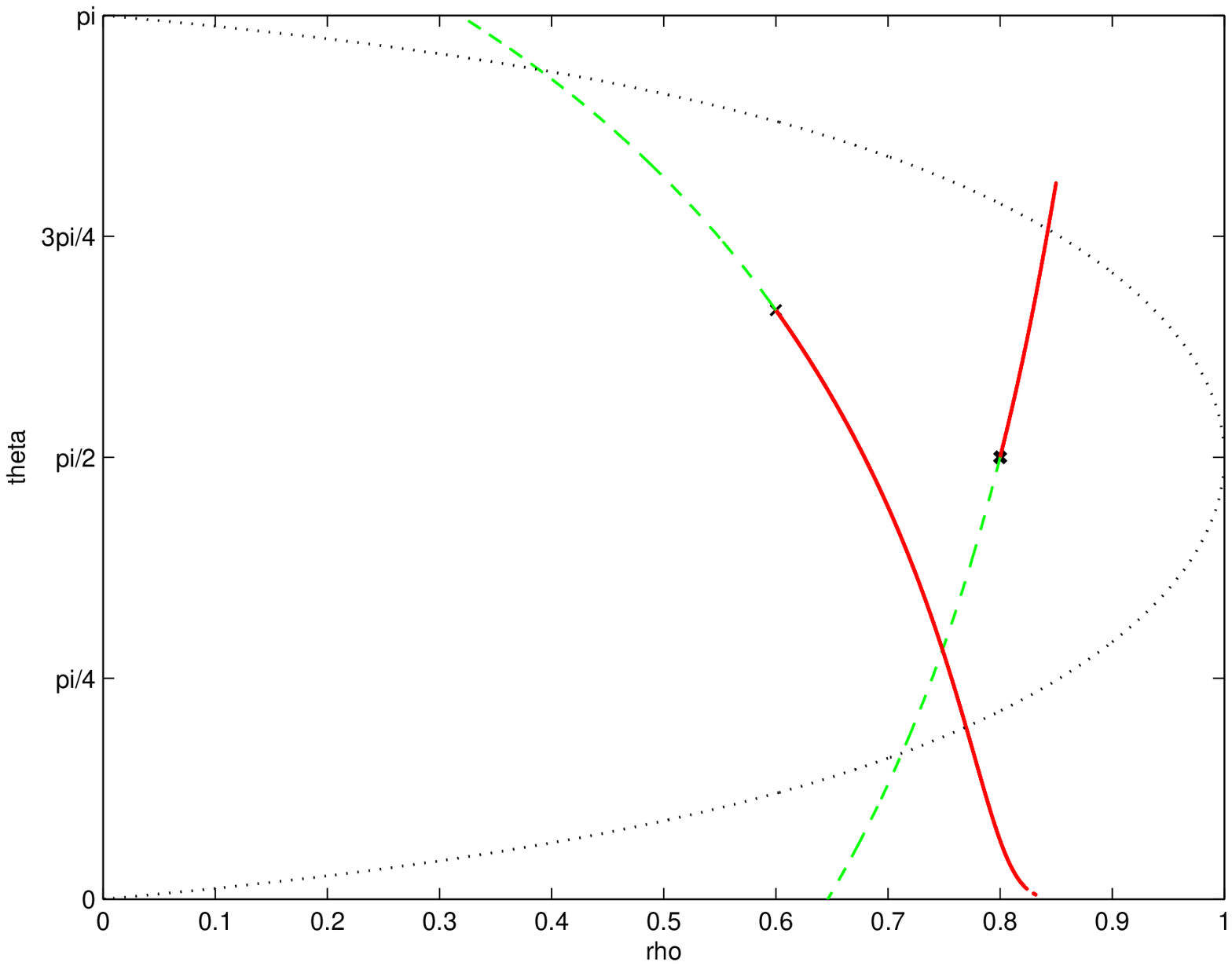}}
\hfill
\subfigure[$\eps = 10^{-1}$]{\includegraphics[width=0.45\textwidth]{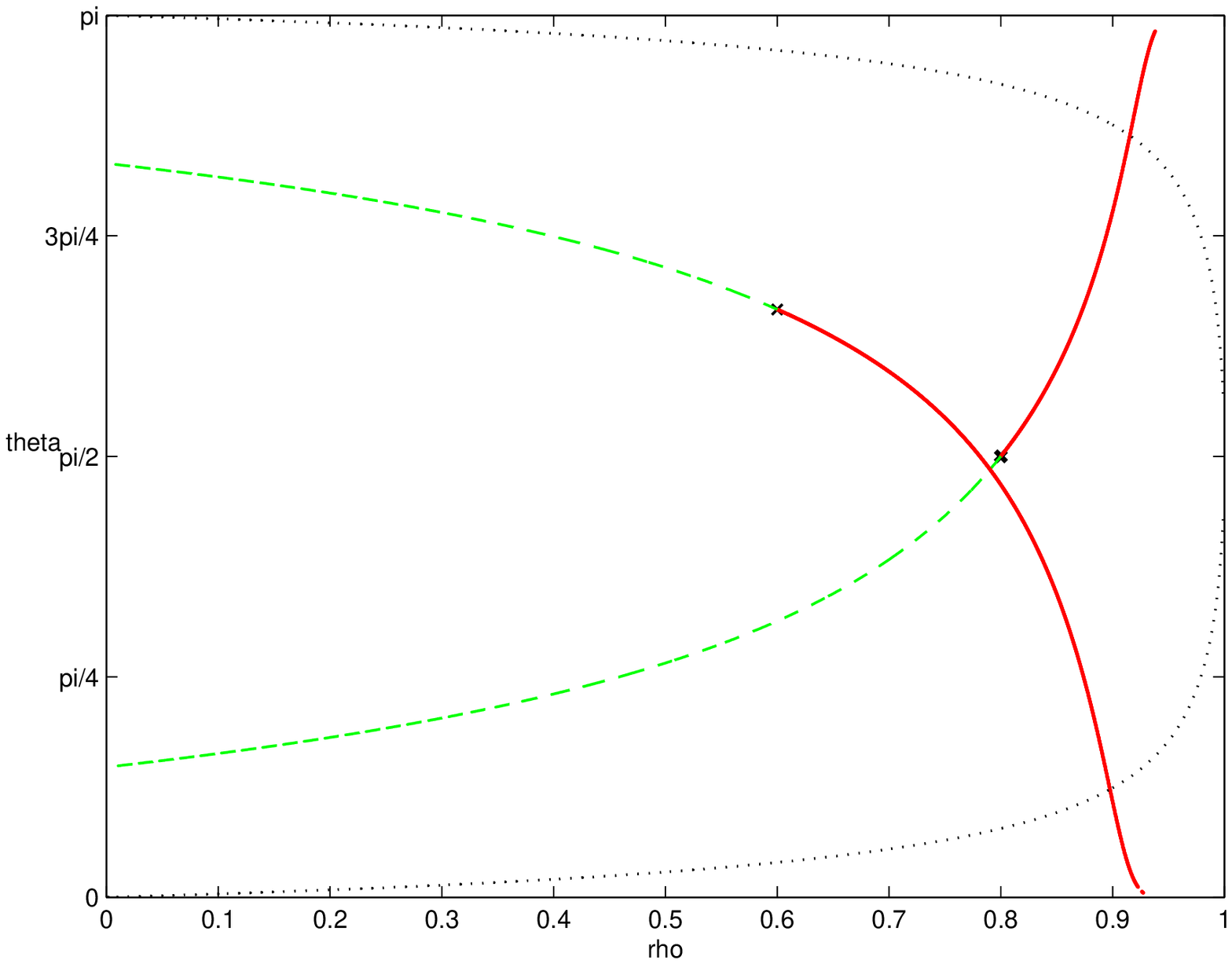}}
\hfill
\null

\null
\hfill
\psfrag{rho}{$\rho$}
\psfrag{theta}{$\theta$}
\subfigure[$\eps = 10^{-2}$]{\includegraphics[width=0.45\textwidth]{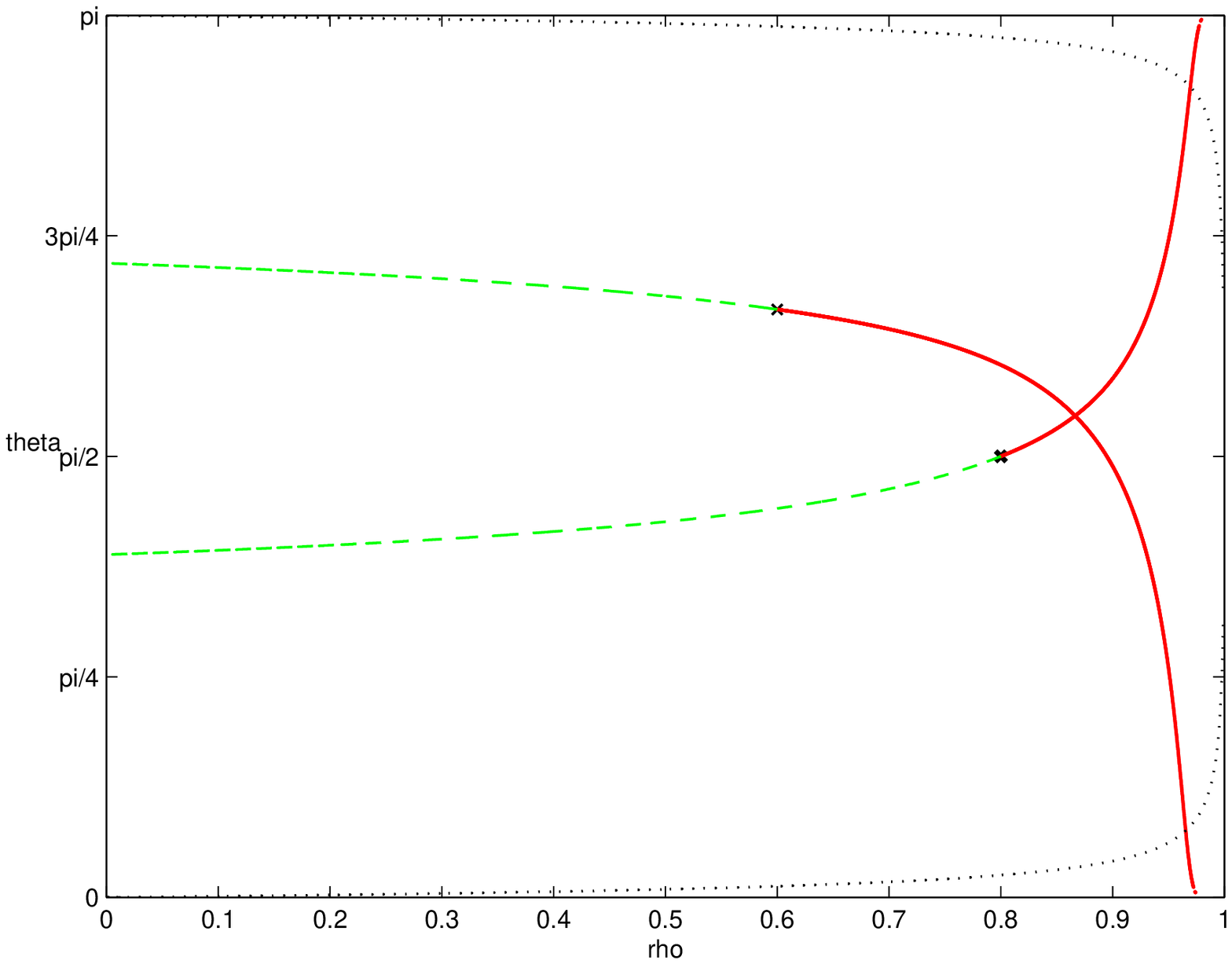}}
\hfill
\psfrag{rho}{$\rho$}
\psfrag{theta}{$\theta$}
\subfigure[$\eps = 10^{-4}$]{\includegraphics[width=0.45\textwidth]{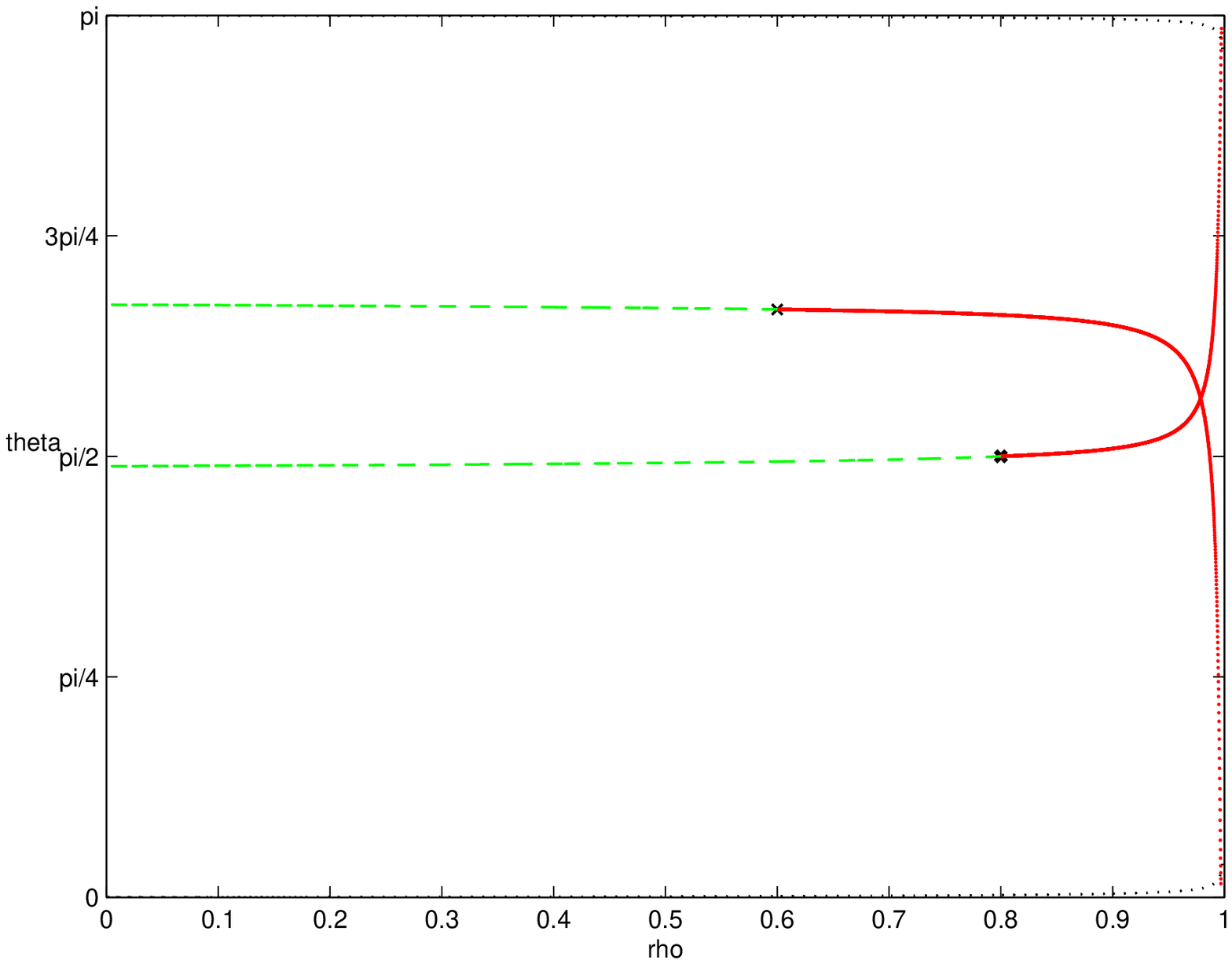}}
\hfill
\null

\caption{Wave curves $W^{f,\eps}_{-}$  for the left state $(\rho_{\ell},\theta_{\ell}) = (0.8,\pi/2)$ and $W^{b,\eps}_{+} $ for the right state $(\rho_{r},\theta_{r}) = (0.6,2\pi/3)$. In dashed green lines: the rarefaction curves. In continuous red lines: the shock curves. In dotted black lines: linearly degenerate sets. $\rho^{\ast} = 1$, $\gamma = 2$.}
\label{Wavecurve}
\end{center}
\end{figure}

\subsubsection{Solutions to the Riemann problem} 

Given a left state $(\rho_{\ell},\theta_{\ell})$ and a right state $(\rho_{r},\theta_{r})$, an entropic solution is found by intersecting the forward 1-wave curve $W^{f,\eps}_{-}$ issued from the left state and the backward 2-wave curve $W^{b,\eps}_{+}$ issued from the right state (cf. fig. \ref{Wavecurve}). In the following study, the curves indexed by - (resp. by +) are implicitly those issued from the left state (resp. from the right state). Because of the monotony of the shock and rarefaction curves, we can classify the different solutions according to the positions of the left and right states in the $(\rho,\theta)$-plane. The following theorem describes the solution of the Riemann problem for small $\eps > 0$ and is illustrated in figure \ref{solriemannGen}. 
\begin{theorem} \label{Thm:Riemann_Small_Epsilon} Considering a left state $(\rho_{\ell},\theta_{\ell})$ and a right state $(\rho_{r},\theta_{r})$, and for $\eps$ small enough, the solution is given by one of the four following cases:
\begin{enumerate}
\item \textbf{Case $\theta_{\ell} = \theta_{r}$.} If $\rho_{\ell} < \rho_{r}$ [resp. $\rho_{\ell} > \rho_{r}$],  the solution consists of a 1-shock [resp. 1-rarefaction] connecting $(\rho_{\ell},\theta_{\ell})$ to  $(\widetilde{\rho},\widetilde{\theta})$ (with $\widetilde{\rho} \in \left]\rho_{\ell},\rho_{r}\right[$ and $\widetilde{\theta} > \theta_{\ell} = \theta_{r}$ [resp. $\widetilde{\rho} \in \left]\rho_{r},\rho_{\ell}\right[$ and $\widetilde{\theta} < \theta_{\ell} = \theta_{r}$]) and then a 2-rarefaction [resp. 2-shock] connecting $(\widetilde{\rho},\widetilde{\theta})$ to $(\rho_{\ell},\theta_{r})$. This is summarized in the following diagram:
 \begin{eqnarray*}
 &&(\rho_{\ell},\theta_{\ell})\quad \stackrel{\mbox{shock}}{\longrightarrow}\quad (\widetilde{\rho},\widetilde{\theta}) \quad\stackrel{\mbox{rarefaction}}{\longrightarrow}\quad (\rho_{r},\theta_{r})\quad\quad \mbox{ if } \rho_{\ell} < \rho_{r} \\
 &&(\rho_{\ell},\theta_{\ell}) \quad\stackrel{\mbox{rarefaction}}{\longrightarrow}\quad (\widetilde{\rho},\widetilde{\theta}) \quad\stackrel{\mbox{shock}}{\longrightarrow}\quad (\rho_{r},\theta_{r})\quad\quad \mbox{ if } \rho_{\ell} > \rho_{r}
 \end{eqnarray*}

	\item \textbf{Case $\theta_{\ell} > \theta_{r}$ ($\cos\theta_{\ell} < \cos\theta_{r}$).} The solution consists of a 1-rarefaction connecting $(\rho_{\ell},\theta_{\ell})$ to  $(0,\widetilde{\theta})$ (with $\widetilde{\rho} < \rho_{\ell}, \rho_{r}$ and $\widetilde{\theta} \in \left]\theta_{r},\theta_{\ell}\right[$) and then a 2-rarefaction wave connecting $(0,\widetilde{\theta})$ to $(\rho_{r},\theta_{r})$. We get the following diagram:
\begin{equation*}
 (\rho_{\ell},\theta_{\ell}) \quad\stackrel{\mbox{rarefaction}}{\longrightarrow}\quad (0,\widetilde{\theta}) \quad\stackrel{\mbox{vacuum}}{\longrightarrow}\quad (0,\widetilde{\widetilde{\theta}}) \quad\stackrel{\mbox{rarefaction}}{\longrightarrow}\quad (\rho_{r},\theta_{r}).
\end{equation*} 

	\item \textbf{Case $\theta_{\ell} < \theta_{r}$ ($\cos\theta_{\ell} > \cos\theta_{r}$).} There are two sub-cases:
\begin{itemize} 
\item if $\rho_{r}^{\eps} < (h_{-}^{\eps})^{-1}(\theta_{r}^{\eps})$ and $\rho_{\ell}^{\eps} < (h_{+}^{\eps})^{-1}(\theta_{\ell}^{\eps})$, the solution consists of a 1-shock connecting $(\rho_{\ell},\theta_{\ell})$ to  $(\widetilde{\rho},\widetilde{\theta})$ (with $\widetilde{\rho} > \rho_{\ell}, \rho_{r}$ and $\widetilde{\theta} \in \left]\theta_{r},\theta_{\ell}\right[$) and then a 2-shock connecting $(\widetilde{\rho},\widetilde{\theta})$ to $(\rho_{r},\theta_{r})$. The diagram is:
\begin{equation*}
(\rho_{\ell},\theta_{\ell}) \quad\stackrel{\mbox{shock}}{\longrightarrow}\quad (\widetilde{\rho},\widetilde{\theta}) \quad\stackrel{\mbox{shock}}{\longrightarrow}\quad (\rho_{r},\theta_{r}) 
 \end{equation*}
 
\item if $\rho_{r}^{\eps} > (h_{-}^{\eps})^{-1}(\theta_{r}^{\eps})$ [resp. $\rho_{\ell}^{\eps} > (h_{+}^{\eps})^{-1}(\theta_{\ell}^{\eps})$], the solution consists of a 1-shock [resp. 1-rarefaction] connecting $(\rho_{\ell},\theta_{\ell})$ to  $(\widetilde{\rho},\widetilde{\theta})$ (with $\widetilde{\rho} \in ]\rho_{\ell},\rho_{r}[$ and $\widetilde{\theta} > \theta_{r}$ [resp. $\widetilde{\rho} \in ]\rho_{r},\rho_{\ell}[$ and $\widetilde{\theta} < \theta_{\ell}$]) and then a 2-rarefaction [resp. 2-shock] connecting $(\widetilde{\rho},\widetilde{\theta})$ to $(\rho_{r},\theta_{r})$. The diagram is as follows:
\begin{eqnarray*}
&&(\rho_{\ell},\theta_{\ell}) \quad\stackrel{\mbox{shock}}{\longrightarrow}\quad (\widetilde{\rho},\widetilde{\theta}) \quad\stackrel{\mbox{rarefaction}}{\longrightarrow}\quad (\rho_{r},\theta_{r})\quad\quad \mbox{ if } \rho_{\ell} < \rho_{r} \\\\ 
&&(\rho_{\ell},\theta_{\ell}) \quad\stackrel{\mbox{rarefaction}}{\longrightarrow}\quad (\widetilde{\rho},\widetilde{\theta}) \quad\stackrel{\mbox{shock}}{\longrightarrow}\quad (\rho_{r},\theta_{r})\quad\quad \mbox{ if } \rho_{\ell} > \rho_{r} \\ 
 \end{eqnarray*}
\end{itemize}
\end{enumerate}
\end{theorem}
The detailed proof of this theorem is developed in appendix \ref{Appendix:Riemann_Small_Epsilon}. Let us provide some ideas of the proof. For finite $\eps$, there exist four kinds of solutions depending on what parts of the curves $W_{-}^{f,\eps}$ and $W_{+}^{b,\eps}$ meet. So, for a fixed left state, the state-space is divided in four subdomains. These subdomains depends on the left state. However, reminding that the limit of the Hugoniot and integral curves are straight lines $\theta = \theta_{\ell}$ or $\rho = \rho^{\ast}$ (cf. propositions \ref{Prop:Hugoniot_loci} and \ref{Prop:Integral_curve}) for all left states, the four subdomains have the same behaviour as $\eps \rightarrow 0$  whatever the left state is.     

\begin{figure}
\begin{center}
\null
\hfill
\psfrag{t}{$t$}
\psfrag{x1}{$x_{1}$}
\psfrag{etat gauche}{$\theta_{\ell}, \rho_{\ell}$}
\psfrag{etat milieu}{\begin{tabular}{l} $\widetilde{\theta} > \theta_{\ell} = \theta_{r}$\\
$\widetilde{\rho} \in \left]\rho_{\ell}, \rho_{r}\right[$ \end{tabular}}
\psfrag{etat droit}{$\theta_{r}, \rho_{r}$}
\subfigure[Case $\theta_{\ell} = \theta_{r}$, $\rho_{\ell} < \rho_{r}$]{\includegraphics[width=0.45\textwidth]{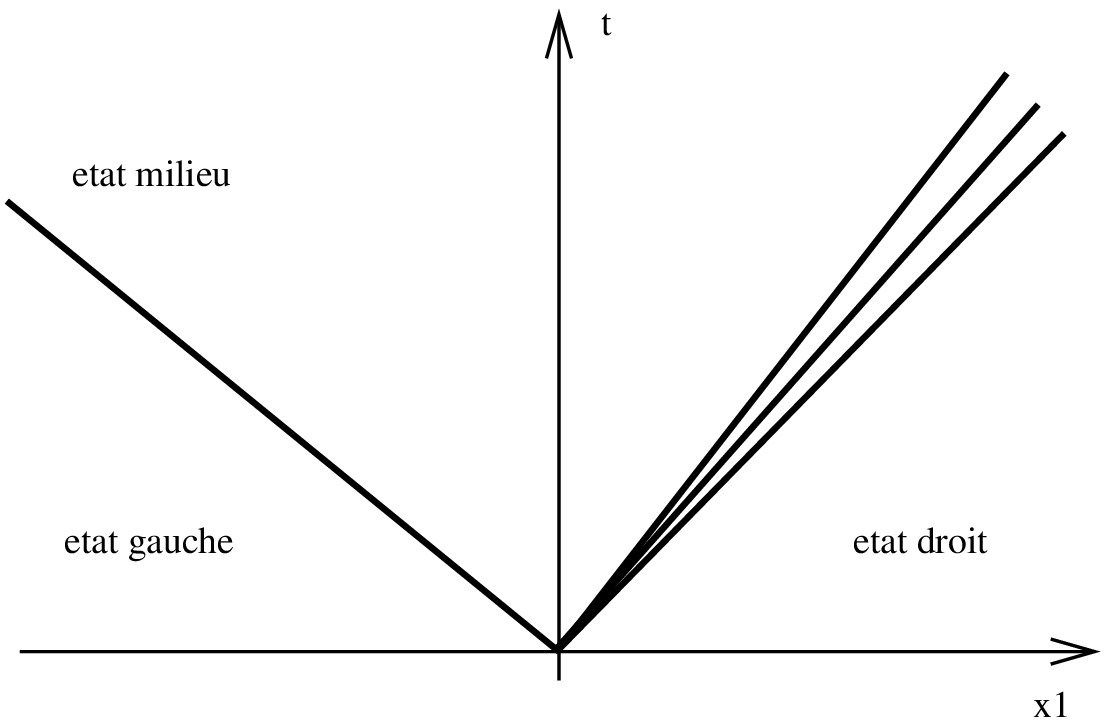}}
\hfill
\psfrag{t}{$t$}
\psfrag{x1}{$x_{1}$}
\psfrag{etat gauche}{$\theta_{\ell}, \rho_{\ell}$}
\psfrag{etat milieu}{\begin{tabular}{l} $\widetilde{\theta} < \theta_{\ell} = \theta_{r}$\\
$\widetilde{\rho} \in \left]\rho_{r}, \rho_{\ell}\right[$ \end{tabular}}
\psfrag{etat droit}{$\theta_{r}, \rho_{r}$}
\subfigure[Case $\theta_{\ell} = \theta_{r}$, $\rho_{\ell} > \rho_{r}$]{\includegraphics[width=0.45\textwidth]{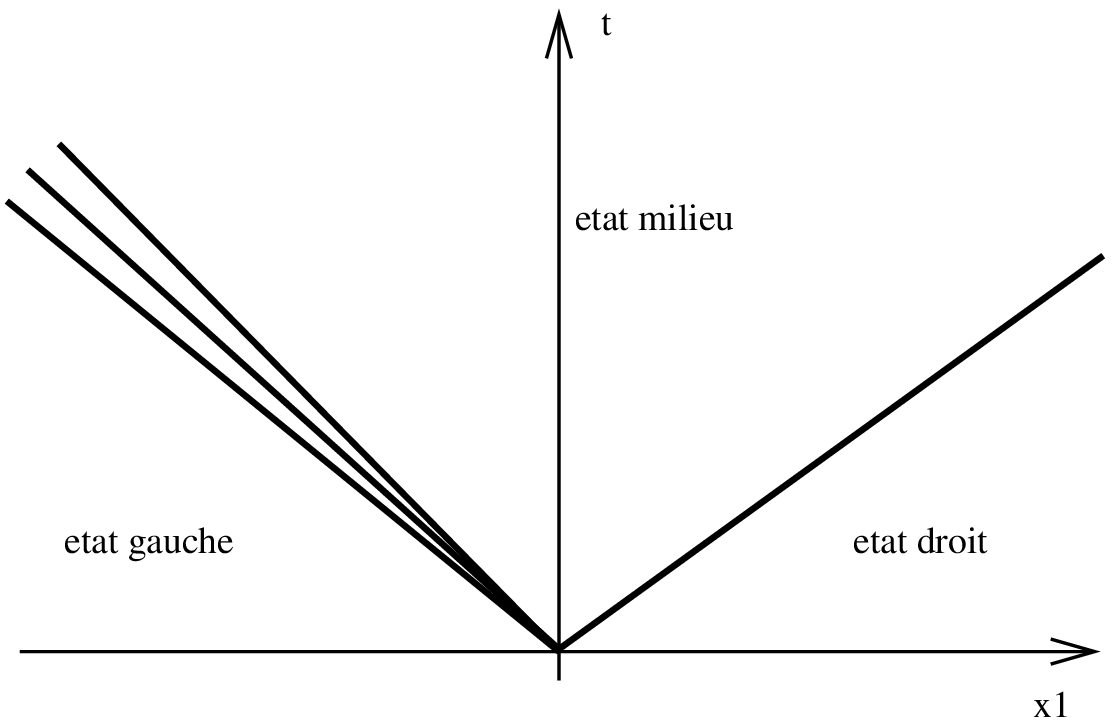}}
\hfill
\null

\null
\hfill
\psfrag{t}{$t$}
\psfrag{x1}{$x_{1}$}
\psfrag{etat gauche}{$\theta_{\ell}, \rho_{\ell}$}
\psfrag{etat milieu}{Vacuum}
\psfrag{etat droit}{$\theta_{r}, \rho_{r}$}
\subfigure[Case $\theta_{\ell} > \theta_{r}$]{\includegraphics[width=0.45\textwidth]{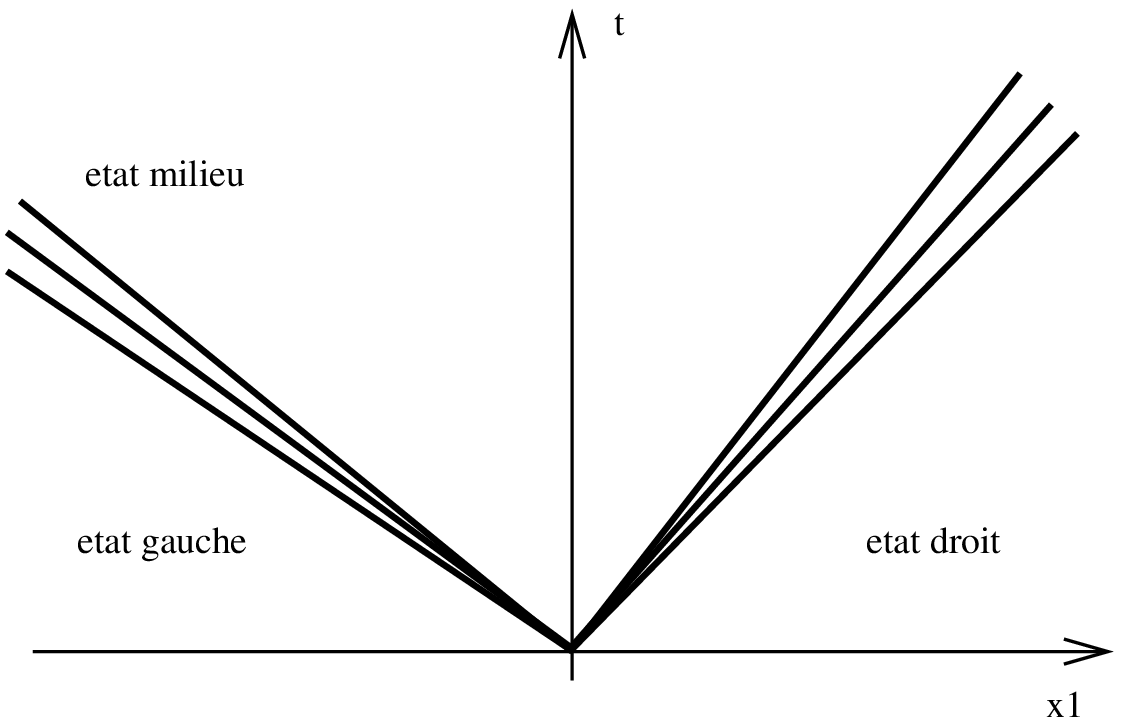}}
\hfill
\psfrag{t}{$t$}
\psfrag{x1}{$x_{1}$}
\psfrag{etat gauche}{$\theta_{\ell}, \rho_{\ell}$}
\psfrag{etat milieu}{\begin{tabular}{l} $\widetilde{\theta} \in \left]\theta_{\ell}, \theta_{r}\right[$\\
$\widetilde{\rho} > \rho_{\ell}, \rho_{r}$ \end{tabular}}
\psfrag{etat droit}{$\theta_{r}, \rho_{r}$}
\subfigure[Case $\theta_{\ell} < \theta_{r}$, first subcase]{\includegraphics[width=0.45\textwidth]{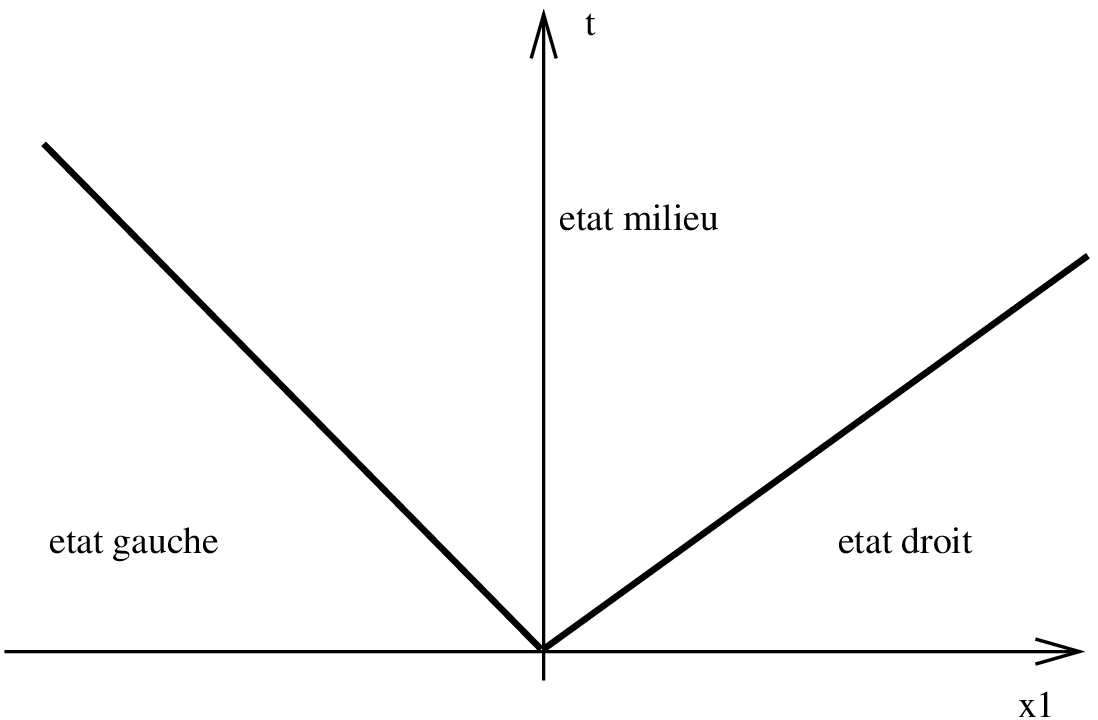}}
\hfill
\null

\null
\hfill
\psfrag{t}{$t$}
\psfrag{x1}{$x_{1}$}
\psfrag{etat gauche}{$\theta_{\ell}, \rho_{\ell}$}
\psfrag{etat milieu}{\begin{tabular}{l} $\widetilde{\theta} > \theta_{\ell}, \theta_{r}$\\
$\widetilde{\rho} \in \left]\rho_{\ell}, \rho_{r}\right[$ \end{tabular}}
\psfrag{etat droit}{$\theta_{r}, \rho_{r}$}
\subfigure[Case $\theta_{\ell} < \theta_{r}$, second subcase, $\rho_{\ell} < \rho_{r}$]{\includegraphics[width=0.45\textwidth]{solriemann.eps}}
\hfill
\psfrag{t}{$t$}
\psfrag{x1}{$x_{1}$}
\psfrag{etat gauche}{$\theta_{\ell}, \rho_{\ell}$}
\psfrag{etat milieu}{\begin{tabular}{l} $\widetilde{\theta} < \theta_{\ell},\theta_{r}$\\
$\widetilde{\rho} \in \left]\rho_{r}, \rho_{\ell}\right[$ \end{tabular}}
\psfrag{etat droit}{$\theta_{r}, \rho_{r}$}
\subfigure[Case $\theta_{\ell} < \theta_{r}$, second subcase, $\rho_{\ell} > \rho_{r}$]{\includegraphics[width=0.45\textwidth]{solriemann2.eps}}
\hfill
\null

\caption{Solutions to the Riemann problem for small $\eps > 0$.}
\label{solriemannGen}
\end{center}
\end{figure}

\subsubsection{The sign of $\theta$}
\label{Sec:Full_Riemann}

The conservative system (\ref{Eq:rho_1D_cons})-(\ref{Eq:theta_1D_cons}) does not determine the sign of $\theta$ (if $\theta$ is supposed to be in $]-\pi,\pi[$). As mentioned above, this is not important since our main goal is to provide connection conditions on $\bar{p}$ between the left and right states. However, it is desirable to determine it in the present analysis, for the sake of completeness. For this goal, we cannot use (\ref{Eq:theta_1D_cons}) because $\Psi(\cos\theta)$ is an even function of $\theta$. Again, we are facing an indetermination due to the non-conservative character of the system. One possible solution is to introduce a contact discontinuity from $\theta$ to $-\theta$ with propagation speed $\cos\theta$ in the domains where $\rho$ is constant and $\cos\theta$ is continuous. If we add such a contact wave, there is only one possible construction given by the following: 
    
\begin{proposition} Suppose that $\theta_{\ell}$, $\theta_{r} \in [-\pi,\pi]$ and $\theta_{\ell}$, $\theta_{r}$ have different signs. 
\begin{enumerate} 
\item In the subcases $\cos\theta_{\ell} = \cos\theta_{r}$ and $\cos\theta_{\ell} > \cos\theta_{r}$ of theorem \ref{Thm:Riemann_Small_Epsilon}, the only one possible contact wave in the domains of constant $\rho$ and continuous $\cos\theta$ is located inside the intermediate state and the propogation speed equals $\cos\tilde{\theta}$. 
\item In the subcase $\cos\theta_{\ell} < \cos\theta_{r}$, the possible contact waves are those located in the vacuum domain. There is no uniqueness of the propagation speed but since this contact discontinuity occurs in the vacuum $\rho = 0$ region, we may consider that $\theta$ is not defined in this region. 
\end{enumerate}
\label{Prop:Riemann_sign}
\end{proposition}    
The proof of this proposition can be found in appendix \ref{Appendix:Riemann_sign}. Two cases of the Riemann problem with $\theta_{r} < 0 < \theta_{\ell}$ are represented in Fig. \ref{solriemannGen_complet}. Note that the position of the contact wave does not depend on $\eps$. So their limits as $\eps$ goes to zero are easily obtained. 

\begin{figure}
\begin{center}
\null
\hfill
\psfrag{t}{$t$}
\psfrag{x1}{$x_{1}$}
\psfrag{etat gauche}{$\theta_{\ell}, \rho_{\ell}$}
\psfrag{etat milieu}{Vacuum}
\psfrag{etat droit}{$\theta_{r}, \rho_{r}$}
\psfrag{th}{$\tilde{\theta}$}
\psfrag{-th}{$-\tilde{\theta}$}
\subfigure[Case $\theta_{\ell} > |\theta_{r}|$]{\includegraphics[width=0.45\textwidth]{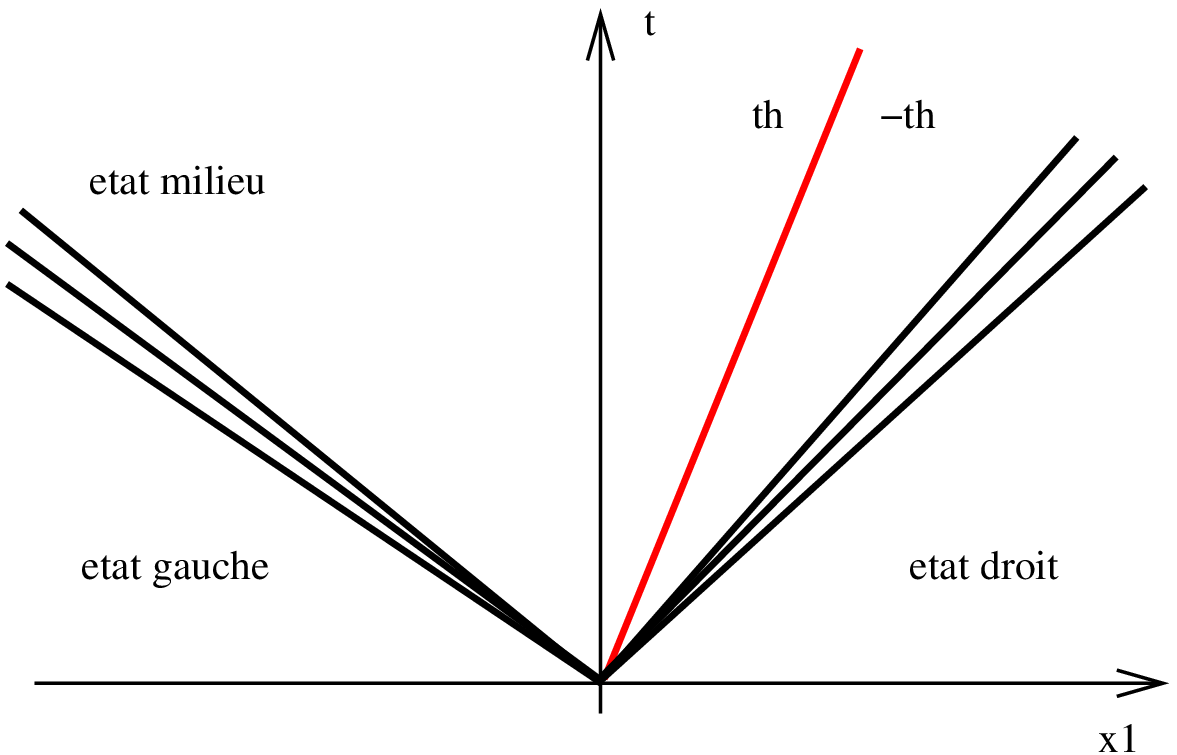}}
\hfill
\psfrag{t}{$t$}
\psfrag{x1}{$x_{1}$}
\psfrag{etat gauche}{$\theta_{\ell}, \rho_{\ell}$}
\psfrag{etat milieu}{\begin{tabular}{l} $\widetilde{\theta} \in \left]\theta_{\ell}, |\theta_{r}|\right[$\\
$\widetilde{\rho} > \rho_{\ell}, \rho_{r}$ \end{tabular}}
\psfrag{etat milieu bis}{\begin{tabular}{l} $-\widetilde{\theta}$\\
$\widetilde{\rho}$ \end{tabular}}
\psfrag{etat droit}{$\theta_{r}, \rho_{r}$}
\psfrag{costh}{$\cos\tilde{\theta}$}
\subfigure[Case $\theta_{\ell} < |\theta_{r}|$, first subcase]{\includegraphics[width=0.45\textwidth]{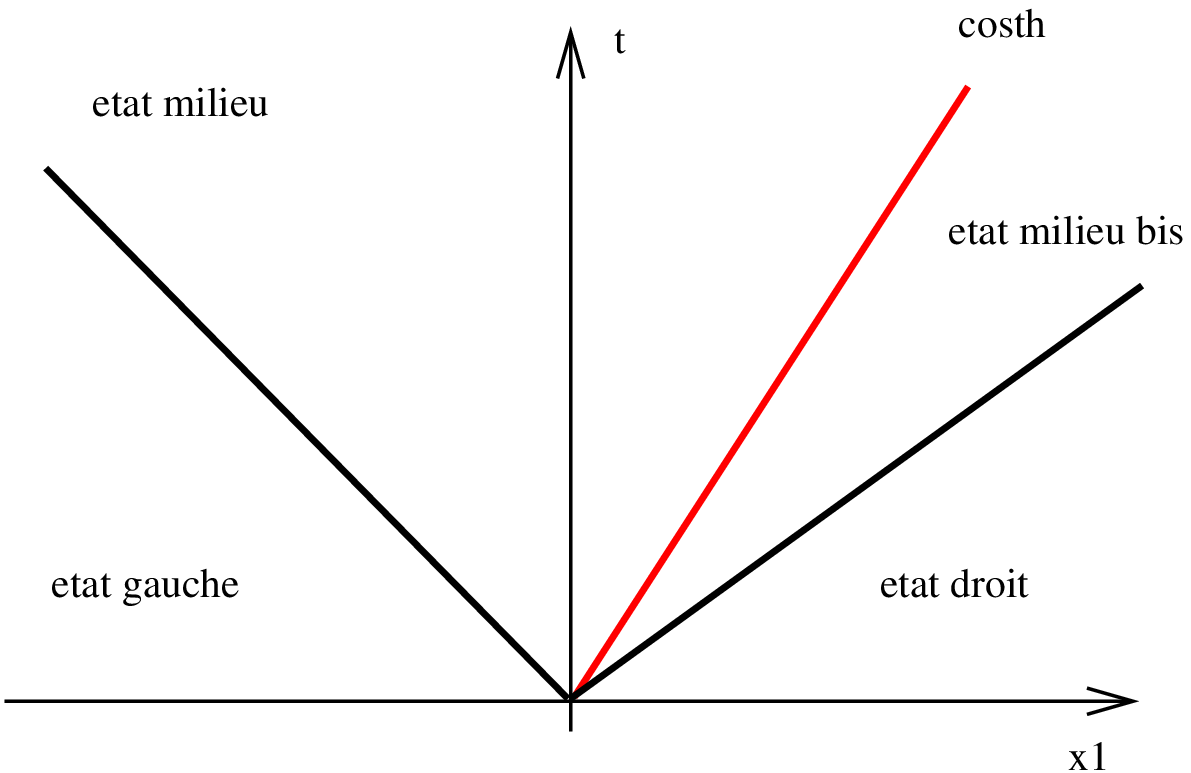}}
\hfill
\null

\caption{Some solutions to the Riemann problem for small $\eps > 0$ and $ -\pi < \theta_{r} < 0 < \theta_{\ell} < \pi$.}
\label{solriemannGen_complet}
\end{center}
\end{figure}

\subsection{The solutions of the Riemann problem in the limit $\eps \rightarrow 0$}

In order to study the limit $\eps \rightarrow 0$, we introduce converging sequences of left and right states
\begin{equation*}
\left((\rho_{\ell}^{\eps},\theta_{\ell}^{\eps}), (\rho_{r}^{\eps},\theta_{r}^{\eps})\right)\quad \underset{\eps \rightarrow 0}{\longrightarrow}\quad ((\rho_{\ell},\theta_{\ell}),(\rho_{r},\theta_{r}))
\end{equation*}
and we look for the limits of the solutions of the associated Riemann problems. There are three cases to consider: either none of the two states tends to the congested state ($\rho_{\ell},\rho_{r} < \rho^{\ast}$), or one of the two does ($\rho_{\ell} < \rho^{\ast},\rho_{r}^{\eps} \rightarrow \rho^{\ast}$) or both of them do ($\rho_{\ell}^{\eps},\rho_{r}^{\eps} \rightarrow \rho^{\ast}$). The case ($\rho_{\ell}^{\eps} \rightarrow \rho^{\ast},\rho_{r} < \rho^{\ast}$) is obtained by symmetry from the case ($\rho_{\ell} < \rho^{\ast},\rho_{r}^{\eps} \rightarrow \rho^{\ast}$): the left and right quantities have to be exchanged and the arrows have to be flipped (like in the first case of theorem \ref{Thm:Riemann_Small_Epsilon}). Since the solutions of the Riemann problem are bounded and monotonous, all the sequences belong to a bounded subset of $BV(\R)$ and consequently, to a compact subset of $L^{1}_{\text{loc}}(\R)$. So we only need to prove the uniqueness of the limit of converging sequences to prove the convergence of the whole sequence and we can consider that the convergence is in the almost everywhere sense (up to the extraction of a subsequence). 

As a guideline, we mention that, compared with the system with finite $\eps$, the limit Riemann problem has two additional properties: the appearance of clusters which corresponds to the saturation of the constraint $\rho \leq \rho^{\ast}$ and the disappearance of rarefaction waves and their transformations into contact waves. In the subsequent statements, the term "limit" is a short-hand for "limit of the solution to the Riemann problem of (\ref{Eq:rho_1D_cons})-(\ref{Eq:theta_1D_cons})" as $\eps \rightarrow 0$.

\subsubsection{Case $\rho_{\ell} < \rho^{\ast},\rho_{r} < \rho^{\ast}$ (see fig. \ref{Fig:limit_Riemann_1})}

 \begin{proposition}\label{Prop:limit_Riemann_1} \textbf{(Case $\rho_{\ell} < \rho^{\ast},\rho_{r} < \rho^{\ast}$)} There are only three cases:
\begin{enumerate}[(a)]
	\item \textbf{Subcase $\theta_{\ell} = \theta_{r}$.} The limit consists 
of only one contact wave connecting  $(\rho_{\ell},\theta_{\ell})$ to $(\rho_{r},\theta_{\ell})$: 
\begin{equation*}
(\rho_{\ell},\theta_{\ell})\quad\stackrel{\mbox{contact}}{\longrightarrow}\quad(\rho_{r},\theta_{r}).
\end{equation*}
The travelling speed is equal to $\cos\theta_{\ell}$.
  \item \textbf{Subcase $\theta_{\ell} > \theta_{r}$.} The limit consists of two contact waves connecting the two states to a vacuum state:
\begin{equation*}
(\rho_{\ell},\theta_{\ell}) \quad\stackrel{\mbox{contact}}{\longrightarrow}\quad(0,\theta_{\ell})\quad\stackrel{\mbox{Vacuum}}{\longrightarrow}\quad(0,\theta_{r})\quad\stackrel{\mbox{contact}}{\longrightarrow}\quad(\rho_{r},\theta_{r}).
\end{equation*} 
The travelling speeds are respectively equal to $\cos\theta_{\ell}$ and $\cos\theta_{r}$.
  \item \textbf{Subcase $\theta_{\ell} < \theta_{r}$.} The limit consists 
of two shocks connecting the left state $(\rho_{\ell},\theta_{\ell})$ to  a congested state $(\rho^{\ast},\widetilde{\theta},\bar{p})$ and then connecting $(\rho^{\ast},\widetilde{\theta},\bar{p})$ to the right state $(\rho_{r},\theta_{r})$:
\begin{equation*}
(\rho_{\ell},\theta_{\ell})\quad\stackrel{\mbox{shock}}{\longrightarrow}\quad(\rho^{\ast},\widetilde{\theta},\bar{p})\quad\stackrel{\mbox{shock}}{\longrightarrow}\quad(\rho_{r},\theta_{r}). 
\end{equation*}
where $\widetilde{\theta}$ is the unique solution of 
\begin{eqnarray*}
&&\left[\Psi(\cos(\theta))\right]_{r}\frac{\left[\rho\cos(\theta)\right]_{r}}{\left[\rho\right]_{r}} - \left[\Psi(\cos(\theta))\right]_{\ell}\frac{\left[\rho\cos(\theta)\right]_{\ell}}{\left[\rho\right]_{\ell}} = \left[\Phi(\cos(\theta))\right]^{\ell}_{r},\\
&&\theta \in \left[\min(\theta_{\ell},\theta_{r}),\max(\theta_{\ell},\theta_{r})\right],
\end{eqnarray*} 
and $\bar{p}$ is given by
\begin{equation*}
\bar{p} =  \frac{\left[\Psi(\cos(\theta))\right]_{\ell}\left[\rho\cos(\theta)\right]_{\ell}}{\left[\rho\right]_{\ell}}  - \left[\Phi(\cos(\theta))\right]_{\ell} =  \frac{\left[\Psi(\cos(\theta))\right]_{r}\left[\rho\cos(\theta)\right]_{r}}{\left[\rho\right]_{r}}  - \left[\Phi(\cos(\theta))\right]_{r}.
\end{equation*}
The shock speeds are given by the Rankine-Hugoniot condition for the density (\ref{Eq:RH_1}).
\end{enumerate}
\end{proposition}
Note that  in clustered region, since $\rho = \rho^{\ast}$, the state is determined by the values of $\theta$ and $\bar{p}$. This is why we add a third component giving the value of $\bar{p}$ to the vector defining the state in the clustered region. 

In this proposition, the quantities $[f]_{\ell} := \widetilde{f} - f_{\ell}$, $[f]_{r} := \widetilde{f} - f_{r}$ denote the difference between the intermediate value and the left (or right) value of the quantity $f$ and $[f]_{\ell}^{r} := f_{r} - f_{\ell}$ denotes the difference between the right and left values of the quantity $f$. 

This proposition covers several kinds of interfaces described in Formal Statement~\ref{Formal_Statement}: the case (a) is an occurence of an interface (UC)-(UC), the case (b) of an interface (UC)-(V) and the case (c) of an interface (C)-(UC). Moreover, the proposition implies that $\theta$ is continuous inside (UC) domains.  

Note that all the intermediate states are explicitly given or are solutions of a non-linear equation and so, are explicitly computable. The proof of this proposition is given in appendix \ref{Appendix:limit_Riemann_1}.

\begin{figure}
\begin{center}
\null
\hfill
\psfrag{t}{$t$}
\psfrag{x1}{$x_{1}$}
\psfrag{etat gauche}{$\theta_{\ell}, \rho_{\ell}$}
\psfrag{vitesse}{{\footnotesize $\cos\theta_{\ell}$}}
\psfrag{etat droit}{$\theta_{\ell}, \rho_{r}$}
\subfigure[Cas $\theta_{\ell} = \theta_{r}$]{\includegraphics[width=0.3\textwidth]{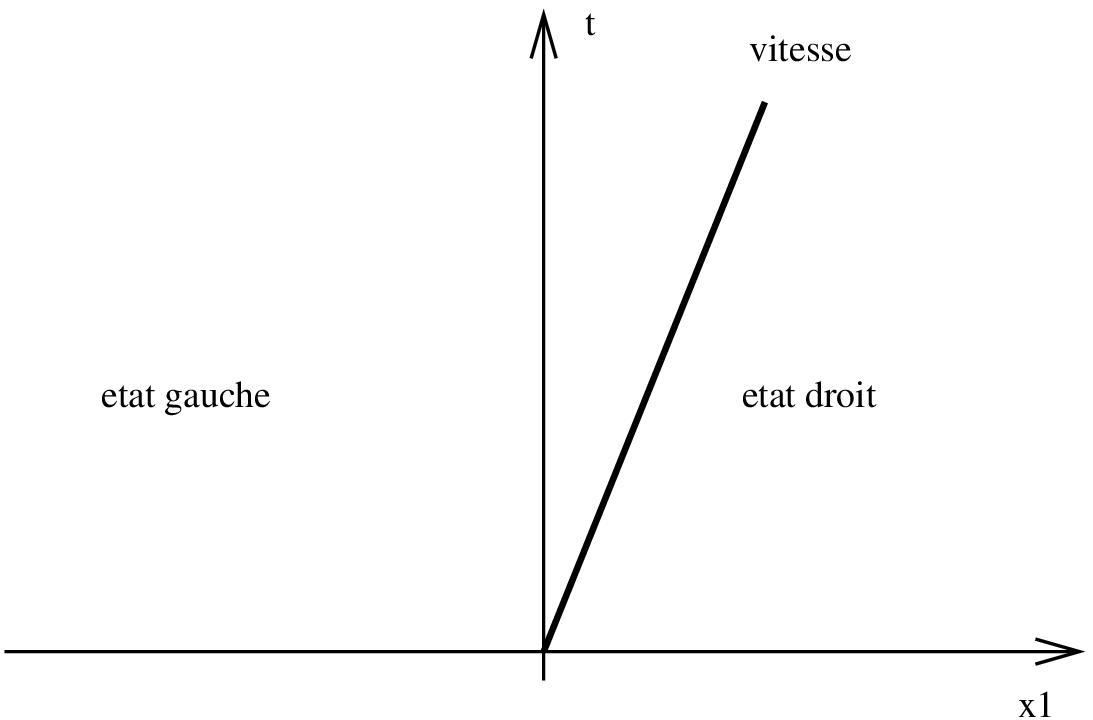}}
\hfill
\psfrag{t}{$t$}
\psfrag{x1}{$x_{1}$}
\psfrag{etat gauche}{$\theta_{\ell}, \rho_{\ell}$}
\psfrag{etat milieu}{Vacuum}
\psfrag{etat droit}{$\theta_{r}, \rho_{r}$}
\psfrag{vitesse1}{{\footnotesize $\cos\theta_{\ell}$}}
\psfrag{vitesse2}{{\footnotesize $\cos\theta_{r}$}}
\subfigure[Cas $\theta_{\ell} > \theta_{r}$]{\includegraphics[width=0.3\textwidth]{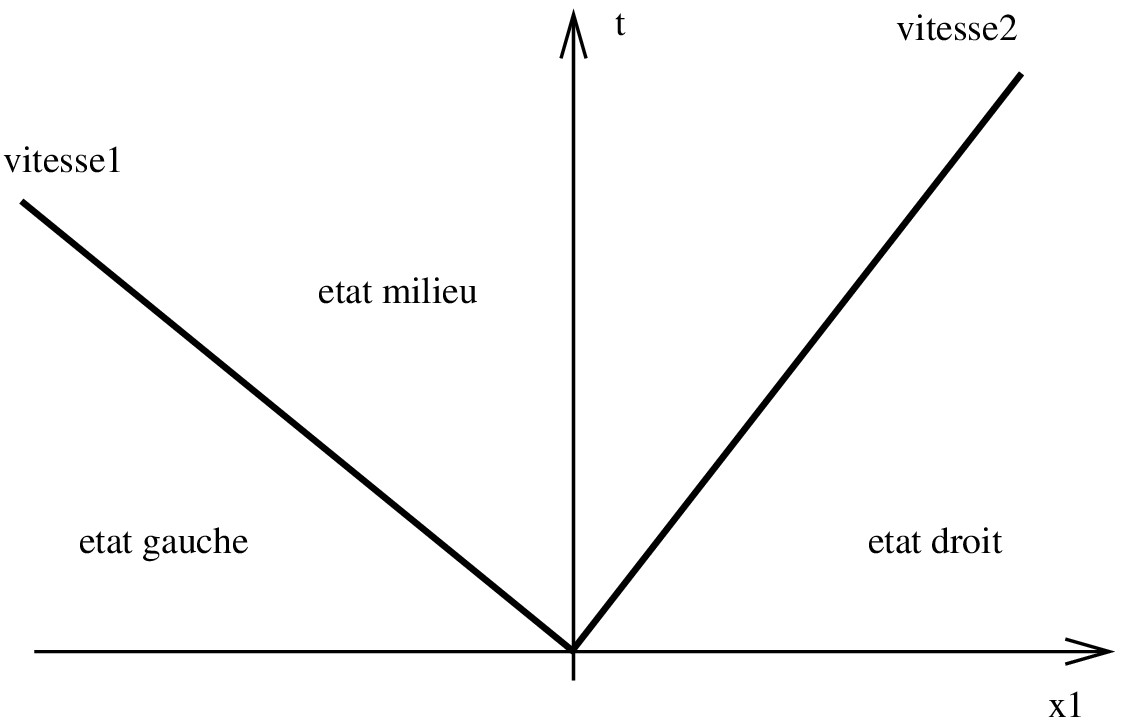}}
\hfill
\psfrag{t}{$t$}
\psfrag{x1}{$x_{1}$}
\psfrag{etat gauche}{$\theta_{\ell}, \rho_{\ell}$}
\psfrag{etat milieu}{$\widetilde{\theta},\rho^{\ast},\bar{p}$}
\psfrag{etat droit}{$\theta_{r}, \rho_{r}$}
\psfrag{vitesse1}{}
\psfrag{vitesse2}{}
\subfigure[Cas $\theta_{\ell} < \theta_{r}$]{\includegraphics[width=0.3\textwidth]{solriemannlim2.eps}}
\hfill
\null
\caption{Limit solutions of the Riemann problem for $\eps = 0$ and $\rho_{\ell},\rho_{r} < \rho^{\ast}$.}
\label{Fig:limit_Riemann_1}
\end{center}
\end{figure}
\begin{figure}
\begin{center}
\end{center}
\end{figure}

\subsubsection{Case $\rho_{\ell} < \rho^{\ast},\rho_{r} = \rho^{\ast}$ (see fig. \ref{Fig:limit_Riemann_2})} 

This case is typical of the situation at a cluster boundary. In this case, the main new feature is the appearance of declustering waves as limits of the rarefaction waves. These declustering waves are instantaneous cancellations of the pressure. The following lemma details this statement:  
\begin{lemma} (Limit of rarefaction waves, declustering wave) Let $(\rho_{r}^{\eps},\theta_{r})$ be a sequence of right states such that $\rho_{r}^{\eps} \rightarrow \rho^{\ast}$ and $\eps p(\rho_{r}^{\eps}) \rightarrow \bar{p}_{r} > 0$. Introduce a converging sequence of states $(\widetilde{\rho}^{\eps},\widetilde{\theta}^{\eps})$ lying on the rarefaction curves issued from the right states, such that $\widetilde{\rho}^{\eps} < \rho_{r}^{\eps}$. 

If $\widetilde{\rho} = \lim \widetilde{\rho}^{\eps} < \rho^{\ast}$, then the rarefaction wave tends to the combination of a contact wave between the state $(\widetilde{\rho},\theta_{r})$ and $(\rho^{\ast},\theta_{r},\bar{p}_{r})$ with speed $\cos\theta_{r} = \lambda_{+}$ and a declustering wave, i.e. a contact wave with infinite speed which cancels the pressure, which provides a transition between $(\rho^{\ast},\theta_{r},\bar{p}_{r})$ and $(\rho^{\ast},\theta_{r},0)$. 

If $\widetilde{\rho} = \lim \widetilde{\rho}^{\eps} = \rho^{\ast}$ then the rarefaction wave tends to a shock wave with infinite speed between the states $(\rho^{\ast},\theta_{r},\bar{\bar{p}})$ and $(\rho^{\ast},\theta_{r},\bar{p}_{r})$, where $\bar{\bar{p}} = \lim \eps p(\widetilde{\rho}^{\eps})$.   
\label{Lemma:Rarefaction_wave_2}
\end{lemma}
The proof of this lemma is developed in appendix \ref{Appendix:Rarefaction_wave_2}.

The next proposition provides the solutions of the limit Riemann problem and figure \ref{Fig:limit_Riemann_2} schematically describes them.
\begin{proposition} \textbf{(Case $\rho_{\ell} < \rho^{\ast},\rho_{r} = \rho^{\ast}$)} There are only three cases:
\begin{enumerate}[(a)]
	\item \textbf{Subcase $\theta_{\ell} = \theta_{r}$.} The limit solution consists of one contact wave connecting the left state $(\rho_{\ell},\theta_{\ell})$ to an intermediate congested state $(\rho^{\ast},\theta_{r},\bar{p} = 0)$ and then a cluster contact (with infinite speed):
\begin{equation*}
 (\rho_{\ell},\theta_{\ell}) \stackrel{\mbox{contact}}{\longrightarrow} (\rho^{\ast},\theta_{r},0) \stackrel{\mbox{declust.}}{\rightarrow} (\rho^{\ast},\theta_{r},\bar{p}).
\end{equation*}
  \item \textbf{Subcase $\theta_{\ell} > \theta_{r}$.} The limit solution consists of one contact wave connecting the left state to vacuum, and then another contact wave connecting the vacuum to a congested and pressureless state $(\rho^{\ast},\theta_{r},0)$ and finally a cluster contact connecting $(\rho^{\ast},\theta_{r},0)$ to $(\rho^{\ast},\theta_{r},\bar{p})$:
\begin{equation*}
 (\rho_{\ell},\theta_{\ell}) \stackrel{\mbox{contact}}{\longrightarrow} (0,\theta_{\ell}) \stackrel{\mbox{vacuum}}{\longrightarrow} (0,\theta_{r}) \stackrel{\mbox{contact}}{\longrightarrow} (\rho^{\ast},\theta_{r},0) \stackrel{\mbox{declust.}}{\rightarrow} (\rho^{\ast},\theta_{r},\bar{p}).
\end{equation*}
  \item \textbf{Subcase $\theta_{\ell} < \theta_{r}$.} The limit solution consists of one shock wave connecting the left state $(\rho_{\ell},\theta_{\ell})$ to an intermediate congested state $(\rho^{\ast},\theta_{r},\bar{\bar{p}})$ and one contact wave with infinite propagation speed connecting $(\rho^{\ast},\theta_{r},\bar{\bar{p}})$ to the right state $(\rho^{\ast},\theta_{r},\bar{p})$:
\begin{equation*}
 (\rho_{\ell},\theta_{\ell}) \stackrel{\mbox{shock}}{\longrightarrow} (\rho^{\ast},\theta_{\ell},\bar{\bar{p}}) \stackrel{\mbox{contact}}{\longrightarrow} (\rho^{\ast},\theta_{r},\bar{p}),
\end{equation*}  
where the intermediate pressure $\bar{\bar{p}}$ is equal to
\begin{equation*}
\bar{\bar{p}} = \left[\Psi(\cos\theta)\right]_{\ell}\frac{\left[\rho\cos\theta\right]_{\ell}}{\left[\rho\right]_{\ell}} - \left[\Phi(\cos\theta)\right]_{\ell}.
\end{equation*}
The shock speed is given by the Rankine-Hugoniot condition for the density (\ref{Eq:RH_1}).
\end{enumerate}
\label{Prop:limit_Riemann_2}
\end{proposition}
In practice, when instantaneous waves occur (i.e. with infinite propagation speed), it means that the initial data of the Riemann problem does not spontaneously appear during the dynamical evolution of the limit problem. They have to be ignored.

Like proposition \ref{Prop:limit_Riemann_1}, this new proposition covers several kinds of interfaces described in the Formal Statement \ref{Formal_Statement}: cases (a) and (c) are occurences of interfaces (C)-(UC), the left wave of case (b)  is an occurence of an interface (UC)-(V) whereas the right wave of the case (b) is an interface (C)-(V).

The proof of this proposition is in appendix \ref{Appendix:limit_Riemann_2}.
\begin{figure}
\begin{center}
\null
\hfill
\psfrag{t}{$t$}
\psfrag{x1}{$x_{1}$}
\psfrag{etat gauche}{$\theta_{\ell}, \rho_{\ell}$}
\psfrag{vitesse1}{{\footnotesize $\cos\theta_{\ell}$}}
\psfrag{etat droit2}{$\theta_{r},\rho^{\ast},0$}
\psfrag{etat droit}{$\theta_{r}, \rho^{\ast},\bar{p}$}
\subfigure[Case $\theta_{\ell} = \theta_{r}$]{\includegraphics[width=0.3\textwidth]{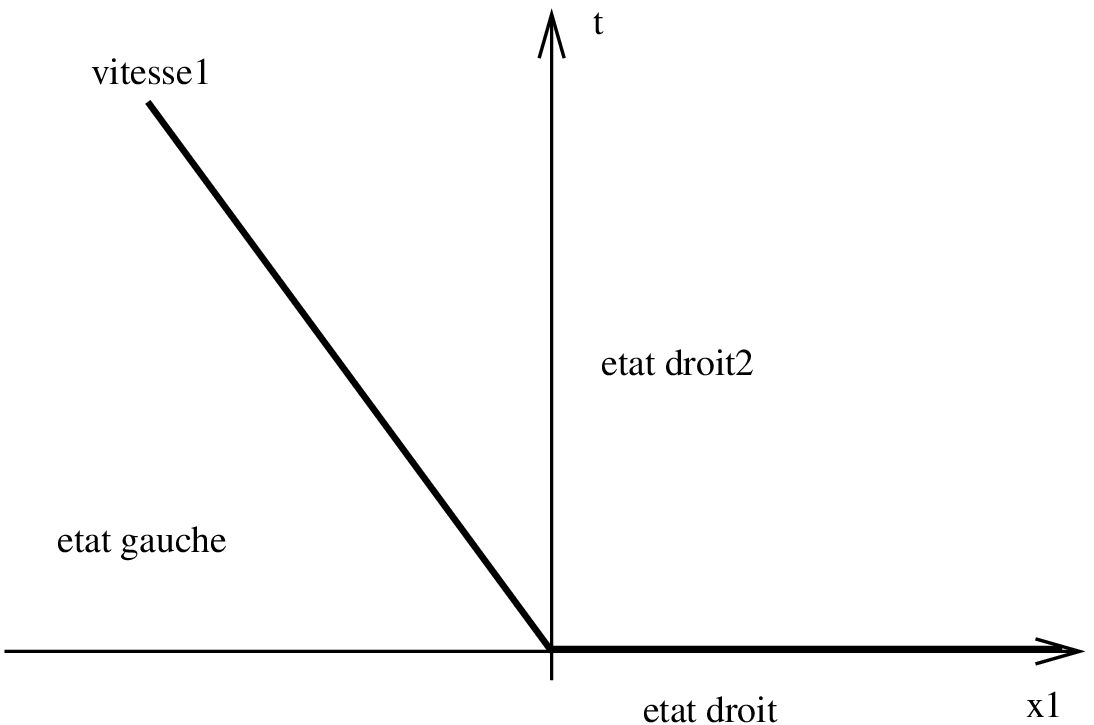}}
\hfill
\psfrag{t}{$t$}
\psfrag{x1}{$x_{1}$}
\psfrag{etat gauche}{$\theta_{\ell}, \rho_{\ell}$}
\psfrag{etat milieu}{Vacuum}
\psfrag{vitesse1}{{\footnotesize $\cos\theta_{\ell}$}}
\psfrag{vitesse2}{{\footnotesize $\cos\theta_{r}$}}
\psfrag{etat droit2}{$\theta_{r},\rho^{\ast},0$}
\psfrag{etat droit}{$\theta_{r}, \rho^{\ast},\bar{p}$}
\subfigure[Case $\theta_{\ell} > \theta_{r}$]{\includegraphics[width=0.3\textwidth]{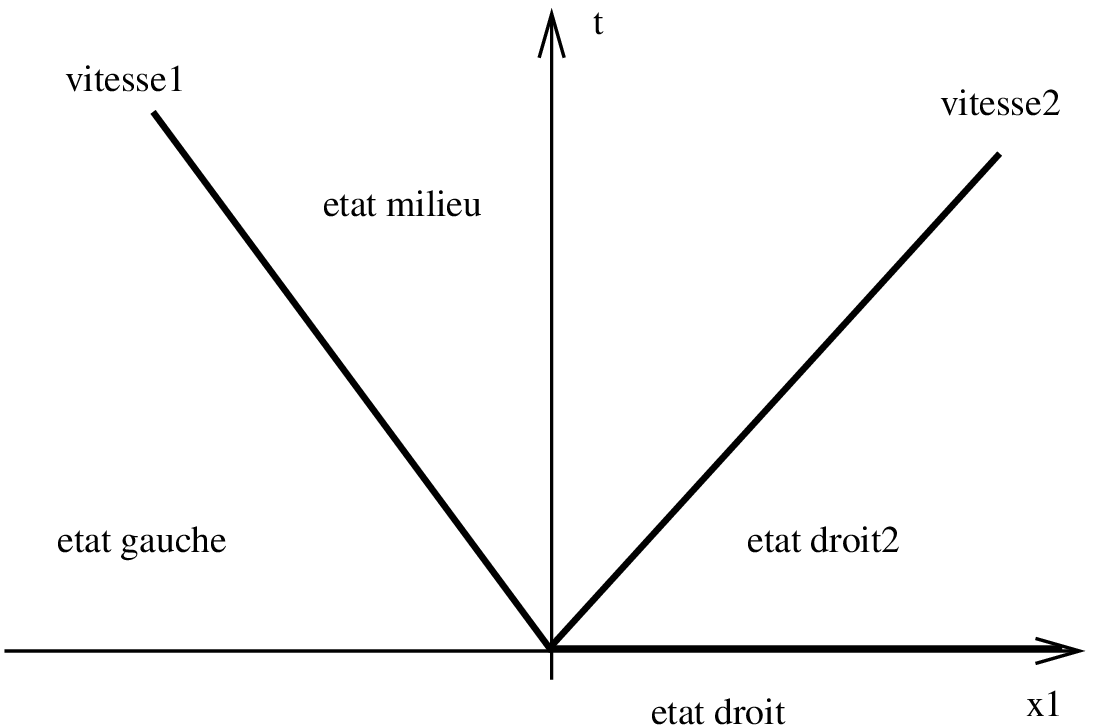}}
\hfill
\psfrag{t}{$t$}
\psfrag{x1}{$x_{1}$}
\psfrag{etat gauche}{$\theta_{\ell}, \rho_{\ell}$}
\psfrag{vitesse1}{}
\psfrag{etat droit2}{$\theta_{r},\rho^{\ast},\bar{\bar{p}}$}
\psfrag{etat droit}{$\theta_{r}, \rho^{\ast},\bar{p}$}
\subfigure[Case $\theta_{\ell} < \theta_{r}$]{\includegraphics[width=0.3\textwidth]{solriemannlim4.eps}}
\hfill
\null
\caption{Limit solutions of the Riemann problem for $\eps = 0$ and $\rho_{\ell} < \rho^{\ast}$, $\rho_{r}^{\eps} \rightarrow \rho^{\ast}$.}
\label{Fig:limit_Riemann_2}
\end{center}
\end{figure}
\begin{figure}
\begin{center}
\end{center}
\end{figure}

\subsubsection{Case $\rho_{\ell} = \rho_{r} = \rho^{\ast}, \rho_{\ell}^{\eps} < \rho_{r}^{\eps}$ (see fig. \ref{Fig:limit_Riemann_3})} 

We assume in addition that $\eps p(\rho_{\ell}^{\eps})$ and $\eps p(\rho_{r}^{\eps})$ have finite positive limits, denoted by $\bar{p}_{\ell} > 0$ and $\bar{p}_{r} > 0$. Figure \ref{Fig:limit_Riemann_3} provides a sketch of the solutions. 
\begin{proposition} \textbf{(Case $\rho_{\ell} = \rho_{r} = \rho^{\ast}, \rho_{\ell}^{\eps} < \rho_{r}^{\eps}$)} There are only three cases:
\begin{enumerate}[(a)]
	\item \textbf{Subcase $\theta_{\ell} = \theta_{r}$.} The limit solution  consists of a uniform constant state $(\rho^{\ast},\theta_{\ell},\bar{p}_{\ell})$.
  \item \textbf{Subcase $\theta_{\ell} > \theta_{r}$.} The limit solution consists of two contact waves and two cluster contact with infinite travelling speed:
\begin{equation*}
 (\rho^{\ast},\theta_{\ell},\bar{p}_{\ell}) \stackrel{\text{declust.}}{\longrightarrow} (\rho^{\ast},\theta_{\ell},0) \stackrel{\text{contact}}{\longrightarrow} (0,\theta_{\ell}) \stackrel{\text{vacuum}}{\longrightarrow} (0,\theta_{r}) \stackrel{\text{contact}}{\longrightarrow} (\rho^{\ast},\theta_{r},0) \stackrel{\text{declust.}}{\longrightarrow} (\rho^{\ast},\theta_{r},\bar{p}_{r}),
\end{equation*}
  \item \textbf{Subcase $\theta_{\ell} < \theta_{r}$.} The limit solution consists of two shock waves with infinite propagation speed connecting the left state $(\rho^{\ast},\theta_{\ell},\bar{p}_{\ell})$ to $(\rho^{\ast},\widetilde{\theta},+\infty)$ and then $(\rho^{\ast},\widetilde{\theta},+\infty)$ to $(\rho^{\ast},\theta_{r},\bar{p}_{r})$:
\begin{equation*}
 (\rho^{\ast},\theta_{\ell},\bar{p}_{\ell}) \stackrel{\text{shock}}{\longrightarrow} (\rho^{\ast},\widetilde{\theta},+\infty) \stackrel{\text{shock}}{\longrightarrow} (\rho^{\ast},\theta_{r},\bar{p}_{r}),
\end{equation*} 
where $\widetilde{\theta}$ is the only solution of 
\begin{equation}  
\frac{[\Psi(\cos(\theta))]_{r}[\cos(\theta)]_{r}}{[\Psi(\cos(\theta))]_{\ell}[\cos(\theta)]_{\ell}} = \left(\frac{\bar{p}_{\ell}}{\bar{p}_{r}}\right)^{\frac{1}{\gamma}}. 
\label{Eq:cluster_contact_interm_angle}
\end{equation}

\end{enumerate}
\label{Prop:limit_Riemann_3}
\end{proposition}
These solutions display only one kind of interface among those discussed in the Formal Statement \ref{Formal_Statement}: the case (b) is an occurence of an interface (C)-(V). According to the cases (a) and (c), the solution inside clusters is continuous. However the case (c) does not provide a meaningful solution since the pressure becomes infinite and this is why Formal Statement \ref{Formal_Statement} does not allow to decide what happens at the interface (C)-(C), i.e. a collision of two clusters. It seems to result from the fact that in this case the Riemann problem models the collision of two infinite one-dimensional clusters. Section \ref{Chap:Clusters_dynamics} provides a description of the collision between finite-size one-dimensional clusters. We have seen that the pressure $\bar{p}$ involves a Dirac delta in time. Indeed, according to case (c), the infinite propagation speed of the waves inside clusters implies the discontinuity of the function $\theta$ in time: $\theta = \theta_{\ell} + (\widetilde{\theta} - \theta_{\ell})H(t-t_{c})$, where $H$ denotes here the Heaviside function. Then, equation (\ref{Eq:theta_1D_cons}) leads to 
\begin{equation}
(\Psi(\widetilde{\theta}) - \Psi(\theta_{\ell}))\delta(t-t_{c}) = \partial_{x} \bar{p},
\end{equation}
which justifies to look for a pressure with a dirac delta in time. The Riemann problem does not allow to take into account such a pressure.   

The proof of the proposition is deferred appendix \ref{Appendix:limit_Riemann_3}.

\begin{figure}
\begin{center}
\null
\hfill
\psfrag{t}{$t$}
\psfrag{x1}{$x_{1}$}
\psfrag{etat gauche}{$\theta_{r}, \rho^{\ast},\bar{p}_{\ell}$}
\psfrag{etat milieu}{$\theta_{r},\rho^{\ast},\bar{p}_{\ell}$}
\psfrag{etat droit}{$\theta_{r}, \rho^{\ast},\bar{p}_{r}$}
\subfigure[Case $\theta_{\ell} = \theta_{r}$]{\includegraphics[width=0.3\textwidth]{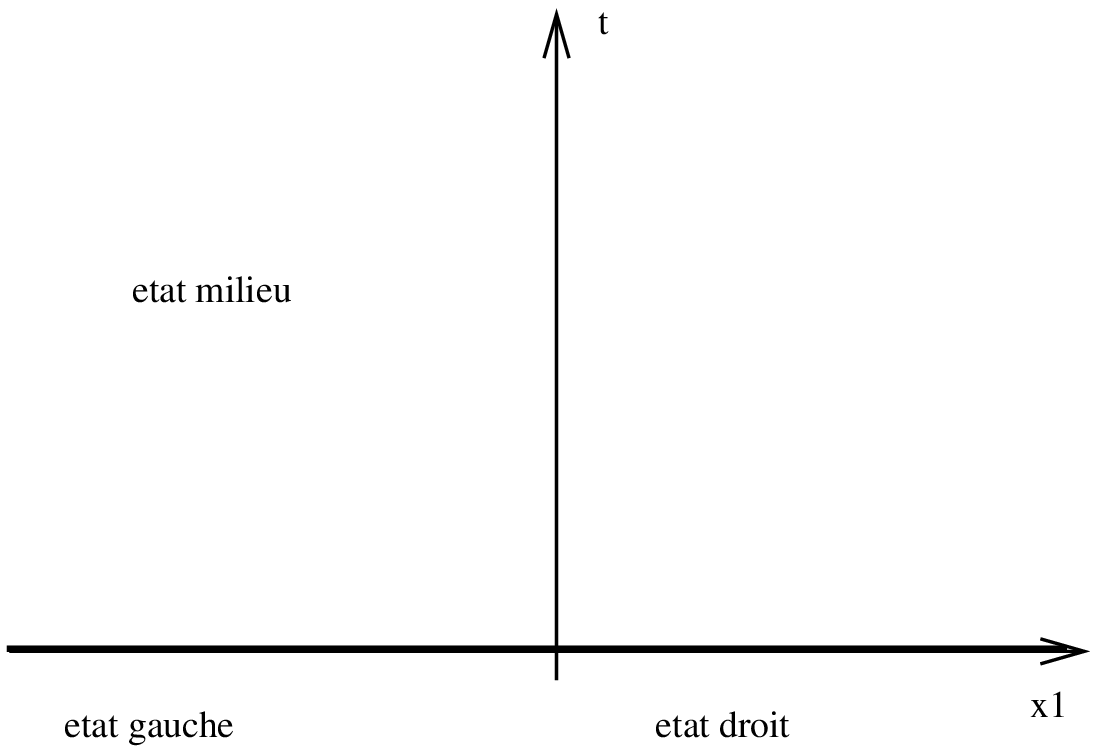}}
\hfill
\psfrag{t}{$t$}
\psfrag{x1}{$x_{1}$}
\psfrag{etat gauche}{$\theta_{\ell}, \rho^{\ast},\bar{p}_{\ell}$}
\psfrag{etat gauche2}{$\theta_{\ell}, \rho^{\ast},0$}
\psfrag{etat milieu}{Vacuum}
\psfrag{vitesse1}{{\footnotesize $\cos\theta_{\ell}$}}
\psfrag{vitesse2}{{\footnotesize $\cos\theta_{r}$}}
\psfrag{etat droit2}{$\theta_{r},\rho^{\ast},0$}
\psfrag{etat droit}{$\theta_{r}, \rho^{\ast},\bar{p}_{r}$}
\subfigure[Case $\theta_{\ell} > \theta_{r}$]{\includegraphics[width=0.3\textwidth]{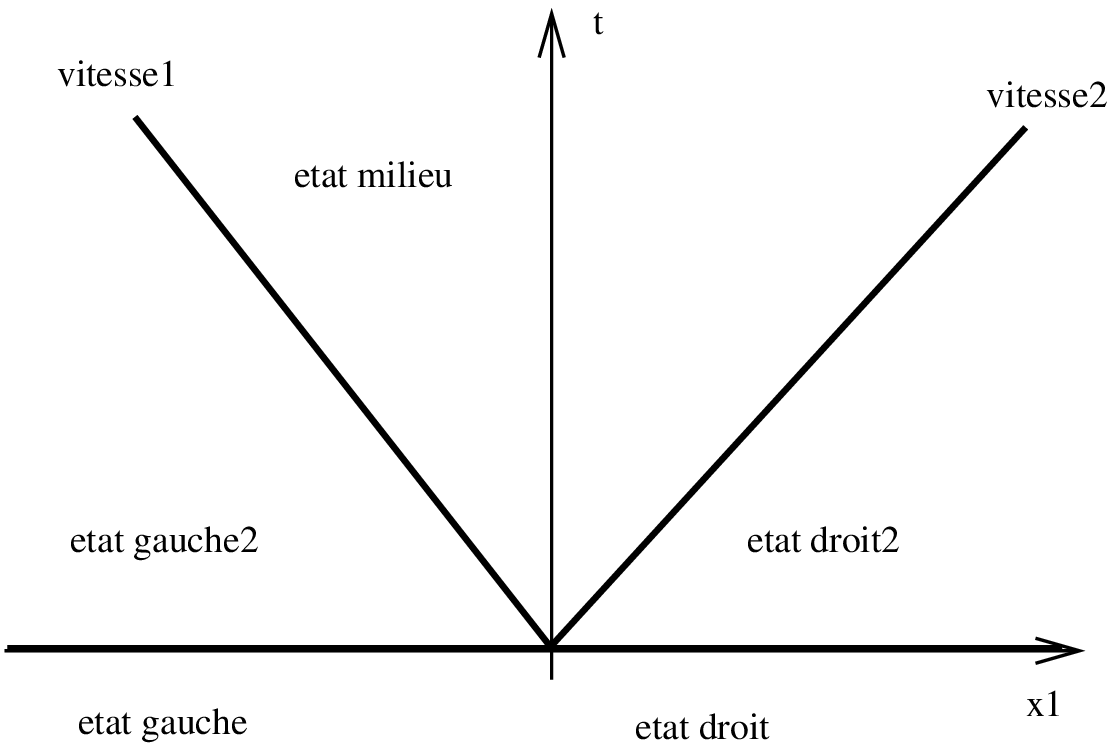}}
\hfill
\psfrag{t}{$t$}
\psfrag{x1}{$x_{1}$}
\psfrag{etat gauche}{$\theta_{\ell}, \rho^{\ast},\bar{p}_{\ell}$}
\psfrag{etat milieu}{$\widetilde{\theta},\rho^{\ast},+\infty$}
\psfrag{etat droit}{$\theta_{r}, \rho^{\ast},\bar{p}_{r}$}
\subfigure[Case $\theta_{\ell} < \theta_{r}$]{\includegraphics[width=0.3\textwidth]{solriemannlim5.eps}}\hfill
\null
\caption{Limit solutions of the Riemann problem for $\eps = 0$ and $\rho_{\ell}^{\eps}, \rho_{r}^{\eps} \rightarrow \rho^{\ast}$, $\rho_{\ell}^{\eps} < \rho_{r}^{\eps}$.}
\label{Fig:limit_Riemann_3}
\end{center}
\end{figure}
\begin{figure}
\begin{center}
\end{center}
\end{figure}

\subsection{Connecting the Riemann problem analysis to the Formal Statement \ref{Formal_Statement}}

We remind that, by contrast with the finite $\eps$ system (\ref{Eq:rho_eps})-(\ref{Eq:constraint_eps}), which is a standard hyperbolic system, the limit system (\ref{Eq:rho_lim})-(\ref{Eq:constraint_rho_lim}) exhibits two additional characteristics:
\begin{enumerate}
\item[(i)] the appearance of clusters corresponding to the saturation of the constraint $\rho = \rho^{\ast}$,

\item[(ii)] the appearance of vacuum.
\end{enumerate}

The previous study helps in understanding the dynamics of the interfaces between unclustered states (UC) $0 < \rho < \rho^{\ast}$, vacuum (V) $\rho = 0$ and clusters (C) $\rho = \rho^{\ast}$. Up to now, no rigorous theory for the limit $\eps \rightarrow 0$ exists and so, we cannot have access to these dynamics rigorously. Our method is to investigate these dynamics through the inspection of the limit $\eps \rightarrow 0$ of the solutions of the Riemann problem of the finite $\eps$ system.

The first remark is that waves with infinite speed correspond to an instantaneous transition from the initial data to some different solution. To some extent, this means that the corresponding initial datum is unstable, and therefore, that it will never appear spontaneously in the course of the evolution of the system. Therefore, we can discard initial data which exhibit this phenomenon, and replace them by the one which is found after the infinite speed wave has been applied.

The second remark is that the various solutions of the Riemann problem can be grouped by situations corresponding to the four cases listed in the formal statement \ref{Formal_Statement}, i.e. interfaces (C)-(UC), (UC)-(V), (C)-(V) and (UC)-(UC). As discussed in the lines following statement \ref{Prop:limit_Riemann_3}, the case (C)-(C) is not accessible by the Riemann problem analysis because the collision dynamics of two clusters depend on their size, and because the Riemann problem only allows to consider infinite size clusters. This is why cluster collisions are analyzed separately in proposition \ref{Prop:cluster_collision}. However, cluster collisions is a complex phenomenon in 2D and proposition \ref{Prop:cluster_collision} only provides a one-dimensional analysis. The two-dimensional analysis is still in progress.
 
To highlight the link between the Riemann problem analysis and the Formal Statement \ref{Formal_Statement}, we point out which solutions of the Riemann problem correspond to which case in the Formal Statement \ref{Formal_Statement}:

\paragraph{Interface (C)-(UC):} it appears in

- prop. \ref{Prop:limit_Riemann_1}, subcase (c): we see that the intermediate congested state is separated from the left and right states by two (C)-(UC) interfaces. We notice that (\ref{Eq:FS_CUC_pressure}) and (\ref{Eq:FS_CUC_speed}) are respectively the relations for the pressure and the shock speed stated at prop. \ref{Prop:limit_Riemann_1}.

- prop. \ref{Prop:limit_Riemann_2}, subcase (c) (if ignoring the contact wave with infinite speed). We can also view subcase (a) as a particular case where the velocities of (C) and (UC) are equal (in which case $\bar{p} = 0$ in the cluster and the interface moves with the common velocity). Again, these two cases are consistent with (\ref{Eq:FS_CUC_pressure}) and (\ref{Eq:FS_CUC_speed}). 

\paragraph{Interface (UC)-(V):} it appears in

- prop. \ref{Prop:limit_Riemann_1}, subcase (b) as two contact waves between the unclustered left and right states and the vacuum intermediate state,

- prop. \ref{Prop:limit_Riemann_2}, subcase (b) where the first wave is a contact between the left unclustered state and the vacuum middle state.
 
In both cases, the velocity of the interface is that of the non-vacuum states and, of course, the pressure is identically zero. Therefore, the situation is as depicted in Formal Statement \ref{Formal_Statement}.
 
\paragraph{Interface (C)-(V):} it appears in

- prop. \ref{Prop:limit_Riemann_2}, subcase (b), where the second wave is a contact wave between the vacuum middle state and the right clustered state. We notice that in this case, the clustered state  must have zero pressure (otherwise, a declustering wave instantaneously relaxes the pressure to zero), 

- prop. \ref{Prop:limit_Riemann_3}, subcase (c), where the left and right clustered states are seperated by a vacuum  intermediate state. Again, in this case, the pressure inside the clusters is identically zero.

In both cases, the velocity of the (C)-(V) interface is that of the cluster. Therefore, the situation is as depicted in the Formal Statement \ref{Formal_Statement}. 
 
\paragraph{Interface (UC)-(UC):} it appears in

- prop. \ref{Prop:limit_Riemann_1}, subcase (a). We see that this situation is that of a standard contact discontinuity  for the uncongested system. The velocities on the  uncongested states are equal and equal to that of the interface, and of course, the pressure is identically zero. Therefore, the situation is again as depicted is the Formal Statement \ref{Formal_Statement}.    
\\
 
We feel that these observation provide a very strong support to the Formal Statement~\ref{Formal_Statement}. As pointed out above, this statement allows to close the system (\ref{Eq:rho_lim})-(\ref{Eq:constraint_rho_lim}) at least until clusters meet. In the one-dimensional framework, proposition \ref{Prop:cluster_collision} provides the cluster collision dynamics. The investigation of cluster dynamics in the two-dimensional case is still work in progress.  

\section{Conclusion}

In this paper, we have studied a continuum model describing a particle system with short-range repulsive and long-range attractive interaction. This is a model for the study of gregariousness among mammal species for instance. We have focused on the effect of the short-range repulsion and looked at the regime where the interaction is turned on suddenly when the local density becomes close to some limit associated to congestion. We have modeled this effect by introducing a perturbation parameter $\eps$ and studied the limit $\eps \rightarrow 0$. We have shown that, in the limit regime, the congested  regions are domains where the flow is incompressible. The complete determination of the limit system requires the knowledge of the interface conditions at the boundaries of the congested regions. We have derived these conditions by looking at a model one-dimensional situation (corresponding to the normal direction to the interface) and analyzing the solutions of the Riemann problem for the perturbation model (with finite $\eps$). Taking the limit $\eps \rightarrow 0$ in the solutions of the Riemann problem allowed us to provide the missing conditions at the interfaces. 

The perspectives of this work are, at the theoretical level, to try to provide more solid justifications to these interface conditions and to analyze the cluster collision dynamics, which is not accessible by the Riemann problem. At the numerical level, we will seek numerical methods to solve this constrained hyperbolic problem and we will perform numerical comparisons between the continuum model and the more fundamental particle system.

\renewcommand{\thesection}{\Alph{section}}
\setcounter{section}{0} 
\section{Derivation of a macroscopic model of short-range repulsive and long-range attractive interactions}
\label{Appendix:derivation}
\subsection{Individual Based Model with speed and congestion constraints.}

We consider $N$ particles in $\R^{2}$ labeled by $k \in \left\{1,..,N\right\}$. These particles are discs of radii $d$. The motion of the particles is described by the time evolution of their positions $\vec{X}_{k}$ and velocity vectors $\vec{\omega}_{k}$. Like in the Vicsek algorithm \cite{95_Vicsek_PhasTrans2d,2008_ContinuumLimit_DM,2002_Couzin_CollectMem}, the velocity magnitude of each particle is the same, is constant in time and supposed equal to $c > 0$. The velocity direction $\vec{\omega}_{k}$ belongs to the unity circle $\S^{1} = \left\{ \vec{\omega} \in \R^{2}, \left|\vec{\omega}\right|^{2} = 1\right\}$. This is a usual assumption in the modeling of several biological systems like flocks of birds \cite{2008_Topologic_Ballerini}, schools of fish \cite{2008_PTW_Gautrais,2008_PTA_D_SM} or herds of sheep \cite{MHPillotM2,MHPillot_Art}.

We start with a simple continuous-in-time model of a particle system subject to attractive-repulsive binary interactions which describe the aggregation of particles with occupation constraints. The evolution of the positions and velocities is given by: 
\begin{eqnarray}
\frac{d\vec{X}_{k}}{dt} &=& c\vec{\omega}_{k},\label{Eq:x}\\
\frac{d\vec{\omega}_{k}}{dt} &=& (\mbox{Id} - \vec{\omega}_{k}\otimes\vec{\omega}_{k})(\nu^{a}_{k}\vec{\xi}_{k}^{a} - \nu^{r}_{k}\vec{\xi}_{k}^{r}),\label{Eq:omega}
\end{eqnarray}
where $\nu^{a}_{k}\vec{\xi}^{a}$ and $\nu^{r}_{k}\vec{\xi}^{r}$ are the attractive and repulsive forces respectively. The matrix $(\mbox{Id} - \vec{\omega}_{k}\otimes\vec{\omega}_{k})$ is the orthogonal projector onto the plane orthogonal to $\vec{\omega}_{k}$. It is applied to both forces in order to keep the magnitude of the speed constant in time. $\vec{\xi}^{a}$ and $\vec{\xi}^{r}$ are the local centers of mass of the particle distribution inside interaction discs centered at $\vec{X}_{k}$ with radii respectively equal to $R_{a}$ and $R_{r}$
\begin{equation*}
\vec{\xi}_{k}^{a} = \frac{\displaystyle{\sum_{j, |\vec{X}_{j} - \vec{X}_{k}| \leq R_{a}}} (\vec{X}_{j} - \vec{X}_{k})}{\displaystyle{\sum_{j, |\vec{X}_{j} - \vec{X}_{k}| \leq R_{a}}} 1},\quad  \vec{\xi}_{k}^{r} = \frac{\displaystyle{\sum_{j, |\vec{X}_{j} - \vec{X}_{k}| \leq R_{r}}} (\vec{X}_{j} - \vec{X}_{k})}{\displaystyle{\sum_{j, |\vec{X}_{j} - \vec{X}_{k}| \leq R_{r}}} 1}.
\end{equation*} 
$\nu_{k}^{a}$ and $\nu_{k}^{r}$ are scaling factors which provide the intensities of the forces. The repulsive force radius $R_{r}$ is supposed much smaller than the attractive force radius $R_{a}$. The resulting force attracts the particles towards the center of mass  of the particle distribution at large distances and repels them from the center of mass of the particle distribution at short distances. To some extent, it is an implementation of the attractive-repulsive scheme proposed by Couzin \cite{2002_Couzin_CollectMem}.

Finally, we suppose that $\nu_{a}$ is a constant and $\nu_{r}$ depends on the local density inside the repulsive interaction disc:
\begin{equation*}
\nu^{a}_{k} = \nu_{a},\quad \nu^{r}_{k} = \nu_{r}\left( \rho_{k}^{r}\right),\quad \rho_{k}^{r} = \frac{\pi d^{2}\sum_{j, |\vec{X}_{j} - \vec{X}_{k}| \leq R_{r}} 1}{\pi R_{r}^{2}}.
\end{equation*}
where $\nu_{r}$ is an increasing function. The function $\nu_{r}$ prevents the local density from exceeding the maximal density $\rho^{\ast}$ which corresponds to the case where all particles are in contact with their neighbours. Clearly, $\rho^{\ast}$ is the ratio of the maximal occupied surface in a disk of radius $R_{r}$ by disks of radii $d$ and is of the order of unity. Therefore, the function $\nu_{r}$ tends to infinity as $\rho_{k}^{r} \rightarrow \rho^{\ast}$. We defer the explicit choice of the function $\nu_{r}$ to the end of the section.
          
\subsection{Mean-field model, hydrodynamic limit and macroscopic model}

The goal of this appendix is to provide a model for large systems of interacting particles according to (\ref{Eq:x})-(\ref{Eq:omega}) at large time and space scales. For this purpose, we will perform a sequence of rescalings. The first rescaling aims at taking into account the large number of interacting particles: it leads to the so-called mean field model. Assuming that the system is included in a fixed box, the limit $N \rightarrow + \infty$ implies that the area $\pi d^{2}$ occupied by each particle tends to $0$ like $1/N$ in such a way that the total area $N\pi d^{2}$ occupied by the particles remains constant. We denote by $\alpha = \lim_{N \rightarrow +\infty} N\pi d^{2}$ the fraction of the surface occupied by the particles. The second rescaling is a hydrodynamic scaling where large time and space scales are considered. Both scaling are classical and are well detailed in \cite{1994_DiluteGazBook_Cercignani,2003_BoltzmannReview_Degond}. They have been applied to swarming model in \cite{2008_ContinuumLimit_DM,2009_Carillo_DoubleMilling,2008_Flocking_HaTadmor}.

\subsubsection{Formal derivation of the mean-field model}

We refer to \cite{spohn_book} for classical references on the mean-field limit. We consider the empirical distribution $f^{N}(\vec{x},\vec{\omega},t)$ defined by
\begin{equation*}
f^{N}(\vec{x},\vec{\omega},t) = \frac{1}{N}\sum_{k = 1}^{N} \delta(\vec{x} - \vec{X}_{k}(t))\delta(\vec{\omega},\vec{\omega}_{k}(t)).
\end{equation*}
$\delta (\vec{x})$ denote the Dirac delta on $\R^{2}$, while $\delta(\vec{\omega},\vec{\omega}_{0})$ denotes the Dirac delta on $\S^{1}$ centered at $\vec{\omega}_{0}$ (i.e. $\delta(\vec{\omega},\vec{\omega}_{0})$ is the probability measure supported by $\left\{\vec{\omega}_{0}\right\}$). It is an easy matter to check that $f^{N}$ satisfies the following kinetic equation 
\begin{equation*}
\partial_{t}f^{N} + c\vec{\omega} \cdot \nabla_{\vec{x}} f^{N} + \nabla_{\vec{\omega}} \cdot \left(\left(F_{a}^{N} - F_{r}^{N}\right)f^{N}\right) = 0,
\end{equation*}
where $F_{a}^{N}$ and $F_{r}^{N}$ are the attractive and repulsive forces, given by
\begin{equation*}
F_{a}^{N}(\vec{x},\vec{\omega},t) = \nu_{a}(\mbox{Id} - \vec{\omega}\otimes\vec{\omega})\vec{\xi_{a}}^{N},\qquad F_{r}^{N}(\vec{x},\vec{\omega},t) = \nu_{r}^{N}(\mbox{Id} - \vec{\omega}\otimes\vec{\omega})\vec{\xi_{r}}^{N},
\end{equation*}
with
\begin{eqnarray*}
&&\vec{\xi_{a}}^{N}(\vec{x},\vec{\omega},t) = \frac{\int K_{a}(\vec{y} - \vec{x})(\vec{y} - \vec{x})\rho^{N}(y,t)dy}{\int K_{a}(\vec{y} - \vec{x})\rho^{N}(\vec{y},t)d\vec{y}},\ \vec{\xi_{r}}^{N}(\vec{x},\vec{\omega},t) = \frac{\int K_{r}(\vec{y} - \vec{x})(\vec{y} - \vec{x})\rho^{N}(\vec{y},t)d\vec{y}}{\int K_{r}(\vec{y} - \vec{x})\rho^{N}(\vec{y},t)d\vec{y}},\\
&&\nu_{r}^{N} = \nu_{r}\left(\frac{N\pi d^{2}\int K_{r}(\vec{y} - \vec{x})\rho^{N}(\vec{y},t)d\vec{y}}{\int K_{r}(\vec{y} - \vec{x})(\vec{y},t)d\vec{y}}\right),
\end{eqnarray*}
where $\rho^{N}(\vec{x},t) = \int_{\Omega \in \S^{1}} f^{N}(\vec{x},\Omega,t)d\Omega$ is the local density and $K_{a}$ (resp. $K_{r}$) is the indicator function of the disc of radius $R_{a}$ (resp. $R_{r}$). Here, it is clear that more general kernels $K_{a}$, $K_{r}$ can be used. 

The formal mean-field limit $N \rightarrow +\infty$ of this model is (we recall that $\alpha = \lim_{N \rightarrow +\infty} N \pi d^{2}$):
\begin{eqnarray}
&&\hspace{-0.5cm}\partial_{t}f + c\vec{\omega} \cdot \nabla_{\vec{x}} f + \nabla_{\vec{\omega}} \cdot \left(\left(F_{a} - F_{r}\right)f\right) = 0,\label{Eq:distrib2}\\ 
&&\hspace{-0.5cm}F_{a}(\vec{x},\vec{\omega},t) = \nu_{a}(\mbox{Id} - \vec{\omega}\otimes\vec{\omega})\vec{\xi_{a}},\qquad F_{r}(\vec{x},\vec{\omega},t) = \nu_{r}(\mbox{Id} - \vec{\omega}\otimes\vec{\omega})\vec{\xi_{r}},\\
&&\hspace{-0.5cm}\vec{\xi_{a}}(\vec{x},\vec{\omega},t) = \frac{\int K_{a}(\vec{y} - \vec{x})(\vec{y} - \vec{x})\rho(\vec{y},t)d\vec{y}}{\int K_{a}(\vec{y} - \vec{x})\rho(\vec{y},t)d\vec{y}},\quad \vec{\xi_{r}}(x,\vec{\omega},t) = \frac{\int K_{r}(\vec{y} - \vec{x})(\vec{y} - \vec{x})\rho(\vec{y},t)d\vec{y}}{\int K_{r}(\vec{y} - \vec{x})\rho(\vec{y},t)d\vec{y}},\\
&&\hspace{-0.5cm}\nu_{r} = \nu_{r}\left(\frac{\int K_{r}(\vec{y} - \vec{x})\rho(\vec{y},t)d\vec{y}}{\alpha\int K_{r}(\vec{y} - \vec{x})d\vec{y}}\right), \rho(\vec{x},t) = \int f(\vec{x},\vec{\omega},t) d\vec{\omega}.\label{Eq:repulsive2}
\end{eqnarray}
Rigorous justifications of this limit are outside the scope of this article.

\subsubsection{Hydrodynamic scaling}

In order to select the relevant scales, we first rewrite our system in dimensionless variables. We consider a space scale $x_{0}$ (typically the range $R_{r}$ of the repulsive force) and we choose a time scale $t_{0} = x_{0}/c$. The associated dimensionless time and space variables are $t' = t/t_{0}$ and $\vec{x'}= \vec{x}/x_{0}$. We also introduce scaled collision kernels $K'_{a,r}$ such that $K_{a,r}(x_{0}\vec{x'}) = K_{a,r}'(\vec{x})$, scaled intensities $\nu_{a}' = \nu_{a}c/x_{0}^{2}$, $\nu_{r}' = \nu_{r}c/x_{0}^{2}$ and a scaled distribution function $f' = \alpha f$. In the case where $K_{a,r}$ are indicator functions of balls of radii $R_{a,r}$, this amounts to rescaling the radii to new values $R'_{a,r} = R_{a,r}/x_{0}$. After removing the primes, the system in the new variables and with the new unknowns is similar to  (\ref{Eq:distrib2})-(\ref{Eq:repulsive2}) but with $c = 1$ and $\alpha = 1$.

To derive the large time and space dynamics, we introduce the following change of variables $\vec{\tilde{x}} = \eta \vec{x}$, $\tilde{t} = \eta t$ with $\eta \ll 1$. In the new variables, the distribution function $f^{\eta}(\vec{\tilde{x}},\vec{\omega},\tilde{t}) = f(\vec{x},\vec{\omega},t)$ satisfies the following system (omitting the tildes):
\begin{eqnarray*}
&&\eta\left(\partial_{t}f^{\eta} + \vec{\omega} \cdot \nabla_{\vec{x}} f^{\eta}\right) + \nabla_{\vec{\omega}} \cdot \left(\left(F_{a}^{\eta} - F_{r}^{\eta}\right)f^{\eta}\right) = 0,\\ 
&&F_{a}^{\eta}(\vec{x},\vec{\omega},t) = \nu_{a}^{\eta}(\mbox{Id} - \vec{\omega}\otimes\vec{\omega})\vec{\xi_{a}}^{\eta},\ \vec{\xi_{a}}^{\eta}(\vec{x},\vec{\omega},t) = \frac{1}{\eta}\frac{\int K_{a}^{\eta}\left(\frac{\vec{y} - \vec{x}}{\eta}\right)(\vec{y} - \vec{x})\rho^{\eta}(\vec{y},t)d\vec{y}}{\int K_{a}^{\eta}\left(\frac{\vec{y} - \vec{x}}{\eta}\right)\rho^{\eta}(\vec{y},t)d\vec{y}},\\
&&F_{r}^{\eta}(\vec{x},\vec{\omega},t) = \nu_{r}^{\eta}(\mbox{Id} - \vec{\omega}\otimes\vec{\omega})\vec{\xi_{r}}^{\eta},\ \vec{\xi_{r}}^{\eta}(\vec{x},\vec{\omega},t) = \frac{1}{\eta}\frac{\int K_{r}^{\eta}\left(\frac{\vec{y} - \vec{x}}{\eta}\right)  (\vec{y} - \vec{x})\rho^{\eta}(\vec{y},t)d\vec{y}}{\int K_{r}^{\eta}\left(\frac{\vec{y} - \vec{x}}{\eta}\right) \rho^{\eta}(\vec{y},t)d\vec{y}},\\
&&\nu_{r}^{\eta} = \nu_{r}^{\eta}\left(\frac{\frac{1}{\eta}\int K_{r}^{\eta}\left(\frac{\vec{y} - \vec{x}}{\eta}\right) \rho^{\eta}(\vec{y},t)d\vec{y}}{\frac{1}{\eta}\int K_{r}^{\eta}\left(\frac{\vec{y} - \vec{x}}{\eta}\right)d\vec{y}}\right),\ \rho^{\eta}(\vec{x},t) = \int f^{\eta}(\vec{x},\vec{\omega},t) d\vec{\omega},
\end{eqnarray*}  
where $K_{a}^{\eta}$ and $K_{r}^{\eta}$ are the scaled interaction kernels and $\nu_{a}^{\eta}$, $\nu_{r}^{\eta}$, the scaled intensities. 

We first suppose that the repulsive kernel $K_{r}^\eta$ and the repulsive intensity $\nu_r^\eta$ are unchanged in the scaling: $K_{r}^{\eta} = K_{r}$, $\nu_{r}^{\eta}(\rho) = \nu_{r}(\rho)$. This means that the range of the repulsive force is supposed of order $\eta$. 
To analyze the limit $\eta \rightarrow 0$, we first need an expansion of $\vec{\xi_{r}}^{\eta}$ in terms of $\eta$. The following lemma provides the result for an isotropic kernel $K_{r}$ ($K_{r}(\vec{z}) = K_{r}(|\vec{z}|)$) 
\begin{lemma} \label{DL1} Under suitable regularity assumptions on $\rho^{\eta}$, we have the expansion
\begin{eqnarray*}
&&\frac{1}{\eta}\int K_{r}\left(\left|\frac{\vec{y} - \vec{x}}{\eta}\right|\right) \rho^{\eta}(\vec{y})d\vec{y} = a\rho^{\eta}(\vec{x}) + o(\eta),\\
&&\vec{\xi_{r}}^{\eta}(\vec{x},\vec{\omega},t) = \eta \frac{{\mathbf B}\nabla_{\vec{x}}\rho^{\eta}(x)}{a\rho^{\eta}(\vec{x})} + o(\eta),\\
&&\nu_{r}^{\eta}\left(\frac{\frac{1}{\eta}\int K_{r}^{\eta}\left(\frac{\vec{y} - \vec{x}}{\eta}\right) \rho^{\eta}(\vec{y},t)d\vec{y}}{\frac{1}{\eta}\int K_{r}^{\eta}\left(\frac{\vec{y} - \vec{x}}{\eta}\right)d\vec{y}}\right) = \nu_{r}(\rho) + o(1),
\end{eqnarray*}
where $a = \int K_{r}(|\vec{z}|)d\vec{z}$ and ${\mathbf B} = \int K_{r}(|\vec{z}|) \vec{z} \otimes \vec{z} d\vec{z} = \left(\int K_{r}(r)r^{3}dr\, \mbox{Id}\right)$.  
\end{lemma}
The proof of this lemma is elementary and omitted. In the case where $K_{r}$ is the indicator function of the disc of radius $R_{r}$, the coefficients $a$ and ${\mathbf B}$ are equal to $a = \pi R_{r}^{2}$ and ${\mathbf B} = \pi \frac{R_{r}^{4}}{4}\mbox{Id}$. Now we consider the scaling of the attractive kernel $K_{a}^{\eta}$ and attractive intensity $\nu_{a}^{\eta}$. We suppose that the attractive force remains non-local as $\eta$ tends to $0$ and weaker than the repulsive force. To express these assumptions, we suppose that the scaled attractive kernel $K_{a}^{\eta}$ and intensity $\nu_{a}^{\eta}$ are given by 
\begin{equation*}
K_{a}^{\eta}(\vec{z}) = K_{a}(\eta \vec{z}),\quad \nu_{a}^{\eta} = \eta^{2}\nu_{a}.
\end{equation*}
For simplicity, we choose $\nu_{a} = 1$. We also fix the space unit $x_{0}$ in such a way that $B/a = 1$. In particular, in the case where $K_{r}$ is the indicator of the ball of radius $R_{r}$, we can fix $x_{0} = R_{r}/2$. 

Under all these modelling assumptions and thanks to lemma \ref{DL1}, the system can be written formally, in the limit $\eta \rightarrow 0$:
\begin{eqnarray}
&&\partial_{t}f + \vec{\omega} \cdot \nabla_{\vec{x}} f + \nabla_{\vec{\omega}} \cdot \left(\left(F_{a} - F_{r}\right)f\right) = 0,\label{cinetik}\\
&&F_{a}(\vec{x},\vec{\omega},t) = (\mbox{Id} - \vec{\omega}\otimes\vec{\omega})\vec{\xi_{a}},\ \vec{\xi_{a}}(\vec{x},t) = \left(\frac{\int K_{a}\left(\left|\vec{y} - \vec{x}\right|\right)(\vec{y} - \vec{x})\rho(\vec{y},t)d\vec{y}}{\int K_{a}\left(\left|\vec{y} - \vec{x}\right|\right)\rho(\vec{y},t)d\vec{y}}\right),\\
&&F_{r}(\vec{x},\vec{\omega},t) = \nu_{r}(\rho)(\mbox{Id} - \vec{\omega}\otimes\vec{\omega})\vec{\widetilde{\xi}_{r}},\ \vec{\widetilde{\xi}_{r}}(\vec{x},t) = \frac{\nabla_{\vec{x}}\rho(\vec{x},t)}{\rho(\vec{x},t)}.
\end{eqnarray}

\subsubsection{Macroscopic model}

The last step is to obtain the dynamics of macroscopic quantities associated to the flow. Here we will only consider the density and momentum. We find that under suitable regularity and decay assumptions on $f$, the density $\rho = \int f d\vec{\omega}$ and momentum $\rho\Omega = \int f \vec{\omega} d\vec{\omega}$ satisfy the following system of mass and momentum balance equations: 
\begin{eqnarray}
&&\partial_{t}\rho + \nabla_{\vec{x}} \cdot \rho\Omega = 0,\label{EQrho1}\\
&& \partial_{t}\rho\Omega + \nabla_{\vec{x}} \cdot \left(\int f \vec{\omega}\otimes\vec{\omega} d\vec{\omega}\right) = \left( \int (\mbox{Id} -  \vec{\omega}\otimes\vec{\omega}) f d\vec{\omega} \right) \,  \left(\vec{\xi_{a}} - \nu_{r}(\rho)\vec{\widetilde{\xi}_{r}}\right).\label{EQrhoOmega1}\\
&&\vec{\xi_{a}}(\vec{x},t) = \left(\frac{\int K_{a}\left(\left|\vec{y} - \vec{x}\right|\right)(\vec{y}-\vec{x})\rho(\vec{y},t)d\vec{y}}{\int K_{a}\left(\left|\vec{y} - \vec{x}\right|\right)\rho(\vec{y},t)d\vec{y}}\right),\quad \vec{\widetilde{\xi}_{r}}(\vec{x},t) = \frac{\nabla_{\vec{x}}\rho(\vec{x},t)}{\rho(\vec{x},t)}.\label{EQxi}
\end{eqnarray}

To close system (\ref{EQrho1})-(\ref{EQrhoOmega1}), we assume that $f$ is a monokinetic distribution: 
\begin{equation}
f(\vec{x},\vec{\omega},t) = \rho(\vec{x},t)\delta(\vec{\omega},\Omega(\vec{x},t)),
\end{equation} 
with  $\left|\Omega(\vec{x},t)\right| = 1$. This assumptions presupposes that a local equilibrium is reached  where all particles are locally aligned. Although no justification of this assumption can be made at this point, the features displayed by the system seem meaningful in view of gregariousness modelling. We find:
\begin{eqnarray}
\partial_{t}\rho + \nabla_{\vec{x}}\cdot \rho\Omega &=& 0,\label{Eq:rho2}\\
\partial_{t}\left(\rho\Omega\right) + \nabla_{\vec{x}}\cdot\left(\rho \Omega\otimes\Omega\right) &=& \rho(\mbox{Id} -  \Omega\otimes\Omega)(\vec{\xi_{a}} - \nu_{r}(\rho)\vec{\widetilde{\xi}_{r}}),\label{Eq:rhoOmega2}
\end{eqnarray}
where $\vec{\xi_{a}}$ and $\vec{\widetilde{\xi}_{r}}$ are given by (\ref{EQxi}). Factoring out $\rho$ in (\ref{Eq:rhoOmega2}), using (\ref{Eq:rho2}), we also get the following form of the system:
\begin{eqnarray}
&&\partial_{t}\rho + \nabla_{\vec{x}}\cdot \rho\Omega = 0,\\
&&\partial_{t}\Omega + \Omega \cdot \nabla_{\vec{x}}\Omega + \nu_{r}(\rho)(\mbox{Id} -  \Omega\otimes\Omega)\vec{\widetilde{\xi}_{r}} = (\mbox{Id} -  \Omega\otimes\Omega)\vec{\xi_{a}},\\ 
&&\vec{\xi_{a}}(\vec{x},t) = \left(\frac{\int K_{a}\left(\left|\vec{y} - \vec{x}\right|\right)(\vec{y}-\vec{x})\rho(\vec{y},t)d\vec{y}}{\int K_{a}\left(\left|\vec{y} - \vec{x}\right|\right)\rho(\vec{y},t)d\vec{y}}\right),\quad \vec{\widetilde{\xi}_{r}}(\vec{x},t) = \frac{\nabla_{\vec{x}}\rho(\vec{x},t)}{\rho(\vec{x},t)}. 
\end{eqnarray} 
\subsection{Repulsive force intensity and macroscopic model}

Let us return now to the choice of the function $\nu_{r}$. This function tends to infinity when $\rho \rightarrow \rho^{\ast}$. Like in the traffic model devised in \cite{2008_Traffic_DegondRascle}, we assume that this function behaves like $\rho^{\gamma}$ when $\rho \ll \rho^{\ast}$ and tends to infinity when $\rho \rightarrow \rho^{\ast}$. The prototype of such a function is
\begin{equation}
p(\rho) = \frac{1}{\left(\frac{1}{\rho^{\ast}} - \frac{1}{\rho}\right)^{\gamma}},
\label{Def:p}
\end{equation} 
where $\gamma \geq 1$. We will keep this example constantly in the paper for simplicity but the results are valid for all functions having the same properties. We consider that $\nu_{r}(\rho) = \rho p'(\rho)$. In this way, we suppose that repulsion acts like a standard presure force in a gas, but, when the density reaches the congestion density $\rho^{\ast}$, the pressure tends to infinity. Since the equation for $\Omega$ is used instead of that for $\rho\Omega$, the interpretation of $p$ in standard gas dynamics terms would rather be that of an enthalpy (i.e. $p'(\rho)  = P'(\rho)/\rho$ where $P$ is the actual fluid mechanical pressure), but the results would be similar if we considered the equation for $\rho\Omega$ instead. Indeed, because of the constraint $|\Omega| = 1$, the system is non-conservative in the projection term $(\mbox{Id} -  \Omega\otimes\Omega)$. Finally, we get the following system
\begin{eqnarray}
&&\partial_{t}\rho + \nabla_{\vec{x}}\cdot \rho\Omega = 0,\label{Eq:rho_annex}\\
&&\partial_{t}\Omega + \Omega \cdot \nabla_{\vec{x}}\Omega + (\mbox{Id} -  \Omega\otimes\Omega)\nabla_{\vec{x}}p(\rho) = (\mbox{Id} -  \Omega\otimes\Omega)\vec{\xi_{a}}, \label{Eq:Omega_annex}\\
&&\vec{\xi_{a}}(\vec{x},t) = \left(\frac{\int K_{a}\left(\left|\vec{y} - \vec{x}\right|\right)(\vec{y}-\vec{x})\rho(\vec{y},t)d\vec{y}}{\int K_{a}\left(\left|\vec{y} - \vec{x}\right|\right)\rho(\vec{y},t)d\vec{y}}\right).
\end{eqnarray} 

This system provides the starting point of the present article. Since this paper is focused on the treatment of congestion phenomena, we remove the non-local attractive force. Indeed, this term is a zero-th order derivative term and does not 
intervene in the jump relations across discontinuities. 

\section{Conservative laws for the one-dimensional system} 
\label{Appendix:conservative_laws}

In this appendix, we are looking for conservative forms of the one-dimensional system (\ref{Eq:rho_1d})-(\ref{Eq:theta_1d}). The most general conservative form is written: 
\begin{equation}
\partial_{t} g(\rho,\theta) + \partial_{x} f(\rho,\theta) = 0.\label{Eq:cons_form} 
\end{equation}
where $g$ and $f$ are smooth functions of $\rho$ and $\theta$. The following proposition exhibits an infinite set of such conservative forms. 
\begin{proposition} If $(g,f)$ is a conservative form of (\ref{Eq:rho_1d})-(\ref{Eq:theta_1d}), then their partial derivatives are related by  
\begin{equation}
\frac{\partial f}{\partial \rho} = \frac{\partial g}{\partial \rho}\cos\theta - \frac{\partial g}{\partial \theta} \sin\theta p'(\rho),\quad \frac{\partial f}{\partial \theta} = \frac{\partial g}{\partial \theta}\cos\theta - \frac{\partial g}{\partial \rho}\rho\sin\theta.
\label{Eq:derivative_relation}
\end{equation}
Moreover, if $g$ is a function with separated variables $g(\rho,\theta) = u(\theta)v(\rho)$, then $u$ and $v$ satisfy
\begin{eqnarray}
&&\rho v''(\rho) = k p'(\rho)v(\rho),\label{Eq:v}\\
&& u''(\theta) + (\text{cotan}\theta) u'(\theta) = ku(\theta),\label{Eq:u}
\end{eqnarray}
where $k$ is a constant real number. Each $k \in \R$ gives rise to possible $(g,f)$ pairs.
\end{proposition}
\begin{proof} Performing the chain rule in (\ref{Eq:cons_form}) and using (\ref{Eq:rho_1d}),(\ref{Eq:theta_1d}), we easily get (\ref{Eq:derivative_relation}). Then, using that differentiations with respect to $\rho$ and $\theta$ commute, (\ref{Eq:derivative_relation}) gives rise to an elliptic equation satisfied by $g$ and inserting the hypothesis of separated variables, we obtain (\ref{Eq:v}),(\ref{Eq:u}). 

Equation (\ref{Eq:u}) is the Legendre differential equation (in polar coordinates). The two-dimensional vector space of solutions of this equation is spanned by the Legendre functions of first and second species and each of them gives rise to possible $(g,f)$ pairs. \hfill$\square$
\end{proof}
The solutions of (\ref{Eq:v}) exists for all $k \in \R$. However, they have a priori no explicit expression exept for $k = 0$. In this case, the 2-dimensional vector space of solutions of (\ref{Eq:v}) is spanned by $\left\{1,\rho\right\}$ and for (\ref{Eq:u}), is $\left\{1,\Psi(\cos\theta)\right\}$. We can actually check that the following $(g,f)$ pairs
\begin{eqnarray}
&&(g,f) = (\rho,\rho\cos\theta),\label{Eq:cons_form1}\\
&&(g,f) = (\Psi(\cos\theta),\Phi(\cos\theta) + p(\rho)),\label{Eq:cons_form2}\\
&&(g,f) = (\rho\Psi(\cos\theta),\rho\cos\theta\Psi(\cos\theta) + P(\rho)),\label{Eq:cons_form3}
\end{eqnarray} where $P$ is an antiderivative of $\rho p'(\rho)$, are non trivial solutions. The conservative form studied in this article corresponds to the pairs (\ref{Eq:cons_form1}) and (\ref{Eq:cons_form2}). The pairs (\ref{Eq:cons_form1}) and (\ref{Eq:cons_form3}) form another such conservative system. 

\section{Proof of proposition \ref{Prop:cluster_collision} (cluster collisions)}
\label{Appendix:cluster_collision}

\begin{proof} 1- Let $x_{0} \in [a(t_{c}),m]$ and $h \in \mathcal{C}_{c}^{\infty}(D'')$, where $D''$ is a neighbourhood of $x_{0}$ in $D$ (cf. figure \ref{Fig:cluster_collision}). We apply the Green formula on the domain $D''$:
\begin{eqnarray}
<\partial_{t}\Psi(\cos\theta) + \partial_{x}\Phi(\cos\theta),h > &=& - \iint_{D''} \Phi(\cos\theta)\partial_{x}h + \Psi(\cos\theta)\partial_{t}h\ dtdx\nonumber\\
&=& \int_{\partial D''} h[(\Phi(\cos\theta),\Psi(\cos\theta))\cdot n] ds\nonumber\\ 
&=& (\Psi(\cos\theta) - \Psi(\cos\theta_{\ell})) \int_{x_{0}-\eta}^{x_{0}+\eta} h(t_{c},x)dx,
\label{Eq:collisioncluster_lim1}
\end{eqnarray} 
where $<.,.>$ denotes the duality brackets. Since we look for $\bar{p}(x,t) = \pi(x)\delta(t-t_{c})$, we also have
\begin{equation}
- <\partial_{x} \bar{p},h> = \iint_{D''}\bar{p}\partial_{x}h\ dtdx = \int_{x_{0}-\eta}^{x_{0}+\eta} \pi(x)\partial_{x}h(t_{c},x)dx = - \int_{x_{0}-\eta}^{x_{0}+\eta} \partial_{x}\pi(x) h(t_{c},x)dx.
\label{Eq:collisioncluster_lim2}
\end{equation}
If (\ref{Eq:theta_1D_cons}) is satisfied, then we have
\begin{equation}
<\partial_{t}\Psi(\cos\theta) + \partial_{x}\Phi(\cos\theta),h > = - <\partial_{x} \bar{p},h>,
\end{equation}
and equations (\ref{Eq:collisioncluster_lim1}) and (\ref{Eq:collisioncluster_lim2}) imply that
\begin{equation*}
(\Psi(\cos\theta) - \Psi(\cos(\theta_{\ell}))) = - \partial_{x}\pi(x_{0}).
\end{equation*}
The same arguments (for any $x \in [a(t_{c}),b(t_{c})]$) lead to 
\begin{equation*}
- \partial_{x}\pi(x) = \left\{\begin{array}{ll}(\Psi(\cos\theta) - \Psi(\cos(\theta_{\ell}))),&\text{ if }x \in [a(t_{c}),m],\\ (\Psi(\cos\theta) - \Psi(\cos(\theta_{r}))),&\text{ if }x \in [m,b(t_{c})],\end{array}\right.
\end{equation*}
and (supposing $\pi$ continuous) to 
\begin{equation*}
\pi(x) = \left\{\begin{array}{ll}(\Psi(\cos\theta) - \Psi(\cos(\theta_{\ell})))(m - x)\\
 \quad + (\Psi(\cos\theta) - \Psi(\cos(\theta_{r})))(b(t_{c}) - m),&\text{ if }x \in [a(t_{c}),m],\\
(\Psi(\cos\theta) - \Psi(\cos(\theta_{r})))(b(t_{c}) - x),&\text{ if }x \in [m,b(t_{c})],\end{array}\right.
\end{equation*}
Supposing that $\bar{p}$ and then $\pi$ equal zero outsidse the clusters, we get
\begin{equation*}
\pi(a(t_{c})) = (\Psi(\cos\theta) - \Psi(\cos(\theta_{\ell})))(m - a(t_{c})) + (\Psi(\cos\theta) - \Psi(\cos(\theta_{r})))(b(t_{c}) - m) = 0.
\end{equation*}

2 - Let $h$ be a test function in the neighbourhood $D'$ of $D$ (cf. figure \ref{Fig:cluster_collision}). We denote $D_{1} = D\cap\left\{t \leq t_{c}\right\}$ and $D_{2} = D\cap\left\{t \geq t_{c}\right\}$. Applying Green's formula, we obtain:  
\begin{eqnarray*}
\lefteqn{<\partial_{t}\rho + \partial_{x}(\rho\cos\theta),h > = - \iint_{D} \rho\cos\theta\partial_{x}h + \rho\partial_{t}h\ dtdx}\\
&=& \int_{\partial D_{1}} h[(\rho\cos\theta,\rho)\cdot n] ds + \int_{\partial D_{2}} h[(\rho\cos\theta,\rho)\cdot n] ds\\ 
&=& - \sum_{i \in \left\{\ell,r\right\}} \left[\int_{t_{c}-\delta}^{t_{c}} (- \rho^{\ast}\cos(\theta_{i}) + a_{i}'(t)\rho^{\ast} ) h(t,a_{i}(t)) dt + \int_{a_{i}(t_{c})}^{b_{i}(t_{c})} \rho^{\ast}h(t_{c},x)dx\right.\\
&& \left.+ \int_{t_{c}-\delta}^{t_{c}} (- \rho^{\ast}\cos(\theta_{i}) + b_{i}'(t)\rho^{\ast} ) h(t,b_{i}(t)) dt\right]\\
&& + \int_{t_{c}}^{t_{c} + \delta} (- \rho^{\ast}\cos\theta + a'(t)\rho^{\ast}) h(t,a(t)) dt + \int_{a(t_{c})}^{b(t_{c})} \rho^{\ast}h(t_{c},x)dx\\
&& - \int_{t_{c}}^{t_{c}+\delta} (-  \rho^{\ast}\cos\theta + b'(t)\rho^{\ast} ) h(t,b(t)) dt\\
&=& - \rho^{\ast}\int_{a_{\ell}(t_{c})}^{ b_{\ell}(t_{c})} h(t_{c},x) dx - \rho^{\ast}\int_{a_{r}(t_{c})}^{ b_{r}(t_{c})} h(t_{c},x) dx + \rho^{\ast} \int_{a(t_{c})}^{b(t_{c})} h(t_{c},x) dx\\
&=& 0.
\end{eqnarray*}
since $n \, ds = \pm(-1,x'(t))dt$ on the left and right sides of the domains $D_{1}$ and $D_{2}$ and $n \, ds = \pm(0,1)dx$ on theirs top and bottom sides. The last equality stems from the identity 
\begin{equation*}
(b_{\ell}(t_{c}) - a_{\ell}(t_{c})) + (b_{r}(t_{c}) - a_{r}(t_{c})) = (b(t_{c}) - a(t_{c})).
\end{equation*} 
The density equation (\ref{Eq:rho_1D_cons}) is satisfied in the distributional sense. 

If we now apply Green's formula with a test function $h \in \mathcal{C}_{c}^{\infty}(D)$, we obtain
\begin{eqnarray*}
\lefteqn{<\partial_{t}\Psi(\cos\theta) + \partial_{x}\Phi(\cos\theta),h > = - \iint_{D} \Phi(\cos\theta)\partial_{x}h + \Psi(\cos\theta)\partial_{t}h\ dtdx}\\
&=& \int_{\partial D_{1}} h[(\Phi(\cos\theta),\Psi(\cos\theta))\cdot n] ds + \int_{\partial D_{2}} h[(\Phi(\cos\theta),\Psi(\cos\theta))\cdot n] ds\\ 
&=& \int_{a(t_{c})}^{m} (\Psi(\cos\theta) - \Psi(\cos\theta_{\ell})) h(t_{c},x)dx + \int_{m}^{b(t_{c})} (\Psi(\cos\theta) - \Psi(\cos\theta_{r})) h(t_{c},x)dx\\
&=& - \int_{a(t_{c})}^{b(t_{c})} \partial_{x}\pi h(t_{c},x)dx = - < \bar{p},h>.  
\end{eqnarray*}
Eq. (\ref{Eq:theta_1D_cons}) is satisfied in the distributional sense. Note that in this case, the test function has a compact support in $D$ since $\Psi(\theta)$ is not defined in the vacuum region $\rho = 0$.  
\qed
\end{proof}


\section{Proof of proposition \ref{Prop:Hugoniot_loci} (study of the Hugoniot loci)}
\label{Appendix:Hugoniot_loci}

In this section, we provide a detailed study of the Hugoniot curves. Let $\left(\rho_{\ell},\theta_{\ell}\right) \in \left]0,\rho^{\ast}\right[\times\left]0,\pi\right[$ be an arbitrary left state. We need to find the geometric behaviour of the Hugoniot loci associated to this left state.
The classical theory of nonlinear conservation laws provides only information on the local behaviour of $\mathcal{H}_{\pm}^{\eps}$. Each $\mathcal{H}_{+}^{\eps}$, $\mathcal{H}_{-}^{\eps}$ consists of a one-dimensional manifold tangent to the integral curves of the right eigenvectors up to the second order. In the $(\rho,\Psi(\cos\theta))$-plane, the 1-Hugoniot curve $\mathcal{H}_{-}^{\eps}$ is thus locally decreasing and the 2-Hugoniot curve $\mathcal{H}_{+}^{\eps}$ is locally increasing because of the direction of the vectors $\vec{r}_{\pm}^{\eps} = (\pm \rho\left|\sin\theta\right|,\sqrt{\eps p'(\rho)\rho})$. In the $(\rho,\theta)$-plane, the 1-Hugoniot curve $\mathcal{H}_{-}^{\eps}$ defines a locally increasing function $\theta = (h_{-}^{\eps})^{-1}(\rho)$ while the 2-Hugoniot curve $\mathcal{H}_{+}^{\eps}$ defines a locally  decreasing function $\theta = (h_{+}^{\eps})^{-1}(\rho)$. Actually, this property is global (i.e. $(h_{-}^{\eps})^{-1}$ (resp. $(h_{+}^{\eps})^{-1}$) is a globally increasing (resp. decreasing) function of $\rho$ for all $\rho \in ]0,\rho^{\ast}[$). To prove this, let us begin with a simple and useful lemma.
\begin{lemma} For all $u \in [-1,1]$, the function $f_{u} : v \in  ]-1,1[ \rightarrow \Phi(v) - u\Psi(v)$ is convex and has a minimum at the point $u$. In particular, we have 
\begin{equation*}
\forall v \neq u,\quad \left(\Phi(v) - \Phi(u)\right) - u\left(\Psi(v) - \Psi(u)\right)  = f_{u}(v) - f_{u}(u) > 0.
\end{equation*}\label{fonctionfu}
\end{lemma}
The proof is elementary and omitted. We now analyze the behaviour of $\mathcal{H}_{\pm}^{\eps}$  in more detail when $\eps$ becomes small. Proposition \ref{Prop:Hugoniot_loci} is an immediate consequence of the following lemma.  
\begin{lemma}\label{Lemma:Hugoniot_loci} The behaviour of $\mathcal{H}_{+}^{\eps}$, $\mathcal{H}_{-}^{\eps}$ does not depend on the left state. Let $(\rho_{\ell},\theta_{\ell})$ be a left state. Then:
	\begin{enumerate} 
	\item[(i)] Suppose $\theta_{r}$ is fixed. The function $\rho_{r} \rightarrow H_{\eps}(\rho_{\ell},\theta_{\ell},\rho_{r},\theta_{r})$ has at most two zeros and there exists $\eps' > 0$ such that for all $\eps < \eps'$, the function $\rho_{r} \rightarrow H_{\eps}(\rho_{\ell},\theta_{\ell},\rho_{r},\theta_{r})$ has only one positive zero. This zero tends to $\rho^{\ast}$ as $\eps$ tends to $0$.  
  \item[(ii)] Suppose $\rho_{r}$ is fixed. Then $\forall \eps > 0$, the function $\theta_{r} \rightarrow H_{\eps}(\rho_{\ell},\theta_{\ell},\rho_{r},\theta_{r})$ has two zeros, one lower and one larger than $\theta_{\ell}$, and both of them tend to $\theta_{\ell}$ as $\eps$ tends to $0$. 
  	\end{enumerate}  
  	The Hugoniot locus tends to the union of the straight lines $\left\{\theta = \theta_{\ell}\right\}$ and $\left\{\rho = \rho^{\ast}\right\}$.
\end{lemma}
Note that these results imply that the Hugoniot locus consists of two monotonous curves as functions of $\rho$ (otherwise $H$ with fixed $\theta_{r}$ would have more than two zeros). The local behaviour of the Hugoniot locus enables us to determine that the increasing curve is associated to the first eigenvalue $\lambda_{-}^{\eps}$ and the decreasing curve to the second eigenvalue $\lambda_{+}^{\eps}$.   
\begin{proof} (i) Let us fix the left state $(\rho_{\ell},\theta_{\ell})$ and the right angle $\theta_{r}$. So as to get a more readible proof, the function $\rho_{r} \rightarrow H_{\eps}(\rho_{\ell},\theta_{\ell},\rho_{r},\theta_{r})$ will be denoted by $H$ but its derivative will be denoted by a partial derivative $\partial_{\rho_{r}}H$. We look for the zero set of $H$ in the interval $\left]0,\rho^{\ast}\right[$. We compute:
\begin{equation*}
\frac{\partial^{2} H}{\partial\rho_{r}^{2}}(\rho_{r}) = \eps\left(2p'(\rho_{r}) + p''(\rho_{r})\rho_{r}\right)\quad > 0.
\end{equation*}
As $p$ and its first two derivatives are strictly positive on $\left]0,\rho^{\ast}\right[$, the function $H$ is strictly convex and thus has at most two zeros. Moreover, the value of $H$ at $\rho_{r} = \rho_{\ell}$ is strictly negative,
\begin{equation*}
H(\rho_{\ell}) = - \rho_{\ell}\left[\Psi(\cos\theta)\right]\left[\cos\theta\right]\quad < 0, 
\end{equation*}
if $\theta_{r}$ is not equal to $\theta_{\ell}$. Like the function $p$, $H$ tends to $+ \infty$ when $\rho_{r}$ tends to the maximal density $\rho^{\ast}$. Then $H$ has only one zero in $\left]\rho_{\ell},\rho^{\ast}\right[$. We have
\begin{equation*}
H(0) = \rho_{\ell}\left(\eps p(\rho_{\ell}) +  \left[\Psi(\cos\theta)\right]\cos\theta_{\ell} - \left[\Phi(\cos\theta)\right]\right).
\end{equation*}
Lemma \ref{fonctionfu} implies that the second term of this expression is strictly negative and thus $H(0)$ becomes stricly negative for small $\eps$. Thanks to its convexity, we deduce that there exists $\eps'$ such that for all $\eps < \eps'$, the function $H$ has no zero in the interval $\left]0,\rho_{\ell}\right[$. 

To show that the only zero of $H$ tends to $\rho^{\ast}$, let us rewrite $H$ as follows
\begin{equation*}
H(\rho_{r})  = \left[\eps p(\rho)\right]\left[\rho\right] + \left[\Phi(\cos\theta)\right]\left[\rho\right] - \left[\Psi(\cos\theta)\right]\left[\rho\right]\cos\theta_{r} - \left[\Psi(\cos\theta)\right]\left[\cos\theta\right]\rho_{\ell},
\end{equation*} 
and thanks to lemma \ref{fonctionfu}, the zero of $H$ satisfies
\begin{eqnarray*}
\left[\eps p(\rho)\right]\left[\rho\right]  &=&  - \left[\Phi(\cos\theta)\right]\left[\rho\right] + \left[\Psi(\cos\theta)\right]\left[\rho\right]\cos(\theta_{r}) + \left[\Psi(\cos\theta)\right]\left[\cos\theta\right]\rho_{\ell} \\&\geq& \left[\Psi(\cos\theta)\right]\left[\cos\theta\right]\rho_{\ell} > 0.
\end{eqnarray*} 
So we can easily conclude that the zero of $H$ tends to $\rho^{\ast}$.

(ii) Like in the first point, let us denote the function $\theta_{r} \rightarrow H_{\eps}(\rho_{\ell},\theta_{\ell},\rho_{r},.)$ by $H$. First, the value taken by $H$ at $\theta_{r} = \theta_{\ell}$ is positive:
\begin{equation*}
H(\theta_{\ell}) =  \eps \left[p(\rho)\right]\left[\rho\right] > 0.
\end{equation*}
Some easy computations leads to the following expression of the first (partial) derivative of $H$: 
\begin{equation*}
\frac{\partial H}{\partial \theta_{r}}(\theta_{r}) = \frac{1}{\sin\theta_{r}}\left( \rho_{\ell}\left[\cos\theta\right] + \left[\Psi(\cos\theta)\right]\rho_{r}\sin^{2}\theta_{r}\right).
\end{equation*}
As $\rho_{\ell}$ and $\rho_{r}$ are positive, the sign of the derivative is the same as the sign of $\left[\cos\theta\right]$. Thus $H$ is increasing on $[0,\theta_{\ell}]$ and decreasing on $[\theta_{\ell},\pi]$. Moreover using the fact that $\Psi(u) = \Phi(u) + \log(1 + u)$, we can write $H$ as
\begin{equation*}
H(\theta_{r}) =  \Phi(\cos\theta_{r})\left[\rho(1 - \cos\theta)\right] + \log(1 + \cos\theta_{r})\left[\rho\cos\theta\right] + A(\eps,\rho_{r},\theta_{r})
\end{equation*}
where $A$ is a bounded function. It implies that $H$ tends to $- \infty$ when $\theta_{r}$ tends to $0$. In the same way and by using the identity $\Psi(u) = - \Phi(u) + \log(1 - u)$, we can show that $H$ tends also to $- \infty$ when $\theta_{r}$ tends to $\pi$. We deduce that $H$ has exactly two zeros. 

Let us remark that 
\begin{equation*}
H_{\eps} = H_{1} - (1 - \eps)\left[p(\rho)\right]\left[\rho\right].
\end{equation*}  
This implies that $H_{\eps}^{-1}(0)  = H_{1}^{-1}((1 - \eps)\left[p(\rho)\right]\left[\rho\right])$ and then that the zeros of $H_{\eps}$ tend to $\theta_{\ell}$ as $\eps$ tends to $0$. 
\qed
\end{proof}   

\section{Proof of proposition \ref{Prop:Integral_curve} (study of the integral curves of the right eigenvectors)}
\label{Appendix:Integral_curve}

\begin{proof} $1.$ We easily check that $\theta'(\rho) = \mp \sqrt{\eps p'(\rho)/\rho}$, leading to the result.

$2.$ For a fixed $\rho$, the quantity
\begin{equation*}
\theta^{\eps} = (i_{\pm}^{\eps})^{-1}(\rho) = \theta_{\ell} \mp \sqrt{\eps}\left(\int_{\rho_{\ell}}^{\rho} \sqrt{\frac{\eps p'(u)}{u}}du\right)
\end{equation*}
 converges to $\theta_{\ell}$ as $\eps$ goes to $0$. For a fixed $\theta$, the quantity $\sqrt{\eps}\left(\int_{\rho_{\ell}}^{\rho^{\eps}} \sqrt{\frac{\eps p'(u)}{u}}du\right) = \theta - \theta_{\ell}$ is a constant. So as $\eps$ tends to $0$, the integral term has to tend to $+\infty$, which implies the convergence of $\rho^{\eps} = i_{+}^{\eps}(\theta)$ to $\rho^{\ast}$. Besides, the function inside the integral behaves like $O\left(\sqrt{\eps}(\rho^{\ast} -u)^{-\frac{\gamma+1}{2}}\right)$ when $\rho \rightarrow \rho^{\ast}$. This leads to a diverging integral for $\gamma > 1$. Then the integral behaves like $O\left(\sqrt{\eps}(\rho^{\ast} -\rho_{d}^{\eps})^{-\frac{\gamma-1}{2}}\right)$ and thus we get $\rho^{\ast} - i_{\pm}^{\eps}(\theta) = O\left(\eps^{\frac{1}{k-1}}\right)$.

$3.$ Let $\eps' > 0$ and $\rho < \rho_{r}^{\eps'}$. From the rarefaction curve equation (\ref{OndeDetente}), $(i^{\eps}_{\pm})^{-1}(\rho)$ satisfies
\begin{equation*}
|(i^{\eps}_{\pm})^{-1}(\rho) - \theta_{r}| \leq  \int_{0}^{\rho_{r}^{\eps}} \sqrt{\frac{\eps p'(u)}{u}}du = \sqrt{\eps}\int_{0}^{\rho_{r}^{\eps}} \frac{\gamma u^{\frac{\gamma-2}{2}}\rho^{\ast \gamma+1} }{(\rho^{\ast} -u)^{\frac{\gamma+1}{2}}}du,
\end{equation*}
 Assuming that the limit of $\eps p(\rho_{r}^{\eps})$ is finite, we get $\rho^{\ast} - \rho_{r}^{\eps} = O(\eps^{\frac{1}{\gamma}})$. Thus, the function inside the integral behaves like $\sqrt{\eps}(\rho^{\ast} -u)^{-\frac{\gamma+1}{2}}$ when $\rho \rightarrow \rho^{\ast}$. This leads to a diverging integral for $\gamma > 1$, and then the integral behaves like $O\left(\sqrt{\eps}(\rho^{\ast} -\rho_{r}^{\eps})^{-\frac{\gamma-1}{2}}\right)$ and thus like $O(\eps^{\frac{1}{2\gamma}})$. 
\qed
\end{proof}

\section{Proofs of theorem \ref{Thm:Riemann_Small_Epsilon} and proposition \ref{Prop:Riemann_sign} (solutions of the Riemann problem for $\eps > 0$)}

\subsection{Proof of theorem \ref{Thm:Riemann_Small_Epsilon}}
\label{Appendix:Riemann_Small_Epsilon}

\begin{proof} Let $(\rho_{\ell},\theta_{\ell})$ and $(\rho_{r},\theta_{r})$ be left and right states respectively and let us suppose that the intersection of the 1-forward wave curve $W_{-}^{f,\eps}$ issued from the left state and the 2-backward wave curve $W_{+}^{b,\eps}$ issued from the right state reduces to one point $(\widetilde{\rho},\widetilde{\theta})$ (in all the proof, the 1-wave will be implicitly relative to the left state while the 2-wave curve will be implicitly relative to the right state). The solution of the Riemann problem depends on which parts of the wave curves meet: for instance, if $(\widetilde{\rho},\widetilde{\theta})$ is the intersection of the 1-shock curve with the 2-rarefaction curve, then the solution will be the combination of a shock wave (with a speed given by the Rankine-Hugoniot condition (\ref{Eq:RH_1})) and a rarefaction wave separated by an intermediate constant state. To find where the intersection on the shock curves is located, the main arguments will be the monotony of the wave curves given by propositions \ref{Prop:Hugoniot_loci} and \ref{Prop:Integral_curve} (independently of the location of the states in the $\left(\rho,\theta\right)$-plane) and their convergence speed to their asymptotic limit. The wave curves will be considered as functions of $\theta$ in their domain of definition: 
\begin{equation*}
w_{-}^{\eps} = \left\{\begin{array}{ll}i_{-}^{\eps} &\mbox{ for }\theta \in [(i_{-}^{\eps} )^{-1}(0),\theta_{\ell}],\\h_{-}^{\eps} &\mbox{ for }\theta \in [\theta_{\ell},\pi[,\end{array}\right.\quad w_{+}^{\eps}  = \left\{\begin{array}{ll}h_{+}^{\eps}&\mbox{ for }\theta \in [0,\theta_{r}],\\i_{+}^{\eps} &\mbox{ for }\theta \in [\theta_{r},(i_{+}^{\eps} )^{-1}(0)[,\end{array}\right.
\end{equation*}
where $w_{-}$ is an increasing function and $w_{+}$ is a decreasing function. The functions $h_{\pm}$ and $i_{\pm}$ are respectively defined in propositions \ref{Prop:Hugoniot_loci} and \ref{Prop:Integral_curve}. Let us examine the different cases suggested by the theorem successively: $\theta_{\ell}$ greater or lower or equal to $\theta_{r}$. For the reader's convenience, the corresponding geometric configurations of the wave curves are illustrated in figure \ref{Fig_appendix:Riemann_Small_Epsilon}. 

\textbf{Case $\theta_{\ell} > \theta_{r}$} (Fig. \ref{Fig_appendix:Riemann_Small_Epsilon}, (a)). From proposition \ref{Prop:Integral_curve}, $(i_{-}^{\eps})^{-1}(0)$ (resp. $(i_{+}^{\eps})^{-1}(0)$) tends to $\theta_{\ell}$ (resp. $\theta_{r}$) as $\eps$ goes to zero (and the third point of the same proposition asserts that it is still the case when $\rho_{r}$ tends to $\rho^{\ast}$). So, assuming that there exists $\alpha$ such that $(i_{-}^{\alpha})^{-1}(0) < (i_{+}^{\eps})^{-1}(0)$, there exists $\beta < \alpha$ such that $\theta_{r} < (i_{-}^{\beta})^{-1}(0) = (i_{+}^{\beta})^{-1}(0) < \theta_{\ell}$. So, since the domains of definition of $w_{-}$ and $w_{+}$ are respectively $[(i_{-}^{\beta})^{-1}(0),\pi[$ and $]0,(i_{+}^{\beta})^{-1}(0)]$, the only intersection point of $w_{+}^{\beta}$ and $w_{-}^{\beta}$ is the intersection of the 1 and 2-rarefaction curves at $(i_{+}^{\beta})^{-1}(0)$. As $\eps$ decreases, the intersection point disappears since the domains of definition are separated. However, the integral curves meet the $\left\{\rho = 0\right\}$ axis at the states $(0,i_{-}^{\eps}(0))$, $(0,i_{+}^{\eps}(0))$ and these states are connected by vacuum.
\begin{figure}
\begin{center}

\null
\hfill
\psfrag{rho}{$\rho$}
\psfrag{theta}{$\theta$}
\psfrag{rhomax}{$\rho^{\ast}$}
\psfrag{pi}{$\pi$}
\psfrag{0}{$0$}
\psfrag{rhol}{\textcolor{blue}{$\rho_{\ell}$}}
\psfrag{rhor}{\textcolor{blue}{$\rho_{r}$}}
\psfrag{thetal}{\textcolor{blue}{$\theta_{\ell}$}}
\psfrag{thetar}{\textcolor{blue}{$\theta_{r}$}}
\psfrag{W+}{$i_{+}$}
\psfrag{W-}{$i_{-}$}
\psfrag{W+(0)}{\footnotesize $(i_{+})^{-1}(0)$}
\psfrag{W-(0)}{\footnotesize $(i_{-})^{-1}(0)$}
\subfigure[$\theta_{\ell} > \theta_{r}$]{\includegraphics[width=0.45\textwidth]{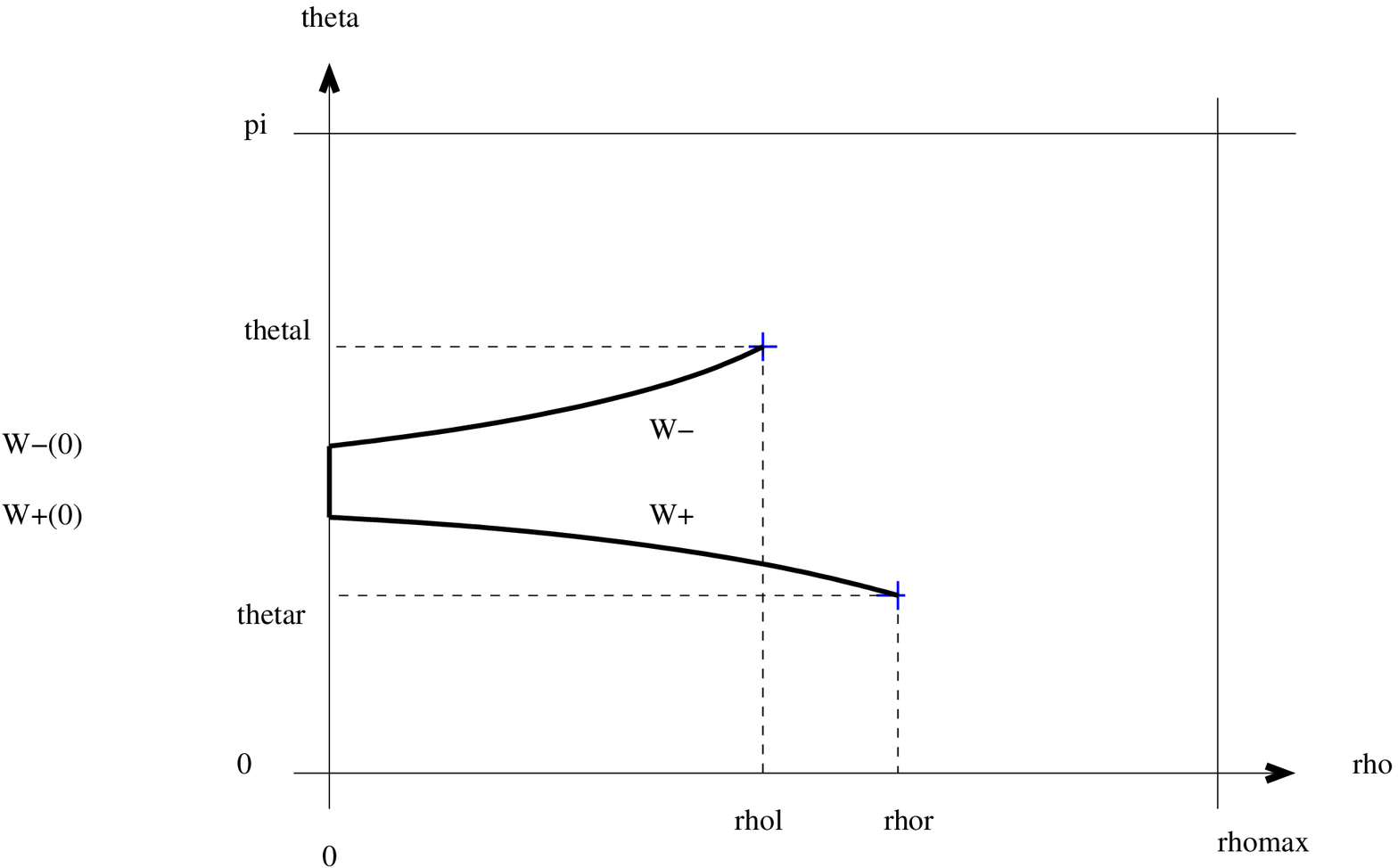}}
\hfill
\psfrag{rho}{$\rho$}
\psfrag{theta}{$\theta$}
\psfrag{rhomax}{$\rho^{\ast}$}
\psfrag{pi}{$\pi$}
\psfrag{0}{$0$}
\psfrag{rhol}{\textcolor{blue}{$\rho_{\ell}$}}
\psfrag{rhor}{\textcolor{blue}{$\rho_{r}$}}
\psfrag{thetal}{\textcolor{blue}{$\theta_{\ell}$}}
\psfrag{thetar}{\textcolor{blue}{$\theta_{r}$}}
\psfrag{W+}{$h_{-}$}
\psfrag{W-}{$i_{+}$}
\psfrag{hl(rhor)}{\footnotesize $h_{-}(\theta_{r})$}
\subfigure[$\theta_{\ell} < \theta_{r}$, $\rho_{r} > h_{-}(\theta_{r})$]{\includegraphics[width=0.45\textwidth]{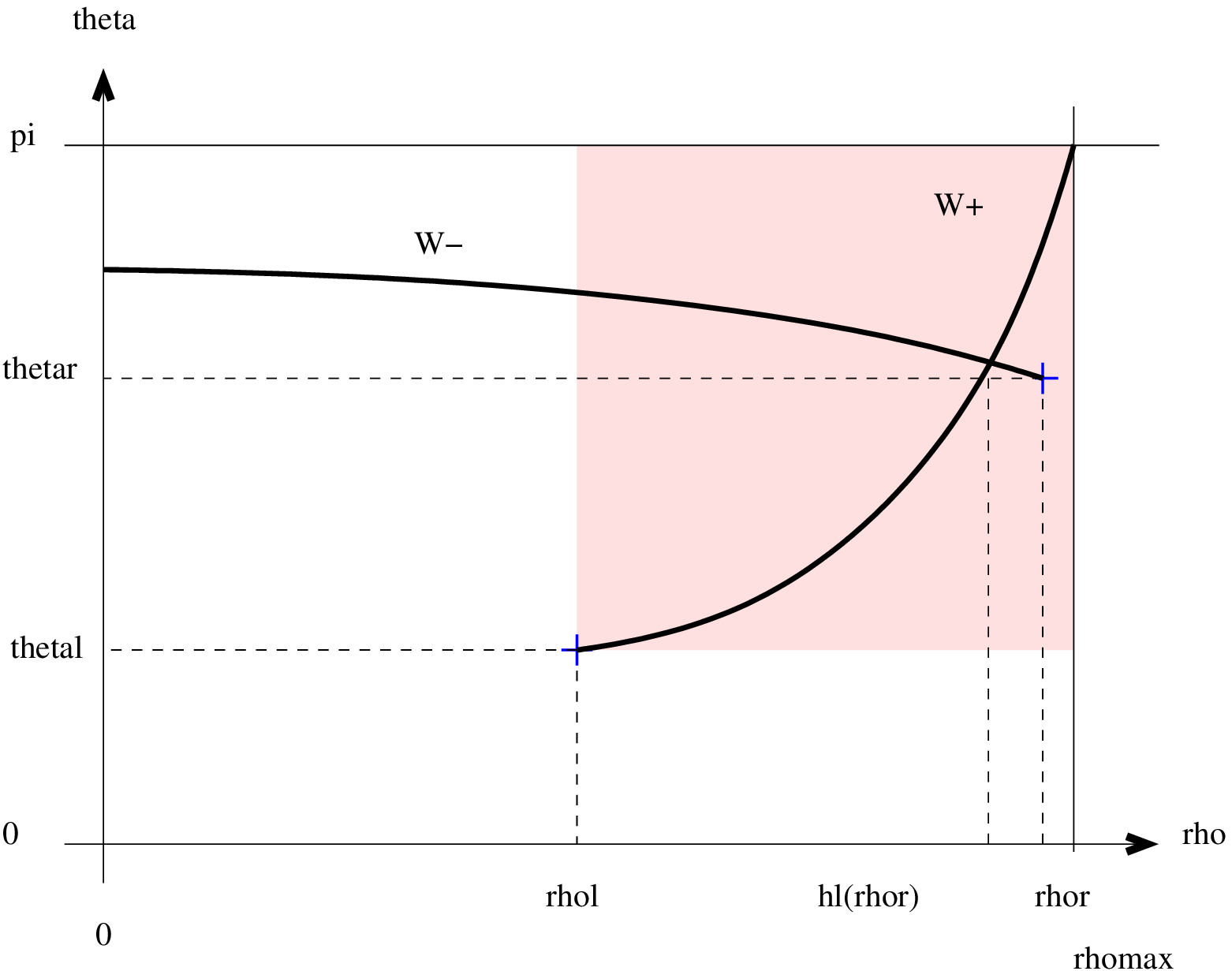}}
\hfill
\hfill
\null

\null
\hfill
\psfrag{rho}{$\rho$}
\psfrag{theta}{$\theta$}
\psfrag{rhomax}{$\rho^{\ast}$}
\psfrag{pi}{$\pi$}
\psfrag{0}{$0$}
\psfrag{rhol}{\textcolor{blue}{$\rho_{\ell}$}}
\psfrag{rhor}{\textcolor{blue}{$\rho_{r}$}}
\psfrag{thetal}{\textcolor{blue}{$\theta_{\ell}$}}
\psfrag{thetar}{\textcolor{blue}{$\theta_{r}$}}
\psfrag{W+}{$h_{-}$}
\psfrag{W-}{$h_{+}$}
\psfrag{hl(rhor)}{\footnotesize $h_{-}(\theta_{r})$}
\subfigure[$\theta_{\ell} < \theta_{r}$, $\rho_{r} < h_{-}(\theta_{r})$]{\includegraphics[width=0.45\textwidth]{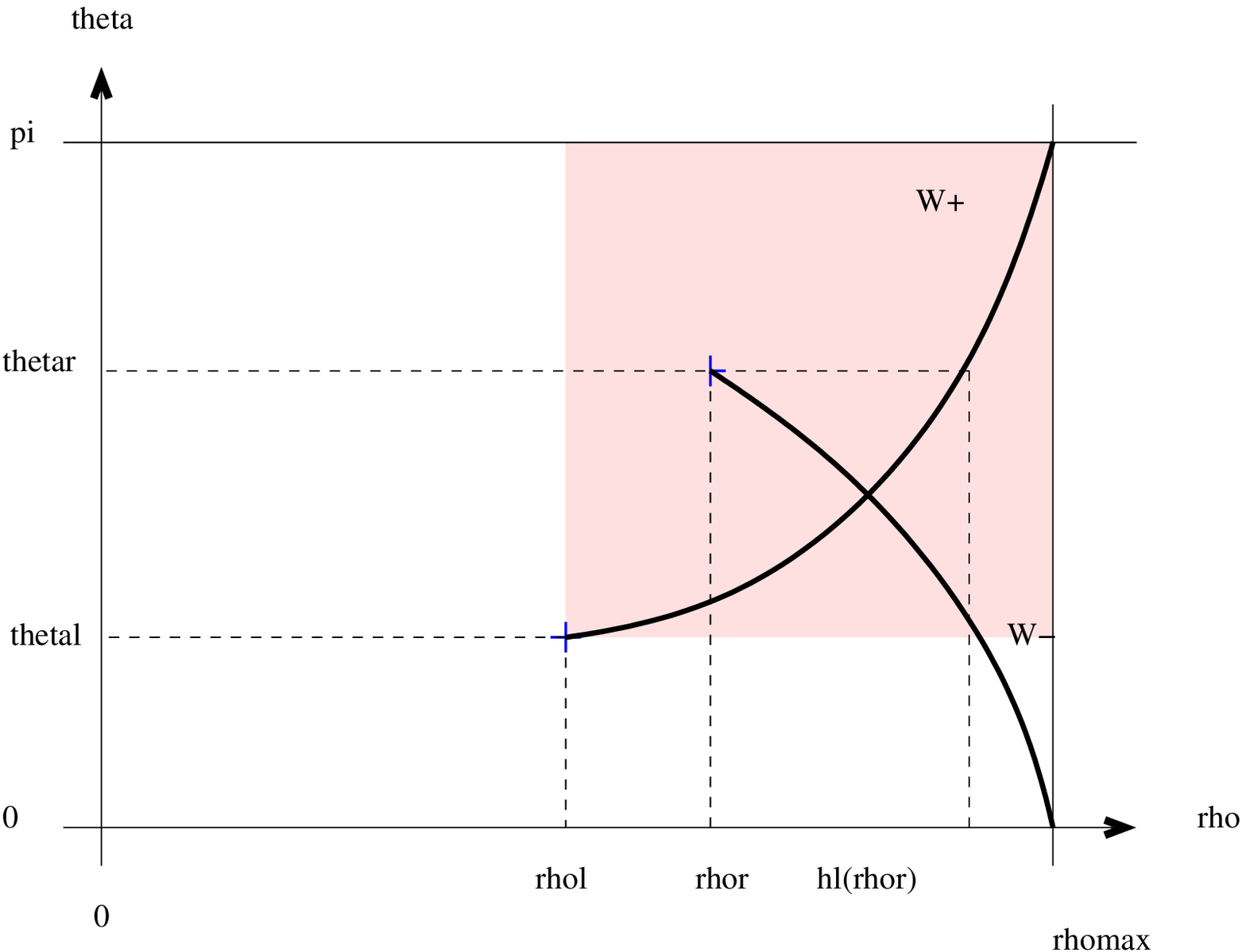}}
\hfill
\psfrag{rho}{$\rho$}
\psfrag{theta}{$\theta$}
\psfrag{rhomax}{$\rho^{\ast}$}
\psfrag{pi}{$\pi$}
\psfrag{0}{$0$}
\psfrag{rhol}{\textcolor{blue}{$\rho_{\ell}$}}
\psfrag{rhor}{\textcolor{blue}{$\rho_{r}$}}
\psfrag{thetal}{\textcolor{blue}{$\theta_{\ell} = \theta_{r}$}}
\psfrag{W+}{$h_{+}$}
\psfrag{W-}{$i_{-}$}
\psfrag{W-(0)}{\footnotesize $(i_{-})^{-1}(0)$}
\subfigure[$\theta_{\ell} = \theta_{r}$]{\includegraphics[width=0.45\textwidth]{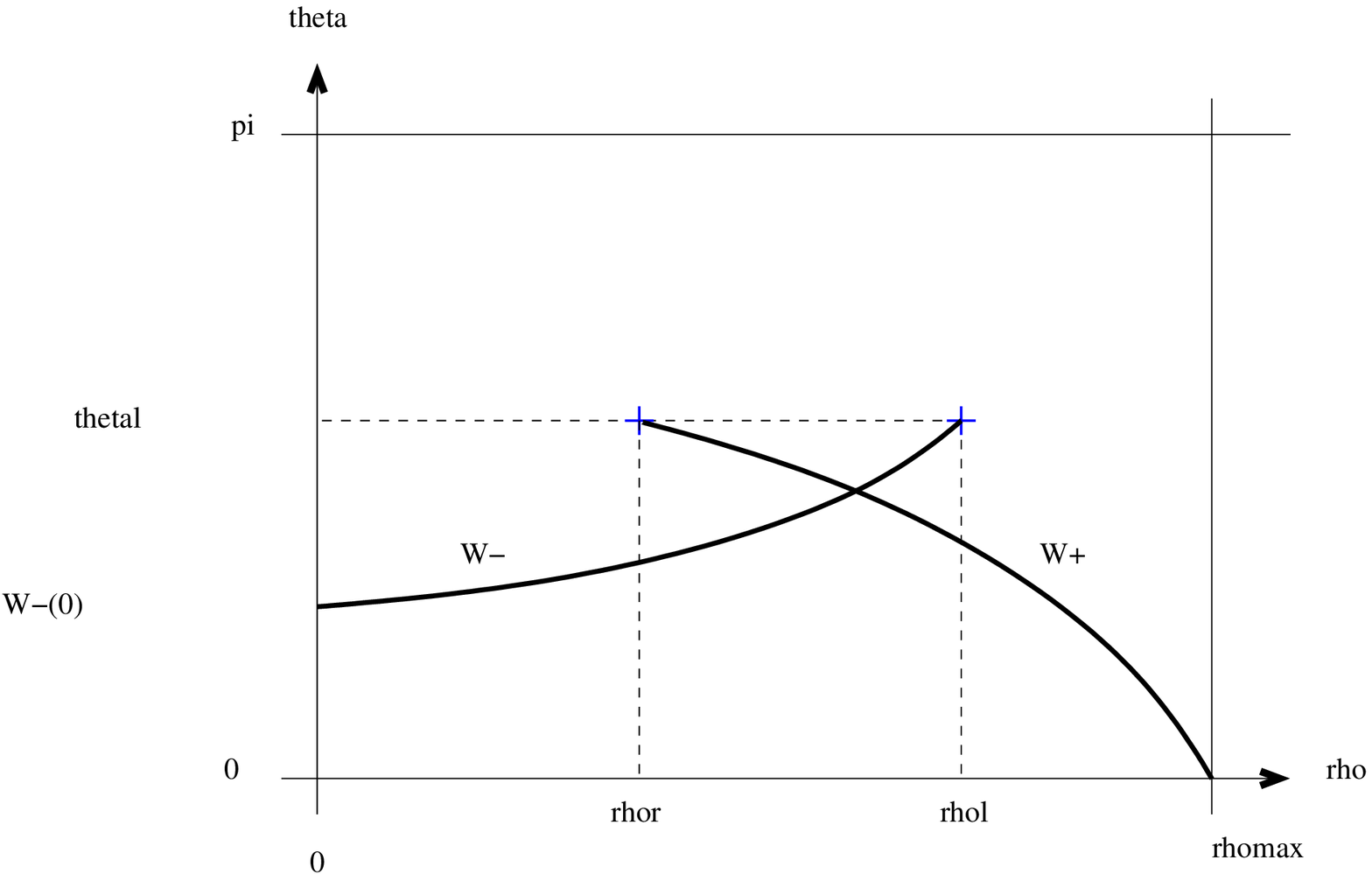}}
\hfill
\null

\caption{Schematics of the intersections of the wave curves (proof of theorem \ref{Thm:Riemann_Small_Epsilon}). Only the parts which meet are represented.}
\label{Fig_appendix:Riemann_Small_Epsilon}
\end{center}
\end{figure}

\textbf{Case $\theta_{\ell} < \theta_{r}$.} We suppose that $\rho_{\ell}$ is lower than $\rho_{r}$. For all $\eps$, the increasing 1-shock curve issued from the left state divides the domain $[\rho_{\ell},\rho^{\ast}]\times[\theta_{\ell},2\pi]$ (to which the right state belongs) in two parts: the left domain where the right state is on the left side of the 1-Hugoniot curve issued from the left state ($\rho_{r}^{\eps} < h_{-}^{\eps}(\theta_{r}^{\eps})$) and the right domain where the right state is on the right side of the 1-Hugoniot curve issued from the left state ($\rho_{r}^{\eps} > h_{-}^{\eps}(\theta_{r}^{\eps})$).

 - Assume that for all $\eps$ the right state is on the right side of the 1-shock curve (Fig. \ref{Fig_appendix:Riemann_Small_Epsilon}, (b)). We consider $w_{-}^{\eps} - w_{+}^{\eps}$ on the interval $[(i_{-}^{\eps})^{-1}(0),(i_{+}^{\eps})^{-1}(0)]$, which is the intersection of the domains of definition of $w_{-}^{\eps}$ and $ w_{+}^{\eps}$. The function $w_{-}^{\eps} - w_{+}^{\eps}$ is increasing. We have $(w_{-}^{\eps} - w_{+}^{\eps})((i_{-}^{\eps})^{-1}(0)) = - (w_{+})((i_{-})^{-1}(0)) < 0$ and  $(w_{-}^{\eps} - w_{+}^{\eps})(\theta_{r}) =  h_{-}^{\eps}(\theta_{r}) - \rho_{r} > 0$. So the only zero of $w_{-}^{\eps} - w_{+}^{\eps}$ is in the interval $[(i_{-}^{\eps})^{-1}(0),\theta_{r}]$. So the intersection point of the two wave curves is the intersection of the 1-shock curve and the 2-rarefaction curve. This corresponds to the second subcase of the third case of the theorem. Note that the limit of the 1-shock curve (proposition \ref{Prop:Hugoniot_loci}) implies that the limit of this case should be considered only if $\rho_{r}$ tends to $\rho^{\ast}$.  

- Assume that for all $\eps$ the right state is on the left side of the 1-shock curve (Fig. \ref{Fig_appendix:Riemann_Small_Epsilon}, (c)). Since $\rho_{\ell}$ is lower than $\rho_{r}$, the left state is also on the left side of the 2-Hugoniot curve issued from the right state: $\rho_{\ell}^{\eps} < \rho_{r}^{\eps} =   h_{+}^{\eps}(\theta_{r}^{\eps}) < h_{+}^{\eps}(\theta_{\ell}^{\eps})$. We again consider the increasing function $w_{-}^{\eps} - w_{+}^{\eps}$ on its domain of definition $[(i_{-}^{\eps})^{-1}(0),(i_{+}^{\eps})^{-1}(0)]$. This function is negative at $\theta_{\ell}$ and positive at $\theta_{r}$. So the intersection point of the two wave curves is the intersection of the two  shock curves and $\widetilde{\rho} > \rho_{\ell},\rho_{r}$. This corresponds to the first subcase of the third case of the theorem.

If $\rho_{\ell}$ is greater than $\rho_{r}$, the decreasing 2-shock curve issued from the right state divides the domain $[\rho_{r},\rho^{\ast}]\times[0,\theta_{r}]$ and the same arguments as before lead to the result. 

\textbf{Case $\theta_{\ell} = \theta_{r}$} (Fig. \ref{Fig_appendix:Riemann_Small_Epsilon}, (d)). Assume that $\rho_{r} < \rho_{\ell}$. We again consider the increasing function  $w_{-}^{\eps} - w_{+}^{\eps}$. It is positive at $\theta = \theta_{\ell}$ and negative for $\theta = (i_{-}^{\eps})^{-1}(0) < \theta_{\ell}$ (since $h_{+}^{\eps}\left((i_{-}^{\eps})^{-1}(0)\right) > \rho_{r} > 0$). So it equals zero for a value $\theta < \theta_{r}$. So $\left(\widetilde{\rho},\widetilde{\theta}\right)$ is the intersection of the 1-rarefaction curve and the 2-shock curve, which leads to the solution given in the first case of the theorem. The case  $\rho_{r} > \rho_{\ell}$ is similar. 
\qed
\end{proof}

\subsection{Proof of proposition \ref{Prop:Riemann_sign}}
\label{Appendix:Riemann_sign}

\begin{proof} 1. Consider the domain where $\left\{\rho =\rho_{\ell}\right\}$. 

If the left state is connected to the intermediate state via a rarefaction wave, then this rarefaction fan is contained between the speeds $\lambda_{-}^{\eps} = \cos\theta_{\ell} - \sqrt{\eps p'(\rho_{\ell})\rho_{\ell}}|\sin\theta_{\ell}|$ and $\tilde{\lambda}_{-}^{\eps} = \cos\tilde{\theta} - \sqrt{\eps p'(\tilde{\rho})\tilde{\rho}}|\sin\tilde{\theta}|$. Since $\lambda_{-}^{\eps} < \cos\theta_{\ell}$, the domain $\left\{\rho =\rho_{\ell}\right\}$ cannot contain the contact wave with speed $\cos\theta_{\ell}$.  

If the left state is connected to the intermediate state via a shock wave, then we have $\cos\tilde{\theta} < \cos\theta_{\ell}$ and $\rho_{\ell} < \tilde{\rho}$, which yields
\begin{eqnarray*}
&&s^{-} = \cos\theta_{\ell} + \frac{\tilde{\rho}}{\rho_{\ell} - \tilde{\rho}}(\cos\theta_{\ell} - \cos\tilde{\theta})\\
&&\quad = \cos\tilde{\theta} + \frac{\rho_{\ell}}{\tilde{\rho} - \rho_{\ell}}(\cos\tilde{\theta} - \cos\theta_{\ell})\quad  < \cos\tilde{\theta},
\end{eqnarray*}
So the domain $\left\{\rho =\rho_{\ell}\right\}$ cannot contain the contact wave with speed $\cos\theta_{\ell}$.  

So in both cases, the domain $\left\{\rho =\rho_{\ell}\right\}$ cannot contain the contact wave with speed $\cos\theta_{\ell}$. The same arguments show that the domain $\left\{\rho =\rho_{r}\right\}$ cannot contain the contact wave with speed $\cos\theta_{r}$.

Finally, we easily check that a contact wave with propagation speed $\cos\tilde{\theta}$ can occur within the intermediate domain $\left\{\rho = \tilde{\rho}\right\}$. Indeed, if the intermediate state is connected to the left state (resp. to the right state) via a rarefaction wave, then we have $\tilde{\lambda}_{-}^{\eps} < \cos\tilde{\theta} < \tilde{\lambda}_{+}^{\eps}$ and if the intermediate state is connected to the left state (resp. to the right state) via a shock wave then $s^{-} < \cos\tilde{\theta}$ (resp. $s^{+} > \cos\tilde{\theta}$).

2. Like in the previous point, the contact wave can be located only in the intermediate state $\left\{\rho = 0\right\}$. But the propagation speed is not unique: it can be all the intermediate speeds between the two fans of rarefaction.   
\qed
\end{proof}

\section{Proofs of lemma \ref{Lemma:Rarefaction_wave_2} and propositions \ref{Prop:limit_Riemann_1}, \ref{Prop:limit_Riemann_2} and \ref{Prop:limit_Riemann_3} (limits of solutions of the Riemann problem)}

We recall that the quantities indexed by - (resp. by +) are implicitly those related to the left state (resp. the right state). The characteristic speeds related to the intermediate state will be denoted by $\widetilde{\lambda}_{\pm}^{\eps}$.  

\subsection{Proof of proposition \ref{Prop:limit_Riemann_1}}
\label{Appendix:limit_Riemann_1}

\begin{proof} (a) Let us suppose that $\rho_{\ell} < \rho_{r}$ (the opposite case is similar). According to theorem \ref{Thm:Riemann_Small_Epsilon} and proposition \ref{Prop:Integral_curve}, the intermediate state angle $\widetilde{\theta}^{\eps}$ tends to $\theta_{\ell}$ (since $\widetilde{\rho}^{\eps}$ belongs to the interval $\left]\rho_{\ell},\rho_{r}\right[$ for each $\eps$). In addition, it is easy to check that the two speeds $\widetilde{\lambda}_{+}^{\eps}$ and $\lambda_{+}$ tend to $\cos\theta_{r}$ as $\eps$ tends to $0$. So the rarefaction wave turns into a contact wave with speed $\cos\theta_{r}$. As regards the shock wave, its speed is given by
\begin{equation*}
s^{\eps} = \frac{\widetilde{\rho}^{\eps}\cos\widetilde{\theta}^{\eps} - \rho_{\ell}\cos\theta_{\ell}}{\widetilde{\rho}^{\eps} - \rho_{\ell}} = \cos\theta_{\ell} + \widetilde{\rho}^{\eps}\frac{\cos\widetilde{\theta}^{\eps} - \cos\theta_{\ell}}{\widetilde{\rho}^{\eps} - \rho_{\ell}},
\end{equation*}
and tends to $\cos\theta_{\ell}$ as $\eps$ tends $0$. Finally, the two waves coincide and make a single contact wave.

(b) As in the first case, it is easy to check that the two rarefaction waves (cf. theorem \ref{Thm:Riemann_Small_Epsilon}) turn into contact waves with speeds respectively equal to $\cos\theta_{r}$ and $\cos\theta_{\ell}$.

(c) According to proposition \ref{Prop:Hugoniot_loci}, $h_{-}^{\eps}(\rho_{r}^{\eps})$ tends to $\theta_{\ell} < \theta_{r}$. So, theorem \ref{Thm:Riemann_Small_Epsilon} implies that we are looking for the limit of the intersection point of the two shock curves issued from the left and the right states. These intermediate states $(\widetilde{\rho}^{\eps},\widetilde{\theta}^{\eps})$ are the solutions of the non-linear systems
\begin{eqnarray}
&&H_{\eps}(\rho_{\ell},\theta_{\ell}, \widetilde{\rho}^{\eps},\widetilde{\theta}^{\eps}) = 0,\ H_{\eps}(\widetilde{\rho}^{\eps},\widetilde{\theta}^{\eps},\rho_{r},\theta_{r}) = 0,\\
&&\rho \geq \max(\rho_{\ell},\rho_{r}),\ \theta \in \left]\min(\theta_{\ell},\theta_{r}),\max(\theta_{\ell},\theta_{r})\right[. 
\label{SystEq:Intersection_shock}
\end{eqnarray}
From proposition \ref{Prop:Hugoniot_loci}, for all $\theta \in \left]\min(\theta_{\ell},\theta_{r}),\max(\theta_{\ell},\theta_{r})\right[$, the largest zero of the function $\rho \rightarrow H_{\eps}(\rho_{\ell},\theta_{\ell},\rho,\theta)$ tends to $\rho^{\ast}$ as $\eps \rightarrow 0$. Indeed, if $\widetilde{\rho}^{\eps}$ does not tend to $\rho^{\ast}$, $\widetilde{\theta}^{\eps}$ simultaneously tend to $\theta_{\ell}$ and to $\theta_{r}$ (which is different from $\theta_{\ell}$), which is absurd. Therefore, $\widetilde{\rho}^{\eps}$ increases and tends to $\rho^{\ast}$. Besides, we have the equality 
\begin{equation*}
\eps p(\widetilde{\rho}^{\eps})[\rho]_{\ell} = [\Psi(\cos\theta)]_{\ell}[\rho\cos\theta]_{\ell} - [\Phi(\cos\theta)]_{\ell}[\rho]_{\ell} + \eps p(\rho^{\ell})[\rho]_{\ell}, 
\end{equation*}
which implies that $\eps p(\widetilde{\rho}^{\eps})$ is bounded as $\eps$ tends to $0$ (because $\widetilde{\theta}^{\eps}$ is bounded too) and we deduce that $\eps p(\widetilde{\rho}^{\eps})$ converges to a non-zero value $\bar{p}$. Finally, we can easily check that the system given in  is equivalent to (\ref{SystEq:Intersection_shock}).
\qed
\end{proof}

\subsection{Proof of lemma \ref{Lemma:Rarefaction_wave_2}}
\label{Appendix:Rarefaction_wave_2}

\begin{proof} Suppose that $\widetilde{\lambda}_{+} = \lim \widetilde{\lambda}_{+}^{\eps}$ is finite. Since $\eps p(\rho_{r}^{\eps}) \rightarrow \bar{p}_{r} > 0$, then $\eps p'(\rho_{r}^{\eps}) \rightarrow +\infty$ and consequently $\lambda_{+} \rightarrow +\infty$. The limit rarefaction wave has a fan for speeds $s$ belonging to $]\widetilde{\lambda}_{+},+ \infty[$. The 2-rarefaction wave satisfies, for all $s \in ]\widetilde{\lambda}_{+},+ \infty[$, 
\begin{equation*}
s = \lambda_{+}(\rho(s),\theta(s)) = \cos(\theta(s)) + \sqrt{\eps p'(\rho(s))\rho(s)}|\sin(\theta(s))|.
\end{equation*} 
So, for a fixed $s$, $\rho(s) \rightarrow \rho^{\ast}$ as $\eps \rightarrow 0$ and we have $(\rho^{\ast} - \rho(s)) = O(\eps^{\frac{1}{\gamma + 1}})$. Thus, $\eps p(\rho(s))$ ($= O(\eps^{1 - \frac{\gamma}{\gamma+1}})$)  $\rightarrow 0$ as $\eps \rightarrow 0$. So the rarefaction wave tends to the combination of a shock wave between the states $(\widetilde{\rho},\theta_{r},0)$ and $(\rho^{\ast},\theta_{r},0)$ with speed $\widetilde{\lambda}_{+}$ and a declustering wave. 

If $\widetilde{\rho} < \rho^{\ast}$, then $\widetilde{\lambda}_{+}$ equals $0$ and the previous arguments apply. Let us look at the case $\widetilde{\rho} = \rho^{\ast}$. If $\widetilde{\lambda}_{+}$ is finite, then in the previous conclusion the shock wave disappears since the two states on both sides of the shock wave are equal. And it confirms that $\bar{\bar{p}} = \lim \eps p(\widetilde{\rho}^{\eps})$ equals zero. If  $\widetilde{\lambda}_{+}$ is infinite like $\lim \lambda_{+}^{\eps}$, then the rarefaction wave turns into a shock wave with an infinite speed between the states $(\rho^{\ast},\theta_{r},\bar{p}_{\ell})$ and $(\rho^{\ast},\theta_{r},\bar{p}_{r})$.~\qed
\end{proof}

\subsection{Proof of proposition \ref{Prop:limit_Riemann_2}}
\label{Appendix:limit_Riemann_2}

\begin{proof} 

(a) We want to apply lemma \ref{Lemma:Rarefaction_wave_2}. Here the intermediate state is the intersection of the 2-rarefaction curve and the 1-shock curve. The intersection state exists for all $\eps$ by the monotony of the two curves (cf. propositions \ref{Prop:Hugoniot_loci} and \ref{Prop:Integral_curve}). Let us note that there is no reason to have a finite limit of $\widetilde{\lambda}^{\eps}$ since the intermediate state can tend to $\rho^{\ast}$. By a compatcness argument, we can restrict ourselves to prove the uniqueness of the limit of convergent solutions. So let us consider several cases:

Case (i) $\widetilde{\rho^{\eps}} \rightarrow \rho_{\ell}$. In this case, the 1-shock disapears and the solution is given by lemma \ref{Lemma:Rarefaction_wave_2}.

Case (ii) $\widetilde{\rho^{\eps}} \rightarrow \widetilde{\rho} \in \left]\rho_{\ell},\rho^{\ast}\right[$. In this case, it is easy to check that the 1-shock becomes a 1-contact and the limit of the 2-rarefaction is given by lemma \ref{Lemma:Rarefaction_wave_2}.

Case (iii) $\widetilde{\rho^{\eps}} \rightarrow \rho^{\ast}$. Now let us look at the limit of $\eps p(\widetilde{\rho}^{\eps})$. We have
\begin{equation*}
H_{\eps}(\rho_{\ell},\theta_{\ell},\widetilde{\rho}^{\eps},\widetilde{\theta}^{\eps}) = \left[\Phi(\cos\theta) + \eps p(\rho)\right]_{\ell}\left[\rho\right]_{\ell} - \left[\Psi(\cos\theta)\right]_{\ell}\left[\rho\cos\theta\right]_{\ell} = 0.
\end{equation*}
Since $\cos\tilde{\theta}$ tends to $\cos\theta_{r}$ (cf. third point of proposition \ref{Prop:Integral_curve}), the terms $[\rho\cos\theta]$ and $[\rho]$ are bounded and taking the limit $\eps \rightarrow 0$, we get 
\begin{equation*}
\left[\eps p(\rho)\right]_{\ell}\left[\rho\right]_{\ell} \rightarrow 0
\end{equation*}
So either $\left[\rho\right]_{\ell}$ tends to $0$ or $\eps p(\widetilde{\rho}^{\eps})$ tends to $0$. Thus $\eps p(\widetilde{\rho}^{\eps})$ tends to $0$. Finally lemma \ref{Lemma:Rarefaction_wave_2} applies and we conclude that the 2-rarefaction tends to a declustering wave. Now let us look at the limit of the shock speed. It is written  
\begin{equation*}
s^{\eps} = \frac{\widetilde{\rho}^{\eps}\cos\widetilde{\theta}^{\eps} - \rho_{\ell}\cos\theta_{\ell}}{\widetilde{\rho}^{\eps} - \rho_{\ell}} = \cos\theta_{\ell} + \widetilde{\rho}^{\eps}\frac{\cos\widetilde{\theta}^{\eps} - \cos\theta_{\ell}}{\widetilde{\rho}^{\eps} - \rho_{\ell}}.
\end{equation*}
and so tends to $\cos\theta_{r}$. So the 1-shock tends to a contact discontinuity.

(b) The limit of the 2-rarefaction is given by lemma \ref{Lemma:Rarefaction_wave_2} (with $\widetilde{\rho}^{\eps} = 0$ and $\lambda_{+} = \cos\theta_{r}$).The 1-rarefaction wave turns into a contact wave as before. 

(c) We first consider the case where the intermediate state is the intersection point of two shock curves (for all $\eps$, $h_{-}^{\eps}(\theta_{r}^{\eps}) > \rho_{r}^{\eps}$): by the monotony of these curves, $\widetilde{\rho}^{\eps}$ is larger than $\rho_{r}^{\eps}$ and so tends to $\rho^{\ast}$ too. Besides, the intermediate state $(\widetilde{\rho}^{\eps},\widetilde{\theta}^{\eps})$ satisfies 
\begin{eqnarray*}
&&H_{\eps}(\rho_{\ell},\theta_{\ell},\widetilde{\rho}^{\eps},\widetilde{\theta}^{\eps}) = \left[\Phi(\cos\theta) + \eps p(\rho)\right]_{\ell}\left[\rho\right]_{\ell} - \left[\Psi(\cos\theta)\right]_{\ell}\left[\rho\cos\theta\right]_{\ell} = 0,\\
&&\widetilde{\rho}^{\eps} \geq \max(\rho_{\ell},\rho_{r}),\quad \widetilde{\theta}^{\eps} \in \left]\theta_{\ell},\theta_{r}\right[,\label{Eq:RH_left}
\end{eqnarray*}
which implies that $\eps p(\widetilde{\rho}^{\eps})$ converges to a value denoted $\bar{\bar{p}}$. By taking the limit $\eps \rightarrow 0$ in the Rankine-Hugoniot relation (\ref{Eq:RH_2}), we obtain
\begin{equation*}
\left[\Psi(\cos\theta)\right]_{r}[\rho\cos\theta]_{r} = \left[\Phi(\cos\theta) + \eps p(\rho)\right]_{r}[\rho]_{r},
\end{equation*}
we obtain 
\begin{equation*}
\left[\Psi(\cos\theta)\right]_{r}[\cos\theta]_{r} =  0
\end{equation*}
($\tilde{\theta}$ is bounded and $[\rho]_{r}$ tends to $0$), which implies that $\tilde{\theta}$ equals $\theta_{r}$. Finally, we have 
\begin{equation*}
[\eps p(\rho)]_{r} = [\Psi(\cos\theta)]_{r} \frac{[\rho \cos\theta]_{r}}{[\rho]_{r}} - [\Phi(\cos\theta)]_{r}. 
\end{equation*}
If the limit $[\eps p(\rho)]_{r}$  is non zero, then the propagation speed is infinite and if it is zero, there is no discontinuity.
    
Consider now the case where the intermediate state is the intersection of the 1-shock curve issued from the left state and the 2-rarefaction curve issued from the right state (for all $\eps$, $h_{-}^{\eps}(\theta_{r}^{\eps}) < \rho_{r}^{\eps}$). From proposition \ref{Prop:Integral_curve} (point 3), the intermediate angle $\widetilde{\theta}^{\eps}$ tends to $\theta_{r}$. Besides, thanks to proposition \ref{Prop:Hugoniot_loci}, the intermediate density $\widetilde{\rho}^{\eps}$ tends to $\rho^{\ast}$. From (\ref{Eq:Shock}), $\eps p(\widetilde{\rho}^{\eps})$ converges to a value denoted by $\bar{\bar{p}}$ which is given by the limit Rankine Hugoniot relation. So the limit intermediate state is $(\rho^{\ast},\theta_{r},\bar{\bar{p}})$. Finally lemma \ref{Lemma:Rarefaction_wave_2} applies: the rarefaction turns into a shock.  
\qed
\end{proof}

\subsection{Proof of proposition \ref{Prop:limit_Riemann_3}}
\label{Appendix:limit_Riemann_3}

\begin{proof} (a) Since the intermediate density $\widetilde{\rho}^{\eps}$ is comprised between the left and right ones: $\rho_{\ell}^{\eps} < \widetilde{\rho}^{\eps} < \rho_{r}^{\eps}$ (cf. theorem \ref{Thm:Riemann_Small_Epsilon}) and since $\rho_{\ell}^{\eps}$, $\rho_{r}^{\eps} \rightarrow \rho^{\ast}$, we also have $\widetilde{\rho}^{\eps} \rightarrow \rho^{\ast}$. For the intermediate angle $\widetilde{\theta}^{\eps}$, the previous proof shows that it tends to $\theta_{r}$. Let us note that the 2-rarefaction wave tends to a contact wave (because $\lambda_{r+}^{\eps},\widetilde{\lambda}_{+}^{\eps} \rightarrow + \infty$). Let $s$ be the limit of the 1-shock speed and let us note that $s$ is lower than $\cos\theta_{\ell}$. Like in prop. \ref{Prop:limit_Riemann_2}, subcase (c), the intermediate pressure is equal to
\begin{equation*}
\bar{\bar{p}} = \bar{p}_{\ell} + \lim \frac{\left[\rho\cos\theta\right]_{\ell}\left[\Psi\right]_{\ell}}{\left[\rho\right]_{\ell}}.  
\end{equation*}
Now let us look at the limit of the shock speed. For finite $\eps$, the shock speed is given by:  
\begin{equation*}
s^{\eps} = \frac{\widetilde{\rho}^{\eps}\cos\widetilde{\theta}^{\eps} - \rho_{\ell}\cos\theta_{\ell}}{\widetilde{\rho}^{\eps} - \rho_{\ell}} = \cos\theta_{\ell} + \widetilde{\rho}^{\eps}\frac{\cos\widetilde{\theta}^{\eps} - \cos\theta_{\ell}}{\widetilde{\rho}^{\eps} - \rho_{\ell}}.
\end{equation*}
Since $\eps p(\rho_{r}^{\eps}) \rightarrow \bar{p}_{r} > 0$, we have $\eps^{\frac{1}{\gamma}} = O(\rho^{\ast} - \rho_{r}^{\eps})$ and then $\eps^{\frac{1}{\gamma}} = O(\rho^{\ast} - \widetilde{\rho}^{\eps})$. On the other hand, from lemma \ref{Prop:Integral_curve}  we have $(i^{\eps})^{-1}(\rho) - \theta_{r} = O(\eps^{\frac{1}{2\gamma}})$ and therefore we get 
\begin{equation*}
\cos\widetilde{\theta}^{\eps} - \cos\theta_{\ell} = - 2\sin\left(\frac{\widetilde{\theta}^{\eps} +\theta_{\ell}}{2}\right)\sin\left(\frac{\widetilde{\theta}^{\eps} - \theta_{\ell}}{2}\right) = O(\eps^{\frac{1}{2\gamma}}).
\end{equation*}
Thus, we easily get that $s^{\eps}$ tends to $\cos\theta_{\ell}$ and then that the pressure $\bar{\bar{p}}$ equals $\bar{p}_{\ell}$. 

(b) Here the proof is similar to the case where only one state converges to the congested state (see proof of prop. \ref{Prop:limit_Riemann_2}).

(c) Consider the case where the solution is the limit of two shock waves. By the monotony of the shock curves, the intermediate density is larger than the right and left ones (cf. theorem \ref{Thm:Riemann_Small_Epsilon}) and so it tends to the congested density too.
Suppose that the intermediate angle is not equal to $\theta_{\ell}$. As regards the 1-shock speed, we have 
\begin{equation*}
s^{\eps} = \frac{[\rho\cos\theta]_{\ell}}{[\rho]_{\ell}} = \cos\widetilde{\theta}^{\eps} + \rho_{\ell}\frac{(\cos\theta_{\ell}- \cos\widetilde{\theta}^{\eps})}{\widetilde{\rho}^{\eps} - \rho_{\ell}}. 
\end{equation*}
and so the limit 1-shock speed is $- \infty$. Besides, we have
\begin{equation*}
[\eps p(\rho)]_{\ell} = [\Psi(\cos\theta)]_{\ell}\frac{[\rho \cos\theta]_{\ell}}{[\rho]_{\ell}} - [\Phi\cos\theta]_{\ell},
\end{equation*}
which implies that $\eps p(\tilde{\rho}^{\eps})$ tends to $+ \infty$. Then we have 
\begin{equation*}
\frac{[\rho \cos\theta]_{r}}{[\rho]_{r}}[\Psi(\cos\theta)]_{r} =  [\eps p(\rho)]_{r} + [\Phi(\cos\theta)]_{r}.
\end{equation*}
Since the right hand side tends to $+\infty$, the 2-shock speed has to tend to $+\infty$ too. If the intermediate angle tends to $\theta_{\ell}$, then it does not tend to $\theta_{r}$ and the same arguments apply. The quantities
\begin{eqnarray*}
&&[\eps p(\rho)]_{\ell}[\rho]_{\ell} = [\Psi(\cos\theta)]_{\ell}[\rho \cos\theta]_{\ell} - [\Phi(\cos\theta)]_{\ell}[\rho]_{\ell},\\
&&[\eps p(\rho)]_{r}[\rho]_{r} = [\Psi(\cos\theta)]_{r}[\rho \cos\theta]_{r} - [\Phi(\cos\theta)]_{r}[\rho]_{r}
\end{eqnarray*}
are bounded. The limit of their quotient is 
\begin{equation*}
\frac{[\eps p(\rho)]_{r}[\rho]_{r}}{[\eps p(\rho)]_{\ell}[\rho]_{\ell}}  \underset{\eps \rightarrow 0}{\longrightarrow} \frac{[\Psi(\cos\theta)]_{r}[\cos\theta]_{r}}{[\Psi(\cos\theta)]_{\ell}[\cos\theta]_{\ell}}.
\end{equation*}
Besides, it is easily checked that  
\begin{equation*}
\frac{[\eps p(\rho)]_{r}[\rho]_{r}}{[\eps p(\rho)]_{\ell}[\rho]_{\ell}} \underset{\eps \rightarrow 0}{\sim} \frac{\tilde{\rho} - \rho_{r}}{\tilde{\rho} - \rho_{\ell}} \underset{\eps \rightarrow 0}{\sim}  \frac{\rho^{\ast} - \rho_{r}}{\rho^{\ast} - \rho_{\ell}},
\end{equation*}
where the last equivalence results from the fact that $(\rho^{\ast} - \tilde{\rho}) = o(\eps^{\frac{1}{\gamma}})$ since $\eps p(\tilde{\rho}) \rightarrow + \infty$ and $(\rho^{\ast} - \rho_{\ell,r}) = O(\eps^{\frac{1}{\gamma}})$, $\eps^{\frac{1}{\gamma}} = O(\rho^{\ast} - \rho_{\ell,r})$. Finally, we have 
\begin{equation*}
\frac{\rho^{\ast} - \rho_{r}}{\rho^{\ast} - \rho_{\ell}} = \left(\frac{\eps p(\rho_{\ell})}{\eps p(\rho_{r})}\right)^{\frac{1}{\gamma}} \rightarrow \left(\frac{\bar{p}_{\ell}}{\bar{p}_{r}}\right)^{\frac{1}{\gamma}}.
\end{equation*}

Consider now the limit of a solution consisting of one shock wave and one rarefaction wave. From lemma \ref{Lemma:Rarefaction_wave_2}, the intermediate angle $\widetilde{\theta}^{\eps}$ tends to $\theta_{r}$ and from the Rankine-Hugoniot relation, we have 
\begin{equation*}
\left[\eps p(\rho)\right]_{\ell}\left[\rho\right]_{\ell} = \left[\Psi(\cos\theta)\right]_{\ell}\left[\rho\cos\theta\right]_{\ell} - \left[\Phi(\cos\theta)\right]_{\ell}\left[\rho\right]_{\ell} .
\end{equation*}
So, since $\theta$ is  bounded and the densities $\rho_{\ell}^{\eps} < \widetilde{\rho}^{\eps}$ tend to $\rho^{\ast}$, the right hand side tends to a non zero value: $\rho^{\ast}\left[\Psi(\cos\theta)\right]_{\ell}\left[\cos\theta\right]_{\ell}$. Because $\left[\rho\right]_{\ell}$ tends to $0$, $\eps p(\widetilde{\rho}^{\eps})$ has to tend to $+ \infty$, which is absurd since $\eps p(\widetilde{\rho}^{\eps}) < \eps p(\rho_{r}^{\eps})$.\qed
\end{proof}

{\bf Acknowledgements:} 
The authors wish to thank Jacques Gautrais, Marie-H\'el\`ene Pillot and Guy Th\'eraulaz for stimulating discussions. This work has been supported by the "Agence Nationale de la Recherche", under contracts "PANURGE", ref. 07-BLAN-0208, and by the Marie Curie Actions of the European Commission in the frame of the DEASE project (MEST-CT-2005-021122).

\end{document}